%% file: hybrid.tex
\documentclass[twoside,runningheads,envcountsame]{llncs}

\usepackage{xcolor,colortbl}

\usepackage{amsthm}
\usepackage{amssymb}
\usepackage{amssymb}
\usepackage{booktabs}
\usepackage{tabularx}
\usepackage{verbatim}
\usepackage{framed}
\usepackage{tikz}
\usetikzlibrary{shapes,snakes}
\usepackage{pgfplots}
\usepackage{alltt}
\usetikzlibrary{arrows}
\usetikzlibrary{calc}
\usetikzlibrary{fit}
\usetikzlibrary{positioning,shadows}
\usetikzlibrary{automata}
\usetikzlibrary{shapes,arrows}
\usetikzlibrary{decorations.text}
\usetikzlibrary{decorations.markings}
\usetikzlibrary{trees,snakes}
\usepackage{graphicx}
\usepackage{paralist}
\usepackage{proof}
\usepackage{lscape}
\usepackage{textcomp}
\usepackage[bbsets,Dfprime]{math}
\usepackage[prefixflatinterpret,bracketmodalinterpret,modernsign,substopindex,shortmquant,mquantifiertype,mconnectiveformal,bracketinterpret,fixformat,setfixinterpret,modifopindex,seqarrow,seqoptional,sidenotecalculus,abbrseqcontext,shortterms,nosigmaterms,novarterms]{logic}
\usepackage[pretest,nocommandblocks]{progreg}
\usepackage[bracketinterpret,prefixflatinterpret,bracketmodalinterpret,fixformat,differentialdL]{dL}
\input{defs}

\let\citet\cite

\usepackage{wrapfig}
\usepackage{etoolbox}

\usepackage{booktabs}

\renewcommand{\proves}[3]{#1\allowbreak\vdash #3}
\renewcommand{\eAsgn}[4]{\lstrike\humod{#2}{\eren{f}{#2}{#1}}\text{ in }#4\rstrike}
\renewcommand{\eAsgneq}[4]{\eAsgn{#1}{#2}{#3}{#4}}
\renewcommand{\eghost}[4]{\textrm{Ghost}[#1=#2](#4)}
\renewcommand{\erep}[3]{{#1}\textrm{ }\kwrep\text{ }~{#2}}
\renewcommand{\emonInfix}[1]{{\circ}}

\newcommand{\pre}{\textsf{pre}}
\newcommand{\post}{\textsf{post}}
\newcommand{\reachavoid}{\textsf{reachAvoid}}
\newcommand{\live}{\textsf{liveness}}
\newcommand{\safe}{\textsf{safe}}
\newcommand{\safety}{\textsf{safety}}

\begin{document}
\title{Constructive Hybrid Games\thanks{%
This research was sponsored by the AFOSR under grant number FA9550-16-1-0288 and the Alexander von Humboldt Foundation. The first author was also funded by an NDSEG Fellowship.}
}

\author{
Brandon Bohrer\inst{1}\orcidID{0000-0001-5201-9895} \and
Andr\'{e} Platzer\inst{1,2}\orcidID{0000-0001-7238-5710}
}
\authorrunning{B.\ Bohrer and A.\ Platzer}
\institute{Computer Science Department, Carnegie Mellon University
\email{\{bbohrer,aplatzer\}@cs.cmu.edu}
\and Fakult\"at f\"ur Informatik, Technische Universit\"at M\"unchen
}

\maketitle
\begin{abstract}
  Hybrid games are models which combine discrete, continuous, and
  adversarial dynamics. Game logic enables proving (classical) existence of winning strategies. We introduce
  \emph{constructive differential game logic} (\CdGL) for hybrid games, where
  proofs that a player can win the game correspond to \emph{computable} winning strategies.
  This is the logical foundation for synthesis of
  correct control and monitoring code for safety-critical
  cyber-physical systems. Our contributions include novel static and
  dynamic semantics as well as soundness and consistency.
\end{abstract}


\section{Introduction}

Logics for program verification can be broadly divided into two groups.
Imperative programs are typically verified with \emph{program logics} such as Hoare calculi \cite{DBLP:journals/cacm/Hoare69} and dynamic logics (\DL) \cite{DBLP:conf/focs/Pratt76} where modalities capture the effect of program execution.
Games logics (\GL)~\cite{DBLP:conf/focs/Parikh83}, studied in this paper, are \DL's with a turn-taking program connective to switch players.
In contrast, functional programs are often studied by writing the program \emph{in} logic.
In a dependent type theory \cite{DBLP:journals/iandc/CoquandH88}, a program's type is its correctness specification.
Classical higher-order logics also specify functions and their correctness.

The Curry-Howard correspondence~\cite{curry1967combinatory,howard1980formulae} explains type-theoretic verification: a constructive proof of a
specification corresponds to a function which provably implements that
specification.  Curry-Howard for program logics is far less explored,
but deserves exploration in order to answer: what \emph{is} the
computational content of a program-logic proof? This paper argues that the intersection of
Curry-Howard and program logic, both fundamental tools, provides an insightful logical foundation. For example, we expect this correspondence will enable
synthesis for programmatic models that are too challenging
for automatic synthesis without a proof.

Our hybrid games are as in differential game logic (\dGL) \cite{DBLP:journals/tocl/Platzer15}, and
combine continuous dynamics, discrete computation, and adversarial
dynamics to provide powerful models of
cyber-physical systems (CPSs) such as transportation systems, energy systems, and medical devices.
But \dGL is classical, so truth of \dGL formulas only implies classical existence of winning strategies for their hybrid games.
Based on discrete Constructive Game Logic (\CGL)~\cite{esop20}, this paper introduces \emph{Constructive Differential Game Logic} (\CdGL) for hybrid games with a Curry-Howard interpretation that proofs that a player can win a hybrid game correspond to programs implementing their winning strategies.

\CdGL is a compelling use case for Curry-Howard-based
synthesis among program logics precisely because synthesizing a
winning strategy of a hybrid game is undecidable until a proof is
provided; the combination of adversarial and continuous dynamics makes
this so not just in theory but in practice. Curry-Howard for games
says proofs correspond to constructive winning strategies, allowing us to reduce
undecidable synthesis questions to verification questions which,
while still undecidable, are routinely verified with human assistance.

\paragraph{Contributions.}
We build directly on Constructive Game Logic (\CGL)~\cite{esop20} for discrete games and Differential Game Logic (\dGL)~\cite{DBLP:journals/tocl/Platzer15} for classical hybrid games.
In combining these logics, we must \emph{constructively} justify differential equation (ODE) reasoning.
This requires foundations in constructive analysis~\cite{bishop1967foundations,bridges2007techniques}, so the proofs in \rref{app:proofs} appeal to constructive formalizations of ODEs~\cite{DBLP:conf/mkm/Cruz-FilipeGW04,DBLP:conf/itp/MakarovS13}.
We also contribute a new type-theoretic semantics, as opposed to the previous realizability semantics~\cite{esop20}.
This clarifies subtle side conditions and should be useful in future constructive program logics.
Our example model and proof, while short, lay the groundwork for future case studies.

\section{Related Work}
We discuss related works on games, constructive logic,  and hybrid systems.




\paragraph{Games in Logic.}
Propositional \GL was introduced by Parikh~\citet{DBLP:conf/focs/Parikh83}.
\GL's are unique in their clear delegation of strategy to the \emph{proof} language rather than the \emph{model} language, allowing succinct game specifications with sophisticated winning strategies.
Succinct specifications are important: specifications are \emph{trusted} because proving the \emph{wrong theorem} would not ensure correctness.
Relatives without this separation include SL~\cite{DBLP:conf/concur/ChatterjeeHP07}, ATL~\cite{DBLP:journals/jacm/AlurHK02}, CATL~\cite{DBLP:conf/atal/HoekJW05}, SDGL~\cite{ghosh2008strategies}, structured strategies~\cite{DBLP:conf/kr/RamanujamS08},
DEL~\cite{DBLP:series/lncs/Benthem15,DBLP:journals/games/BenthemPR11,van2001games}, evidence logic~\cite{DBLP:journals/sLogica/BenthemP11}, and Angelic Hoare Logic~\cite{DBLP:journals/corr/Mamouras16}.

\paragraph{Constructive Modal Logics.}
We are interested in the semantics of games, thus we review constructive modal semantics generally.
This should not  be confused with game semantics~\cite{DBLP:journals/iandc/AbramskyJM00}, which give a semantics to programs \emph{in terms of} games.
The main semantic approaches for constructive modal logics are intuitionistic Kripke semantics~\citet{DBLP:journals/apal/Wijesekera90} and realizability semantics~\cite{DBLP:journals/mscs/Oosten02,lipton1992constructive}.
\CGL~\cite{esop20} used a realizability semantics which operate on a state, reminiscent of state in Kripke semantics, whereas we interpret \CdGL formulas into type theory.

Modal Curry-Howard is relatively little-studied, and each author has their own emphasis.
Explicit proof terms are considered for \CGL~\cite{esop20} and a small fragment thereof~\cite{kamide2010strong}.
Others~\cite{DBLP:journals/apal/WijesekeraN05,degen2006towards,DBLP:journals/fuin/Celani01} focus on intuitionistic semantics for their logics, fragments of \CGL.
Our semantics should be of interest for these fragments.
We omit proof terms for space. \CdGL proof terms would extend \CGL proof terms~\cite{esop20} with a constructive version of existing classical ODE proof terms~\cite{DBLP:journals/corr/abs-1908-05535}.
Propositional modal logic~\cite{DBLP:conf/lics/VIICHP04} has been interpreted as a type system.

\paragraph{Hybrid Systems Synthesis.}
Hybrid games synthesis is a motivation of this work.
Synthesis of hybrid \emph{systems} (1-player games) is an active area.
The unique strength of proof-based synthesis is expressiveness: they can synthesize every provable system.
\CdGL proofs support first-order regular games with first-order (e.g., semi-algebraic) initial and goal regions.
While synthesis and proof are both undecidable, interactive proof for undecidable logics is well-understood.
The \ModelPlex~\cite{DBLP:journals/fmsd/MitschP16} synthesizer for \CdGL's classical systems predecessor \dL~\cite{Platzer18} recently added~\cite{DBLP:conf/pldi/BohrerTMMP18} proof-based synthesis to improve expressiveness.
\CdGL aims to provide a computational foundation for a more systematic proof-based synthesizer in the more general context of games.

Fully-automatic synthesis, in contrast, restricts itself to some fragment in order to sidestep undecidability.
Studied classes include rectangular hybrid games~\cite{10.1007/3-540-48320-9_23},
switching systems~\cite{DBLP:conf/emsoft/TalyT10},
linear systems with polyhedral sets~\cite{DBLP:journals/tac/KloetzerB08,DBLP:conf/emsoft/TalyT10},
and  discrete abstractions~\cite{DBLP:conf/iros/FinucaneJK10,DBLP:conf/IEEEcca/FilippidisDLOM16}.
A well-known~\cite{tomlin2000game} \emph{systems} synthesis approach translates specifictions \emph{into} finite-alternation games.
Arbitrary first-order games are our \emph{source} rather than \emph{target} language.
Their approach is only known to terminate for simpler classes~\cite{shakernia2001semi,shakernia2000decidable}.


\section{Constructive Hybrid Games}
Hybrid games in \CdGL are 2-player, zero-sum, and perfect-information, where continuous subgames are ordinary differential equations (ODEs) whose duration is chosen by a player.
Hybrid games should not be confused with \emph{differential games} which compete continuously~\cite{IsaacsDifferentialGames}.
The player currently controlling choices is always called Angel in this paper, while the player waiting to play is always called Demon.
For any game $\alpha$ and formula $\phi,$ the modal formula $\ddiamond{\alpha}{\phi}$ says Angel can play $\alpha$ to ensure postcondition $\phi,$ while $\dbox{\alpha}{\phi}$ says Demon can play $\alpha$ to ensure postcondition $\phi$.
These generalize safety and liveness modalities from \DL.
\GL's are distinguished from other \DL's by the dual game $\pdual{\alpha},$ which implements turn-taking by switching the Angel and Demon roles in game $\alpha$.
The Curry-Howard interpretation of proof of a modality $\ddiamond{\alpha}{\phi}$ or $\dbox{\alpha}{\phi}$ is a program which performs each player's winning strategy.
A game might have several winning strategies, each represented by a different proof.

\subsection{Syntax of \CdGL}
We introduce the language of \CdGL. We introduce three classes of expressions $e$: terms $f,g,$ games $\alpha,\beta,$  and formulas $\phi, \psi.$

The valuation of a term is a real number understood constructively a l\`a Bishop~\cite{bishop1967foundations,bridges2007techniques}:
effectivity requires that all \emph{functions on real numbers} are computable, yet effectivity \emph{does not} require that variables range over \emph{only} computable reals.
Bishop-style real analysis preserves many classical intuitions (e.g., uncountability) about $\reals$ while ensuring effectivity.

The simplest terms are \emph{game variables} $x, y \in \allvars$ where $\allvars$ is the set of variable identifiers.
The game variables, which are mutable, contain the state of the game, which is globally scoped.
For every base game variable $x$ there is a primed counterpart $\D{x}$ whose purpose is to track the time derivative of $x$ within an ODE.
Real-valued terms $f,g$ are simply (Type-2~\cite{DBLP:series/txtcs/Weihrauch00}) effective functions, usually from states to reals.
Type-2 effectivity means $f$ must be computable when values are represented as streams of bits.
It is occasionally useful for $f$ to return a tuple of reals, which are computable when every component is computable.
\begin{definition}[Terms]
A \emph{term} $f, g$ is any computable function over the game state.
The following constructs appear in our example:
\[f,g ~\bebecomes~  \cdots \alternative c \alternative x \alternative f + g \alternative f \cdot g \alternative \ediv{f}{g} \alternative \min(f,g) \alternative \max(f,g) \alternative \der{f}\]
where $c \in \mathbb{R}$ is a real literal, $x$ a game variable, $f + g$ a sum, $f \cdot g$ a product, and $\ediv{f}{g}$ is the quotient in real division of $f$ by $g$. Divisors $g$ are assumed to be nonzero.
Minimum and maximum of terms $f$ and $g$ are written $\min(f,g)$ and $\max(f,g)$.
Any differentiable term $f$ has a definable~(\rref{sec:cdgl-semantics}) spatial differential term $\der{f},$
which agrees with the time derivative within an ODE.
\label{def:terms}
 \end{definition}
Because \CdGL is constructive, Angel strategies must make their choices computably.
Until his turn, Demon just observes Angel's choices, and does not care whether Angel made them computably.
We informally discuss how a game is played here, then give full winning conditions in \rref{sec:semantics}.
In \textcolor{red}{red} are the ODE and dual games, which respectively distinguish hybrid games from discrete games and games from systems.
\begin{definition}[Games]
The set of \emph{games} $\alpha,\beta$ is defined recursively as such:
\[\alpha,\beta ~\bebecomes~ \ptest{\phi} \alternative \humod{x}{f} \alternative \prandom{x} \alternative {\textcolor{red}{\pevolvein{\D{x}=f}{\ivr}}} \alternative  \pchoice{\alpha}{\beta} \alternative \alpha;\beta \alternative \prepeat{\alpha} \alternative {\textcolor{red}{\pdual{\alpha}}}\]
\end{definition}
The \emph{test game} $\ptest{\phi},$ is a no-op if Angel proves $\phi,$ else Demon wins by default since Angel ``broke the rules''.
A deterministic assignment $\humod{x}{f}$ updates game variable $x$ to the value of term $f$.
Nondeterministic assignments $\prandom{x}$ ask Angel to compute the new value of $x : \reals$.
The ODE game $\pevolvein{\D{x}=f}{\ivr}$ evolves ODE $\D{x}=f$ for duration $d \geq 0$ chosen by Angel such that Angel proves the domain constraint formula $\ivr$ is true throughout.
We require that term $f$ is locally Lipschitz-continuous on domain $\ivr$, so solutions are unique.
ODEs are explicit-form, so no primed variable $\D{y}$ for $y \in \allvars$ is mentioned in $f$ or $\ivr$.
Systems of ODEs are supported, we present single equations for readability.
In the choice game $\alpha \cup \beta,$ Angel chooses whether to play game $\alpha$ or game $\beta$.
In the sequential composition game $\alpha;\beta$, game $\alpha$ is played first, then $\beta$ from the resulting state.
In the repetition game $\prepeat{\alpha},$ Angel chooses after each repetition of $\alpha$ whether to continue playing, but must not repeat $\alpha$ infinitely.
The exact number of repetitions is not known in advance, because it may depend on Demon's reactions.
In the dual game $\pdual{\alpha},$ Angel takes the Demon role and vice-versa while playing $\alpha$.
Demon strategies ``wait'' until a dual game $\pdual{\alpha}$ is encountered, then contain an Angelic strategy for $\alpha$.
We parenthesize games with braces $\{ \alpha \}$ when necessary.

\begin{definition}[\CdGL Formulas]
The \CdGL \emph{formulas} $\phi$ (also $\psi, \rho$) are:
\[ \phi ~\bebecomes~ \ddiamond{\alpha}{\phi} \alternative \dbox{\alpha}{\phi} \alternative f \sim g\]
\label{def:cgl-formula}
\end{definition}
Above, $f \sim g$ is a comparison formula for $\sim\mathop{\in}\{\leq, <, =, \neq, >, \geq\}$.
The defining formulas of \CdGL (and \GL) are the modalities $\ddiamond{\alpha}{\phi}$ and $\dbox{\alpha}{\phi}$.
These mean that Angel or Demon respectively have a \emph{constructive} strategy to play $\alpha$ and prove postcondition $\phi$.
We do not develop modalities for existence of classical strategies because those cannot be synthesized to executable code.

Standard connectives are defined from games and comparisons.
Verum ($\btt$) is defined $1 > 0$ and falsum ($\bff$) is $0 > 1$.
Conjunction $\phi \land \psi$ is defined $\ddiamond{\ptest{\phi}}{\psi},$
disjunction $\phi \lor  \psi$ is defined $\ddiamond{\ptest{\phi} \cup \ptest{\psi}}{\btt},$
and implication $\phi \limply \psi$ is defined $\dbox{\ptest{\phi}}{\psi}$.
Real quantifiers $\lforall{x}{\phi}$ and $\lexists{x}{\phi}$ are defined $\dbox{\prandom{x}}{\phi}$ and
 $\ddiamond{\prandom{x}}{\phi},$ respectively.
As usual, equivalence $\phi \lequiv \psi$ reduces to $(\phi \limply \psi) \land (\psi \limply \phi),$
negation $\neg \phi$ is defined $\phi \limply \bff$, and inequality is defined by $f \neq g \equiv \neg(f = g)$.
Semantics and proof rules are needed only for core constructs, but we use derived constructs when they improve readability.
Keep these definitions in mind, because the semantics and rules for some game connectives mirror first-order connectives.

For convenience, we also write derived operators where Demon is given control of a single choice before returning control to Angel.
The \emph{Demonic choice} $\dchoice{\alpha}{\beta},$ defined $\pdual{\{{\pchoice{\pdual{\alpha}}{\pdual{\beta}}}\}},$ says Demon chooses which branch to take, but Angel controls the subgames.
\emph{Demonic repetition} $\drepeat{\alpha}$ is defined likewise by $\{\{\alpha\pdual{{}}\} \prepeat{{}}\} \pdual{{}}$.

We write $\eren{\phi}{x}{y}$ (likewise for $\alpha$ and $f$) for the \emph{renaming} of $x$ for $y$ and vice versa in formula $\phi$,
and write $\tsub{\phi}{x}{f}$ for the result of \emph{substitution} of term $f$ for game variable $x$ in $\phi$,
if the substitution is admissible (\rref{def:lem-admit}).

\subsection{Example Game}
\label{sec:example-game}
We give an example game and theorem statements, \textbf{proven in \rref{app:example-proofs}}.
Automotive systems are a major class of CPS, so we consider simple time-triggered 1-dimensional driving with adversarial timing.
For maximum time $T$ between control cycles, we let Demon choose any duration in $[T/2,T]$.
This forces Angel's controller to be robust to realistic timing constraints, yet prohibits Demon from pathological ``Zeno'' behaviors.

We write $x$ for the position of the car, $v$ for the velocity, $a$ for the current acceleration, $A > 0$ for the maximum positive acceleration, and $B > 0$ for the maximum braking rate.
We assume $x=v=0$ initially to simplify arithmetic.
In time-triggered control, the controller runs at least once every $T > 0$ time units.
Time and physics are continuous, $T$ simply says how often the controller runs.
Local clock $t$ marks the current time within the current timestep, then resets at each step.
The control game ($\ctrl$) says Angel can pick any acceleration $a$ that is physically achievable ($-B \leq a \leq A$).
The clock $t$ is then reinitialized to $0$.
The plant game ($\plant$) says Demon can evolve physics for duration $t \in [T/2,T]$ such that $v \geq 0$ throughout, then returns control to Angel.
The lower bound on $t$ rules out Zeno strategies where Demon ``cheats'' by exponentially decreasing durations to effectively stop time.
The limit $t \geq T/2$ is chosen for simplicity.

Typical theorems in \DL's and \GL's are \emph{safety} and \emph{liveness}: are unsafe states always avoided and are goals eventually reached?
Safety and liveness of the 1D \emph{system} has been proven previously: safe driving ($\safety$) never goes past goal $g,$ while live driving eventually reaches $g$ ($\live$).
\begin{align*}
  \pre &\equiv T > 0 \land A > 0 \land B > 0 \land v=0 \land x=0 \qquad  \post\ \equiv (\dvar=x \land v=0)\\
  \ctrl &\equiv \prandom{a}; \ptest{-B \leq a \leq A}; \humod{t}{0} \\
  \plant &\equiv \pdual{\{\pevolvein{\D{t}=1, \D{x}=v,\D{v}=a}{t \leq T \land v \geq 0}\}} \\
  \safety &\equiv \pre \limply \ddiamond{\drepeat{(\ctrl;\plant)}}{x \leq \dvar} \qquad \live\ \equiv \pre \limply \ddiamond{\prepeat{(\ctrl;\plant)}}{(x \geq \dvar)}  
\end{align*}
Safety and liveness theorems, if designed carelessly, have trivial solutions.
It is safe to remain at $x=0$ and is live to maintain $a = A,$ but not vice-versa.
In contrast to \DL's, \GL's easily express the requirement that \emph{the same} strategy is both safe and live: we must remain safe \emph{while} reaching the goal.
This specification is called \emph{reach-avoid}, which we use because it is immune to trivial solutions.
We state and prove a new reach-avoid result for 1D driving.
\begin{example}[Reach-avoid]
The following is provable in \dGL and \CdGL:
\[\reachavoid \equiv \pre \limply \ddiamond{\prepeat{\{\ctrl;\plant;\ptest{x \leq \dvar};\pdual{\{\ptest{t > T/2}\}}\}}}{\post}\]
\label{ex:reach-avoid}
\end{example}
Angel \emph{reaches} $v=0 \land \dvar=x$ while safely \emph{avoiding} states where $x \leq \dvar$ does not hold.
Angel is safe at \emph{every} iteration for \emph{every} time $t \in [0,T]$, thus safe \emph{throughout} the game.
The test $t \in [T/2,T]$ appears second, allowing Demon to win if Angel violates safety during $t < T/2$.

\begin{wrapfigure}{r}{2.5in}
  \vspace*{-2.5\baselineskip}
  \includegraphics[width=2.5in]{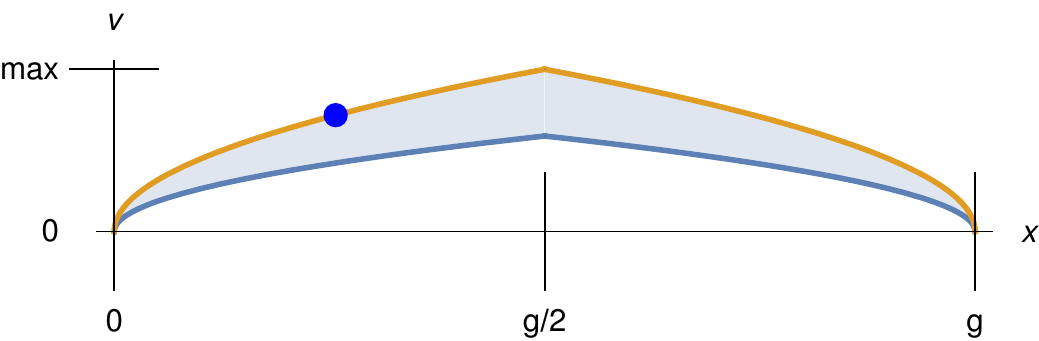}
\caption{Safe driving envelope}
  \vspace*{-1.5\baselineskip}
  \label{fig:simple-envelope}
\end{wrapfigure}
1D driving is well-studied for classical systems, but the constructive reach-avoid proof (\rref{app:example-proofs}) is subtle.
The proof constructs a envelope of safe upper and live lower bounds on velocity as a function of position \rref{fig:simple-envelope}.
The blue point indicates where Angel must begin to brake to ensure time-triggered safety.
It is surprising that Angel can achieve postcondition $\dvar=x \land v=0$, given that trichotomy ($f < g \lor f = g \lor f > g$) is constructively invalid.
The key (\rref{app:example-proofs}) is  comparison \emph{terms} $\min(f,g)$ and $\max(f,g)$ \emph{are} exact under Type 2 effectivity where bits of $\min$ and $\max$ may be computed lazily.
Our exact result encourages us that constructivity is not overly burdensome in practice.
When decidable comparisons ($f < g + \delta \lor f > g$) \emph{are} needed, the alternative is a weaker guarantee $x \in [\dvar-\veps,\dvar]$ for parameter $\veps > 0$.
This relaxation is often enough to make the theorem provable, and reflects the fact that real agents only expect to reach their goal within finite precision.

\section{Type-theoretic Semantics}
\label{sec:semantics}
In this section, we define the semantics of games and game formulas in type theory.
We start with assumptions on the underlying type theory.

\subsection{Type Theory Assumptions}
We assume a  Calculus of Inductive and Coinductive Constructions (CIC)-like type theory~\cite{DBLP:journals/iandc/CoquandH88,DBLP:conf/colog/CoquandP88,COQ} with polymorphism and dependency.
We assume first-class anonymous constructors for (indexed~\cite{DBLP:journals/fac/Dybjer94}) inductive and coinductive types.
We write $\tau$ for type families and $\kappa$ for kinds (those type families inhabited by other type families).
Inductive type families are written $\lindty{t:\kappa}{\tau},$ which denotes the \emph{smallest} solution \texttt{ty} of kind $\kappa$ to the fixed-point equation $\texttt{ty} = \tsub{\tau}{t}{\texttt{ty}}.$
Coinductive type families are written $\lcoty{t:\kappa}{\tau},$ which denotes the \emph{largest} solution \texttt{ty} of kind $\kappa$ to the fixed-point equation $\texttt{ty} = \tsub{\tau}{t}{\texttt{ty}}.$
Per Knaster-Tarski~\cite[Thm.\ 1.12]{DBLP:harel2000}, the type-expression $\tau$ must be monotone in $t$ to ensure that smallest and largest solutions exist.
We allow arbitrary proofs that $\tau$ is monotone; a major reason we did not mechanize this work is that prominent proof assistants such as Coq reject definitions where monotonicity requires nontrivial proof.

We use a single predicative universe which we write $\alltype$ and Coq writes \texttt{Type\ 0}.
Predicativity is an important assumption because our semantic definition is a large elimination, a feature known to interact dangerously with impredicativity.
We write $\pity{x}{\tau_1}{\tau_2}$ for a dependent function type with argument named $x$ of type $\tau_1$ and where return type $\tau_2$ may mention $x$.
We write $\sity{x}{\tau_1}{\tau_2}$ for a dependent pair type with left component named $x$ of type $\tau_1$ and  right component of type $\tau_2,$ possibly mentioning $x$.
These specialize to the simple types $\sfun{\tau_1}{\tau_2}$ and $\sprod{\tau_1}{\tau_2}$ respectively when $x$ is not mentioned in $\tau_2$.
Lambdas $(\lambda x:\tau.\, M)$ inhabit function types.
Pairs  $(M,N)$ inhabit dependent pair types.
Application is $M\ N$.
Let-binding unpacks pairs and $\sprojL{M}$ and $\sprojR{M}$ are left and right projection.
We write $\tau_1 + \tau_2$ for disjoint unions inhabited by $\ell \cdot M$ and $r \cdot M,$ and write $\scase{A}\slbranch{\pvx}{B}~\srbranch{\pvy}{C}$ for case analysis.

We assume a real number type $\xty$ and a Euclidean state type \sty.
The positive real numbers are written $\reals^+$, nonnegative reals $\reals_{\geq0}$.
We assume scalar and vector sums, products, inverses, and units.
A state $s:\sty$ assigns values to every variable $x \in \allvars$ and supports operations $\lget{s}{x}$ and $\lset{s}{x}{v}$ which respectively retrieve the value of $x$ or update it to $v$.
The usual axioms of setters and getters~\cite{foster2010bidirectional} are satisfied.

\subsection{Semantics of \CdGL}
\label{sec:semantics}
Terms $f,g$ are interpreted into type theory as functions of type $\sfun{\sty}{\xty}$.
We will need differential terms $\der{f},$ a definable term construct when $f$ is differentiable.
Not every term $f$ need be differentiable, so we give a \emph{virtual} definition, defining when $\der{f}$ is equal to some term $g$.
If $\der{f}$ does not exist, $\der{f} = g$ is not provable.
We define the (total) differential as the dot product ($\vdot$) of the gradient (variable name: $\nabla$) with $\D{s},$ which is the vector of values $s\ \D{x}$ assigned to primed variables.
To show that $\nabla$ is the gradient, we define the gradient as a limit, which we express in $(\varepsilon,\delta)$ style.
In this definition, $f$ and $g$ are scalar-valued, and the minus symbol is used for both scalar and vector difference.
\begin{align*}
  (\der{f}\ s = g\ s) &\equiv
    \lexists[{\xty^{\abs{\D{s}}}}]{\nabla}{}
    (g\ s = \nabla \vdot \D{s}) \kwprod
    \pity{\varepsilon}{\reals^+}{\sity{\delta}{\reals^+}{\pity{r}{\sty}{}}}\\
    &\sfun{(\norm{r - s} < \delta )}
        {\abs{f\ r - f\ s - \nabla \vdot (r - s)} \leq \varepsilon \norm{r - s}}
\end{align*}
For practical proofs, a library of standard rules for automatic, syntactic differentiation of common arithmetic operations can be proven.

We model a formula $\phi$ as a predicate over states, i.e., a type family $\ftrans{\phi} : \sty \to \alltype$.
A predicate of kind $\sty \to \alltype$ is also understood as a \emph{region}, e.g., $\ftrans{\phi}$ is the region containing states where $\phi$ is provable.
We say the formula $\phi$ is valid if there exists a term $M$ such that $\cdot \vdash M : (\pity{s}{\sty}{\ftrans{\phi}\ s})$.
That is, a valid formula is provable in every state.
The witness may inspect the state, but must do so constructively.
The formula semantics are defined in terms of the Angelic and Demonic semantics of games, which determine how to win a game $\alpha$ whose postcondition is $\phi$.
We write $\atrans{\alpha} : (\sty \to \alltype) \to (\sty \to \alltype)$ for the Angelic semantics of $\alpha$ and $\dtrans{\alpha} : (\sty \to \alltype) \to (\sty \to \alltype)$ for its Demonic semantics.
Angel and Demon strategies for a game $\alpha$ with postcondition $P$ are inhabitants of $\atrans{\alpha}\ P$ and $\dtrans{\alpha}\ P,$ respectively. 

\begin{definition}[Formula semantics]
\begin{align*}
\ftrans{\dbox{\alpha}{\phi}}\ s     &= \dtrans{\alpha}\ \ftrans{\phi}\ s &
\ftrans{\ddiamond{\alpha}{\phi}}\ s &= \atrans{\alpha}\ \ftrans{\phi}\ s &
\ftrans{f \sim g}\ s  &= ((f\ s) \sim (g\ s))
\end{align*}
\end{definition}
Modality $\ddiamond{\alpha}{\phi}$ is provable in $s$ when $\atrans{\alpha}\ \ftrans{\phi}\ s$ is inhabited so Angel has an $\alpha$ strategy from $s$ to reach region $\ftrans{\phi}$ on which $\phi$ is provable.
Modality $\dbox{\alpha}{\phi}$ is provable in $s$ when  $\dtrans{\alpha}\ \ftrans{\phi}\ s$ is inhabited so
 Demon has an $\alpha$ strategy from  $s$ to reach region $\ftrans{\phi}$ on which $\phi$ is provable.
For $\sim~\in~\{\leq,<,=,>,\geq,\neq\},$ the values of $f$ and $g$ are compared at state $s$ in $f \sim g$.
The game and formula semantics are simultaneously inductive.
In each case, the connectives which define $\atrans{\alpha}$ and $\dtrans{\alpha}$ are duals, because $\dbox{\alpha}{\phi}$ and $\ddiamond{\alpha}{\phi}$ are dual.
Below, $P$ refers to the postcondition of the game and $s$ to the initial state.

\begin{definition}[Angel semantics]~\\
We define $\atrans{\alpha} : (\allstate \to \alltype) \to (\allstate \to \alltype)$ inductively (by a large elimination) on $\alpha$:
{\small{\begin{minipage}{0.3\textwidth}
  \begin{align*}
\atrans{\ptest{\psi}}\ P\ s     &= \ftrans{\psi}\ s \mathop{\kwprod} P\ s\\
\atrans{\humod{x}{f}}\ P\ s     &= P\ (\lset{s}{x}{(f\ s)})\\
\atrans{\prandom{x}}\ P\ s      &= \sity{v}{\tau}{\,P\ (\lset{s}{x}{v})}\\
\atrans{\alpha\cup\beta}\ P\ s  &= \atrans{\alpha}\ P\ s\, \mathop{\kwsum}\, \atrans{\beta}\ P\ s\\
\atrans{\alpha;\beta}\ P\ s     &= \atrans{\alpha}\ (\atrans{\beta}\ P)\ s
\end{align*}
\end{minipage}}
{\begin{minipage}{0.5\textwidth}
  \begin{align*}
\atrans{\pdual{\alpha}}\ P\ s   &= \dtrans{\alpha}\ P\ s\\
\atrans{\pevolvein{\D{x}=f}{\ivr}}&\ P\ s =  \sity{d}{{\reals_{\geq0}}}{\sity{sol}{[0,d]\to\xty}{}} \\
\hskip-1in &(\solves{sol}{s}{d}{\D{x}=f})\\
\hspace{-1in} \mathop{\kwprod}\, &(\pity{t}{{[0,d]}}{\ftrans{\ivr}\ {(\lset{s}{x}{(sol\ t)})}})\\
\hspace{-1in} \mathop{\kwprod}\, &P\ (\lset{s}{(x,\D{x})}{}\\
              &\ \ \ \ \ (sol\ d, f\ (\lset{s}{x}{(sol\ d)})))
  \end{align*}
\end{minipage}}\\
\[\atrans{\prepeat{\alpha}}\ P\ s = (\lindty{\tau'\mathrel{:}(\sty \to \alltype)}{\lambda{t:\sty}.\,{
(P\ t \to \tau'\ t)
\,\mathop{\kwsum}\,
(\atrans{\alpha}\ \tau'\ t \to \tau'\ t)
}})\ s\]}
\end{definition}
Angel wins $\ddiamond{\ptest{\psi}}{P}$ by proving both $\psi$ and $P$ at $s$.
Angel wins the deterministic assignment $\humod{x}{f}$ by performing the assignment, then proving $P$.
Angel wins nondeterministic assignment $\prandom{x}$ by constructively choosing a value $v$ to assign, then proving $P$.
Angel wins $\alpha \cup \beta$ by choosing between playing $\alpha$ or $\beta,$ then winning that game.
Angel wins $\alpha;\beta$ if he wins $\alpha$ with the postcondition of winning $\beta$.
Angel wins $\pdual{\alpha}$ if he wins $\alpha$ in the Demon role.
Angel wins ODE game $\pevolvein{\D{x}=f}{\ivr}$ by choosing some solution $sol$ of some duration $\durvar$ which satisifies the ODE and domain constraint throughout and the postcondition $\phi$ at time $\durvar$.
While top-level postconditions rarely mention $\D{x},$ intermediate invariant steps do,
thus $x$ and $\D{x}$ both appear updated in the postcondition.
The construct $\solves{sol}{s}{\durvar}{\D{x}=f},$ saying $sol$ solves $\D{x}=f$ from state $s$ for time $\durvar$, is defined:
\[\small{(\solves{sol}{s}{\durvar}{\D{x}=f}) \equiv
\left(
\sprod{(\lget{s}{x} = sol\ 0)}
{\pity{r}{[0,\durvar]}{
   (\der{sol}\ r = f\ (\lset{s}{x}{(sol\ r)}))
}}\right)}\]
Angel strategies for $\prepeat{\alpha}$ are inductively defined: either choose to stop the loop and prove $P$ now, else play a round of $\alpha$ before repeating inductively.
By Knaster-Tarski~\cite[Thm.\ 1.12]{DBLP:harel2000}, this least fixed point exists because the interpretation of a game is monotone in its postcondition (\rref{lem:monotone}).

\begin{lemma}[Monotonicity]
  Let $P, Q : \sty \to \alltype$.
  If $\pity{s}{\sty}(P\ s \to Q\ s)$ is inhabited, then so are
      $(\pity{s}{\sty}{\atrans{\alpha}\ P\ s \to \atrans{\alpha}\ Q\ s})$
  and $(\pity{s}{\sty}{\dtrans{\alpha}\ P\ s \to \dtrans{\alpha}\ Q\ s})$
\label{lem:monotone}
\end{lemma}

\begin{definition}[Demon semantics]~\\
We define $\dtrans{\alpha} : (\allstate \to \alltype) \to (\allstate \to \alltype)$ inductively (by a large elimination) on $\alpha$:
{\small{\begin{minipage}{0.4\textwidth}
\begin{align*}
\dtrans{\ptest{\psi}}\ P\ s     &= \ftrans{\psi}\ s \limply P\ s\\
\dtrans{\humod{x}{f}}\ P\ s     &= P\ {(\lset{s}{x}{(f\ s)})}\\
\dtrans{\prandom{x}}\ P\ s      &= \pity{v}{\tau}{P\ {(\lset{s}{x}{v})}}\\
\dtrans{\alpha\cup\beta}\ P\ s  &= \dtrans{\alpha}\ P\ s\, \mathop{\kwprod}\ \dtrans{\beta}\ P\ s\\
\dtrans{\alpha;\beta}\ P\ s     &= \dtrans{\alpha}\ (\dtrans{\beta}\ P)\ s
\end{align*}
\end{minipage}}
 {\begin{minipage}{0.5\textwidth}
\begin{align*}
\dtrans{\pdual{\alpha}}\ P\ s   &= \atrans{\alpha}\ P\ s\\
\dtrans{\pevolvein{\D{x}=f}{\ivr}}&\ P\ s =
  \pity{d}{\reals_{\geq0}}{\pity{sol}{[0,d]\to\xty}{}} \\
     &(\solves{sol}{s}{d}{\D{x}=f})\\
 \to &(\pity{t}{[0,d]}{\ftrans{\ivr}\ {(\lset{s}{x}{(sol\ t)})}})\\
 \to &P\ (\lset{s}{(x,\D{x})}{}\\
     &\ \ \ \ \  (sol\ d, f\ (\lset{s}{x}{(sol\ d)})))
\end{align*}
 \end{minipage}}\\
\[\dtrans{\prepeat{\alpha}}\ P\ s = (\lcoty{\tau'\mathrel{:}(\sty \to \alltype)}{\lambda{t:\sty}.\,{
(\tau'\ t \to \dtrans{\alpha}\ \tau'\ t)
\,\mathop{\kwprod}\,
(\tau'\ t \to P\ t)
}})\ s\]
}
\end{definition}
Demon wins $\dbox{\ptest{\psi}}{P}$ by proving $P$ under assumption $\psi$, which Angel must provide (\rref{sec:synthesis}).
Demon's deterministic assignment is identical to Angel's.
Demon wins $\prandom{x}$ by proving $\psi$ for \emph{every} choice of $x$.
Demon wins $\alpha \cup \beta$ with a pair of winning strategies.
Demon wins $\alpha;\beta$ by winning $\alpha$ with a postcondition of winning $\beta$.
Demon wins $\pdual{\alpha}$ if he can win $\alpha$ after switching roles with Angel.
Demon wins $\pevolvein{\D{x}=f}{\ivr}$ if for an arbitrary duration and arbitrary solution which satisfy the domain constraint, he can prove the postcondition.
Demon wins $\dbox{\prepeat{\alpha}}{P}$ if he can prove $P$ no matter how many times Angel makes him play $\alpha$.
Demon repetition strategies are coinductive using some invariant $\tau'$.
When Angel decides to stop the loop, Demon responds by proving $P$ from $\tau'$.
Whenever Angel chooses to continue, Demon proves that $\tau'$ is preserved.
Greatest fixed points exist by Knaster-Tarski~\cite[Thm.\ 1.12]{DBLP:harel2000} and \rref{lem:monotone}.

It is worth comparing the Angelic and Demonic semantics of $\prandom{x}$.
An Angel strategy says how to compute $x$.
A Demon strategy simply accepts $x \in \reals$ as its input, even uncomputable numbers.
This is because Angel strategies supply a computable real while Demon act with computable outputs given real \emph{inputs}.
In general, each strategy is constructive but permits its opponent to play classically.
In the cyber-physical setting, the opponent is indeed rarely a computer.

\section{Proof Calculus}
\label{sec:proof-calculus}
To enable direct syntactic proof, we give a natural deduction system for \CdGL.
We write $\G = \psi_1, \ldots, \psi_n$ for a context of formulas and $\proves{\G}{}{\phi}$ for the natural-deduction sequent with conclusion $\phi$ and context $\G$.
We begin with rules shared by \CGL~\cite{esop20} and \CdGL, then present the new rules for ODEs.
We write $\eren{\G}{x}{y}$ for the renaming of game variable $x$ to $y$ and vice versa in context $\G$.
Likewise $\tsub{\G}{x}{f}$ is the substitution of term $f$ for game variable $x$.
To avoid repetition, we write $\dmodality{\alpha}{\phi}$ to indicate that the same rule applies for $\ddiamond{\alpha}{\phi}$ and $\dbox{\alpha}{\phi}$.
These rules write $\pmodality{\alpha}{\phi}$ for the dual of $\dmodality{\alpha}{\phi}$.
We write $\freevars{e}$ and $\boundvars{\alpha}$ for the free variables of expression $e$ and bound variables of game $\alpha$ respectively, i.e.,  variables which might influence the meaning of an expression or be modified during game execution.

\begin{figure}
\centering
{\small\begin{calculuscollections}{0.5\columnwidth}
\begin{calculus}
\cinferenceRule[bchoiceI|{$[\cup]${I}}]{}
{\linferenceRule[formula]
  {\proves{\G}{M}{\dbox{\alpha}{\phi}} & \proves{\G}{N}{\dbox{\beta}{\phi}}}
  {\proves{\G}{\ebcons{M}{N}}{\dbox{\alpha\cup\beta}{\phi}}}
}{}
\cinferenceRule[dchoiceIL|{$\langle\cup\rangle${I1}}]{}
{\linferenceRule[formula]
  {\proves{\G}{M}{\ddiamond{\alpha}{\phi}}}
  {\proves{\G}{\edinjL{M}}{\ddiamond{\alpha\cup\beta}{\phi}}}
}{}
\cinferenceRule[dtestI|{$\langle?\rangle$}{I}]{}
{\linferenceRule[formula]
  {\proves{\G}{M}{\phi} & \proves{\G}{N}{\psi}}
  {\proves{\G}{\edcons{M}{N}}{\ddiamond{\ptest{\phi}}{\psi}}}
}{}
\cinferenceRule[btestI|{$[?]$}{I}]{}
{\linferenceRule[formula]
  {\proves{\G,\phi}{M}{\psi}}
  {\proves{\G}{(\eplam{\phi}{M})}{\dbox{\ptest{\phi}}{\psi}}}
}{}
\end{calculus}
\end{calculuscollections}\hfill
\begin{calculuscollections}{0.33\columnwidth}
\begin{calculus}
\cinferenceRule[bchoiceEL|{$[\cup]${E1}}]{}
{\linferenceRule[formula]
  {\proves{\G}{M}{\dbox{\alpha\cup\beta}{\phi}}}
  {\proves{\G}{\ebprojL{M}}{\dbox{\alpha}{\phi}}}
}{}
\cinferenceRule[dchoiceIR|{$\langle\cup\rangle${I2}}]{}
{\linferenceRule[formula]
  {\proves{\G}{M}{\ddiamond{\beta}{\phi}}}
  {\proves{\G}{\edinjR{M}}{\ddiamond{\alpha\cup\beta}{\phi}}}
}{}
\cinferenceRule[dtestEL|{$\langle?\rangle$}{E1}]{}
{\linferenceRule[formula]
  {\proves{\G}{M}{\ddiamond{\ptest{\phi}}{\psi}}}
  {\proves{\G}{\edprojL{M}}{\phi}}
}{}
\cinferenceRule[btestE|{$[?]$}{E}]{}
{\linferenceRule[formula]
  {\proves{\G}{M}{\dbox{\ptest{\phi}}{\psi}} & \proves{\G}{N}{\phi}}
  {\proves{\G}{(\eapp{M}{N})}{\psi}}
}{}
\end{calculus}
\end{calculuscollections}\hfill
\begin{calculuscollections}{0.33\columnwidth}
  \begin{calculus}
\cinferenceRule[bchoiceER|{$[\cup]${E2}}]{}
{\linferenceRule[formula]
  {\proves{\G}{M}{\dbox{\alpha\cup\beta}{\phi}}}
  {\proves{\G}{\ebprojR{M}}{\dbox{\beta}{\phi}}}
}{}
\cinferenceRule[hyp|{hyp}]{}
{\linferenceRule[formula]
  {}
  {\proves{\G,\phi}{\pvx}{\phi}}
}{}
\cinferenceRule[dtestER|{$\langle?\rangle$}{E2}]{}
{\linferenceRule[formula]
  {\proves{\G}{M}{\ddiamond{\ptest{\phi}}{\psi}}}
  {\proves{\G}{\edprojR{M}}{\psi}}
}{}
  \end{calculus}
\end{calculuscollections}

\begin{calculuscollections}{\columnwidth}
\begin{calculus}
\cinferenceRule[dchoiceE|{$\langle\cup\rangle${E}}]{}
{
\linferenceRule[formula]
{\proves{\G}{A}{\ddiamond{\alpha\cup\beta}{\phi}}
        &\proves{\G,\ddiamond{\alpha}{\phi}}{B}{\psi}
        &\proves{\G,\ddiamond{\beta}{\phi}}{C}{\psi}}
{\proves{\G}{\edcase{A}{B}{C}}{\psi}}
}{}
\end{calculus}
\end{calculuscollections}}
\caption{\CdGL proof calculus: Propositional game rules}
\label{fig:cgl-rules-prop}
\end{figure}
\rref{fig:cgl-rules-prop} gives the propositional game rules.
Rule \irref{btestE} is modus ponens and \irref{btestI} is implication introduction because $\dbox{\ptest{\phi}}{\psi}$ is implication.
Angelic choices are disjunctions introduced by \irref{dchoiceIL} and \irref{dchoiceIR} and case-analyzed by \irref{dchoiceE}.
Angelic tests and Demonic choices are conjuctions introduced by \irref{dtestI} and \irref{bchoiceI}, eliminated by \irref{dtestEL}, \irref{dtestER}, \irref{bchoiceEL}, and \irref{bchoiceER}.
Rule \irref{hyp} applies an assumption.

\begin{figure}
  \centering
\hfill%
{\small\begin{calculuscollections}{\columnwidth}
\begin{calculus}
\cinferenceRule[brandomI|{$[{:}{*}]${I}}]{}
{
\linferenceRule[formula]
{\proves{\tsub{\G}{x}{y}}{M}{\phi}}
{\proves{\G}{(\etlam{\reals}{M})}{\dbox{\prandom{x}}{\phi}}}
}{} 
\cinferenceRule[drandomI|{$\langle{:}*\rangle${I}}]{}
{
  \linferenceRule[formula]
    {\proves{\G}{}{\ddiamond{\humod{x}{f}}{\phi}}}
    {\proves{\G}{\etcons{f}{M}}{\ddiamond{\prandom{x}}{\phi}}}
}{} 
\cinferenceRule[seqI|{$\lstrike{;}\rstrike$}I]{}
{\linferenceRule[formula]
  {\proves{\G}{M}{\dmodality{\alpha}{\dmodality{\beta}{\phi}}}}
  {\proves{\G}{\eSeq{M}}{\dmodality{\alpha;\beta}{\phi}}}
}{}
\cinferenceRule[asgnI|{$\lstrike{:=}\rstrike$}I]{}
{\linferenceRule[formula]
{\proves{\tsub{\G}{x}{y},x=\tsub{f}{x}{y}}{M}{\phi}}
{\proves{\G}{\eAsgneq{y}{x}{\pvx}{M}}{\dmodality{\humod{x}{f}}{\phi}}}
}{} 
\end{calculus}\hfill
\begin{calculus}
\cinferenceRule[brandomE|{$[{:}{*}]${E}}]{}
{
      \linferenceRule[formula]
        {\proves{\G}{M}{\dbox{\prandom{x}}{\phi}}}
        {\proves{\G}{\eapp{M}{f}}{\tsub{\phi}{x}{f}}}
}{}
\cinferenceRule[drandomE|{$\langle{:}*\rangle${E}}]{}
{
\linferenceRule[formula]
{\proves{\G}{M}{\ddiamond{\prandom{x}}{\phi}} & \proves{\G}{N}{\lforall{x}{(\phi \limply \psi)}}}
{\proves{\G}{\eunpack{M}{N}}{\psi}}
}{$x \notin \freevars{\psi}$}
\cinferenceRule[mon|{M}]{}
{\linferenceRule[formula]
  {\proves{\G}{M}{\dmodality{\alpha}{\phi}} & \proves{\earen{\G}{\alpha},\phi}{N}{\psi}}
  {\proves{\G}{\emon{M}{N}{\pvx}}{\dmodality{\alpha}{\psi}}}
}{}
\cinferenceRule[dualI|{$\lstrike{}^d\rstrike$}I]{}
{\linferenceRule[formula]
  {\proves{\G}{M}{\pmodality{\alpha}{\phi}}}
  {\proves{\G}{\eSwap{M}}{\dmodality{\pdual{\alpha}}{\phi}}}
}{}
\end{calculus}
\end{calculuscollections}
}
\caption{\CdGL proof calculus: First-order games ($y$ fresh, $f$ computable, $\tsub{\phi}{x}{f}$ admissible)}
\label{fig:cgl-compos}
\end{figure}
\rref{fig:cgl-compos} covers assignment, choice, sequencing, duals, and monotonicity.
Repetition games can be folded and unfolded (\irref{bunroll},\irref{broll}).
Angelic games have injectors (\irref{dstop},\irref{dgo}) and cases analysis (\irref{drcase}).
Duality \irref{dualI} switches players by switching modalities.
Sequential games (\irref{seqI}) are decomposed as nested modalities.

Monotonicity \irref{mon} is \rref{lem:monotone} in rule form.
The second premiss writes $\earen{\G}{\alpha}$ to indicate that the bound variables of $\alpha$ must be freshly renamed in $\G$ for soundness.
Rule \irref{mon} is used for generalization because all \GL's are subnormal, lacking axiom K (modal modus ponens) and necessitation.
Common uses include concise right-to-left symbolic execution proofs and, in combination with \irref{seqI}, Hoare-style sequential composition reasoning.

Nondeterministic assignments quantify over real-valued game variables.
Assignments \irref{asgnI} remember the initial value of $x$ in fresh variable $y$ ($\tsub{\G}{x}{y}$) for sake of completeness, then provides an assumption that $x$ has been assigned to $f$.
Skolemization \irref{brandomI} bound-renames $x$ to $y$ in $\G$, written $\eren{\G}{x}{y}$.
Specialization \irref{brandomE} instantiates $x$ to a term $f$.
Existentials are introduced by giving a witness $f$ in \irref{drandomI}.
Herbrandization \irref{drandomE} unpacks existentials, soundness requires $x$ is not free in $\psi$.

\begin{figure}
{\small\cinferenceRule[drcase|{$\langle*\rangle${E}}]{}
{
\linferenceRule[formula]
{ \proves{\G}{A}{\ddiamond{\prepeat{\alpha}}{\phi}} 
& \proves{\G,\phi}{B}{\psi} 
& \proves{\G,\ddiamond{\alpha}{\ddiamond{\prepeat{\alpha}}{\phi}}}{C}{\psi}}
{\proves{\G}{\ercase{A}{B}{C}}{\psi}}
}{}\hfill\cinferenceRule[bunroll|{$[*]${E}}]{}
{\linferenceRule[formula]
  {\proves{\G}{M}{\dbox{\prepeat{\alpha}}{\phi}}}
  {\proves{\G}{\ebunroll{M}}{\phi \land \dbox{\alpha}{\dbox{\prepeat{\alpha}}{\phi}}}}
}{}}

\begin{calculuscollections}{\columnwidth}
\begin{calculus}
\cinferenceRule[dstop|{$\langle*\rangle$S}]{}
{
\linferenceRule[formula]
{\proves{\G}{M}{\phi}}
{\proves{\G}{\estop{M}}{\ddiamond{\prepeat{\alpha}}{\phi}}}
}{}
\cinferenceRule[broll|{$[*]${R}}]{}
{\linferenceRule[formula]
  {\proves{\G}{M}{\phi \land \dbox{\alpha}{\dbox{\prepeat{\alpha}}{\phi}}}}
  {\proves{\G}{\ebroll{M}}{\dbox{\prepeat{\alpha}}{\phi}}}
}{}
\end{calculus}
\begin{calculus}
\cinferenceRule[dgo|{$\langle*\rangle$G}]{}
{
\linferenceRule[formula]
{\proves{\G}{M}{\ddiamond{\alpha}{\ddiamond{\prepeat{\alpha}}{\phi}}}}
{\proves{\G}{\ego{M}}{\ddiamond{\prepeat{\alpha}}{\phi}}}
}{}
\cinferenceRule[bloopI|{loop}]{}
{\linferenceRule[formula]
  {\proves{\G}{A}{J} & \proves{J}{B}{\dbox{\alpha}{J}}
  &\proves{J}{C}{\phi}}
  {\proves{\G}{(\erep{A}{B,C}{\pvx})}{\dbox{\prepeat{\alpha}}{\phi}}}
}{}
\end{calculus}
\begin{calculus}
\cinferenceRule[dloopE|FP]{}
{
\linferenceRule[formula]
{\proves{\G}{A}{\ddiamond{\prepeat{\alpha}}{\phi}}
        &\proves{\phi}{B}{\psi} & \proves{\ddiamond{\alpha}{\psi}}{C}{\psi}}
{\proves{\G}{\efp{A}{B}{C}}{\psi}}
}{}
\end{calculus}

\cinferenceRule[dloopI|{$\langle*\rangle${I}}]{}
{\linferenceRule[formula]
{
\deduce
  {\proves{\conv,(\met \metgr \metz  \land \met_0 = \met)}{B}{\ddiamond{\alpha}{(\conv\land \met_0 \metgr \met)}}}
  {\proves{\G}{A}{\conv} & \proves{\conv, \metz \metgeq \met}{C}{\phi}}
}
{\proves{\G}{\efor{A}{B}{C}}{\ddiamond{\prepeat{\alpha}}{\phi}}}
}{}
\end{calculuscollections}

\caption{\CdGL proof calculus: loops ($\met_0$ fresh)}
\label{fig:cgl-loops}
\end{figure}
\rref{fig:cgl-loops} provides rules for repetitions.
In rule \irref{dloopI}, $\met$ indicates an arbitrary termination metric where $\metgr$ denotes an arbitrary (effectively) well-founded~\cite{DBLP:journals/aml/HofmannOS06} partial order with some zero element $\metz$.
$\met_0$ is a fresh variable which remembers $\met$.
Angel plays $\prepeat{\alpha}$ by repeating an $\alpha$ strategy which always decreases the termination metric.
Angel maintains a formula $\conv$ throughout, and stops once $0 \metgeq \met$.
The postcondition need only follow from termination condition $0 \metgeq \met$ and convergence formula $\conv$.
Simple real comparisons $x \geq y$ are not well-founded, but inflated comparisons like $x \geq y + 1$ are.
Well-founded metrics ensure convergence in finitely (but often unboundedly) many iterations.
In the simplest case, $\met$ is a real-valued term.
Generalizing $\met$ to tuples enables, e.g., lexicographic termination metrics.
For example, the metric in the proof of \rref{ex:reach-avoid} is the distance to the goal, which must decrease by some minimum amount each iteration.

Rule \irref{dloopE} says $\ddiamond{\prepeat{\alpha}}{\phi}$ is a least pre-fixed-point.
It works backwards: first show $\psi$ holds after $\prepeat{\alpha},$ then preserve $\psi$ when each iteration is unwound.
Rule \irref{bloopI} is the repetition invariant rule.
Demonic repetition is eliminated by \irref{bunroll}.

Like any first-order program logic, \CdGL proofs contain first-order reasoning at the leaves.
Decidability of constructive real arithmetic is an open problem~\cite{constructiveRealAlgebra}, so first-order facts are proven manually in practice.
Our semantics embed \CdGL into type theory;
we defer first-order arithmetic proving to the host theory.
Note that even effectively-well-founded $\metgeq$ need not have decidable guards ($\met \metleq \metz \lor \met \metgeq 0$) since exact comparisons are not computable~\cite{bishop1967foundations}.
We may not be able to distinguish $\met = \metz$ from very small positive values of $\met,$ leading to one unnecessary loop iteration, after which $\met$ is certainly $\metz$ and the loop terminates.
Comparison up to $\veps > 0$ is decidable~\cite{bridges2007techniques} ($f > g \lor (f < g + \veps)$).


\begin{figure}
\begin{calculuscollections}{\textwidth}
\begin{calculus}
\cinferenceRule[di|DI]{}
{\linferenceRule[formula]
  {\proves{\G}{}{\phi} & \proves{\G}{}{\lforall{x}{(\ivr \limply \dbox{\humod{\D{x}}{f}}{\der{\phi}})}}}
  {\proves{\G}{}{\dbox{\pevolvein{\D{x}=f}{\ivr}}{\phi}}}
}{}
\cinferenceRule[dc|DC]{}
{\linferenceRule[formula]
  {\proves{\G}{}{\dbox{\pevolvein{\D{x}=f}{\ivr}}{R}} & \proves{\G}{}{\dbox{\pevolvein{\D{x}=f}{\ivr \land R}}{\phi}}}
  {\proves{\G}{}{\dbox{\pevolvein{\D{x}=f}{\ivr}}{\phi}}}
}{}
\end{calculus}
\begin{calculus}
    \cinferenceRule[dw|DW]{}
{\linferenceRule[formula]
  {\proves{\G}{}{\lforall{x}{\lforall{\D{x}}{(\ivr \limply \phi)}}}}
  {\proves{\G}{}{\dbox{\pevolvein{\D{x}=f}{\ivr}}{\phi}}}
}{}
\cinferenceRule[dg|DG]{}
{\linferenceRule[formula]
  {\proves{\G}{}{\lexists{y}{\dbox{\pevolvein{\D{x}=f,\D{y}=a(x)y + b(x)}{\ivr}}{\phi}}}}
  {\proves{\G}{}{\dbox{\pevolvein{\D{x}=f}{\ivr}}{\phi}}}
}{}
\end{calculus}
\end{calculuscollections}

\begin{calculus}
\cinferenceRule[dv|DV]{$t$ fresh, $x,\D{x},t,\D{t}$ not free in $d,\veps$}
{\linferenceRule[formula]
  {
  \deduce
    {\proves{\G}{}{\ddiamond{\humod{t}{0};\{\pevolvein{\D{t}=1,\D{x}=f}{\ivr}\}}{t \geq d}} 
\quad \proves{\G}{}{\dbox{\pevolve{\D{x}=f}}{(\der{h}-\der{g}) \geq \veps}}}
    {\proves{\ivr, h \geq g}{}{\phi} 
   & \proves{\G}{}{d > 0 \land \veps > 0 \land h-g \geq -d\veps}}
  }
  {\proves{\G}{}{\ddiamond{\pevolvein{\D{x}=f}{\ivr}}{\phi}}}
}{}
\cinferenceRule[bsolve|bsolve]{}
{\linferenceRule[formula]
  {\proves{\G}{}{\lforall[{\reals_{\geq0}}]{t}{((\lforall[{[0,t]}]{r}{\dbox{\humod{t}{r};\humod{x}{sln}}{\psi}})\limply\dbox{\humod{x}{sln};\humod{\D{x}}{f}}{\phi})}}}
  {\proves{\G}{}{\dbox{\pevolvein{\D{x}=f}{\ivr}}{\phi}}}
}{}
\cinferenceRule[dsolve|dsolve]{}
{\linferenceRule[formula]
  {\proves{\G}{}{\lexists[{\reals_{\geq0}}]{t}{((\lforall[{[0,t]}]{r}{\ddiamond{\humod{t}{r};\humod{x}{sln}}{\psi}})\land\ddiamond{\humod{x}{sln};\humod{\D{x}}{f}}{\phi})}}}
  {\proves{\G}{}{\ddiamond{\pevolvein{\D{x}=f}{\ivr}}{\phi}}}
}{}
\end{calculus}
\caption{\CdGL proof calculus: ODEs. In bsolve and dsolve, $sln$ solves $x'=f$ globally, $t$ and $r$ fresh, $x' \notin \freevars{\phi}$}
\label{fig:ode-rules}
\end{figure}
\rref{fig:ode-rules} gives the ODE rules, which are a constructive version of those from \dGL~\cite{DBLP:journals/tocl/Platzer15}.
For nilpotent ODEs such as the plant of \rref{ex:reach-avoid}, reasoning via solutions is possible.
Since \CdGL supports nonlinear ODEs which often do not have closed-form solutions, we provide invariant-based rules, which are complete~\cite{DBLP:conf/lics/PlatzerT18} for invariants of polynomial ODEs.
\emph{Differential induction} \irref{di}~\cite{DBLP:journals/logcom/Platzer10} says $\phi$ is an invariant of an ODE if it holds initially and if its \emph{differential formula}~\cite{DBLP:journals/logcom/Platzer10} $\der{\phi}$ holds throughout, for example $\der{f \geq g} \equiv (\der{f} \geq \der{g})$.
Soundness of \irref{di} requires differentiability, and $\der{\phi}$ is not provable when $\phi$ mentions nondifferentiable terms.
\emph{Differential cut} \irref{dc} proves $R$ invariant, then adds it to the domain constraint.
\emph{Differential weakening} \irref{dw} says that if $\phi$ follows from the domain constraint, it holds throughout the ODE.
\emph{Differential ghosts} \irref{dg} permit us to augment an ODE system with a fresh dimension $y,$ which enables~\cite{DBLP:conf/lics/PlatzerT18} proofs of otherwise unprovable properties.
We restrict the right-hand side of $y$ to be linear in $y$ and (uniformly) continuous in $x$ because soundness requires that ghosting $\D{y}$ does not change the duration of an ODE.
A linear right-hand side is guaranteed to be Lipschitz on the whole existence interval of equation $\D{x} = f,$ thus ensuring an unchanged duration by (constructive) Picard-Lindel\"{o}f~\cite{DBLP:conf/itp/MakarovS13}.
\emph{Differential variants}~\cite{DBLP:journals/logcom/Platzer10,DBLP:conf/fm/TanP19} \irref{dv} is an Angelic counterpart to \irref{di}.
The schema parameters $d$ and $\veps$ must not mention $x,\D{x},t,\D{t}$.
To show that $f$ eventually exceeds $g,$ first choose a duration $d$ and a sufficiently high minimum rate $\veps$ at which $h-g$ will change.
Prove that $h-g$ is decreases at rate at least $\veps$ and that the ODE has a solution of duration $d$ satisfying constraint $\ivr$.
Thus at time $d,$ both $h \geq g$ and its provable consequents hold.
Rules \irref{bsolve} and \irref{dsolve} assume as a side condition that $sln$ is the unique solution of $\D{x}=f$ on domain $\ivr$.
They are convenient for ODEs with simple solutions, while invariant reasoning supports complicated ODEs.

\section{Theory: Soundness}
\label{sec:soundness}

Following constructive counterparts of the classical soundness proofs for \dGL, we prove that the \CdGL proof calculus is sound: provable formulas are true in the CIC semantics.
We begin with standard lemmas.
\textbf{Full details in \rref{app:proofs}.}
\begin{lemma}[Uniform renaming]
  If $\proves{\G}{M}{\phi}$ then $\proves{\eren{\G}{x}{y}}{\eren{M}{x}{y}}{\eren{\phi}{x}{y}}$.
\label{lem:renaming}
\end{lemma}
\begin{lemma}[Coincidence]
Only free variables affect expressions' meaning.
\label{lem:coincidental}
\end{lemma}
\begin{lemma}[Bound effect]
Game execution modifies only bound variables.
\label{lem:bound-effect}
\end{lemma}
\begin{definition}[{Term substitution admissibility~\cite[{Def.\ 6}]{DBLP:journals/jar/Platzer08}}]
For a formula $\phi,$ (likewise for context $\G,$ term $f,$ and game $\alpha$) we say $\tsub{\phi}{x}{f}$ is \emph{admissible}
if $x$ never appears in $\phi$ under a binder of $\freevars{f} \cup \{x\}$.
\label{def:lem-admit}
\end{definition}
\begin{lemma}[Term substitution]
 If $\proves{\G}{M}{\phi}$ and the substitutions $\tsub{\G}{x}{f},$ and $\tsub{\phi}{x}{f}$ are admissible, then 
 $\proves{\tsub{\G}{x}{f}}{\tsub{M}{x}{f}}{\tsub{\phi}{x}{f}}$.
\label{lem:atsub}
\end{lemma}

Soundness of the proof calculus follows from the lemmas, and soundness of the ODE rules employing several known results from constructive analysis.
\begin{theorem}[Soundness]
  If $\proves{\G}{M}{\phi}$ is provable then $\seq{\G}{\phi}$ is valid.
  As a special case, if $(\proves{\Gemp}{M}{\phi})$ is provable, then $\phi$ is valid.
\label{thm:proof-calculus-sound}
\end{theorem}
\begin{proof}[Proof Sketch]
By induction on the derivation.
The assignment case holds by \rref{lem:atsub} and \rref{lem:renaming}.
\rref{lem:coincidental} and \rref{lem:bound-effect} are applied when maintaining truth of a formula across changing state.
The equality and inequality cases of \irref{di} and \irref{dv} employ the constructive mean-value theorem (\rref{thm:app-mvt} in \rref{app:proofs}), which has been formalized, e.g., in Coq~\cite{DBLP:conf/mkm/Cruz-FilipeGW04}.
Rules \irref{dw},  \irref{bsolve}, and \irref{dsolve} follow from the semantics of ODEs.
Rule \irref{dc} uses the fact that prefixes of solutions are solutions.
Rule \irref{dg} uses constructive Picard-Lindel\"{o}f~\cite{DBLP:conf/itp/MakarovS13}, which constitutes an algorithm for arbitrarily approximating the solution of any Lipschitz ODE, with a convergence rate depending on its Lipschitz constant.
\end{proof}
We have shown that every provable formula is true in the type-theoretic semantics.
Because the soundness proof is constructive, it amounts to an extraction algorithm from \CdGL into type theory:
for each proof, there exists a program in type theory which inhabits the corresponding type.

\section{Theory: Extraction and Execution}
\label{sec:synthesis}
Another perspective on constructivity is that provable properties must have witnesses.
We show Existential and Disjunction properties providing witnesses for existentials and disjunctions.
For modal formulas $\ddiamond{\alpha}{\phi}$ and $\dbox{\alpha}{\phi}$ we show proofs can be \emph{used as} winning strategies: a big-step operational semantics $\kwplay$ allows playing strategies against each other to extract a proof that their goals hold in some final state $s$.
Our presentation is more concise than defining the language, semantics, and properties of strategies, while providing key insights.
\begin{lemma}[Existential Property]
Let $s \in \sty$. If $M : (\ftrans{\sity{x}{\tau}{\phi}}\ s)$ then there exist terms $y:\tau$ and $N : (\ftrans{\tsub{\phi}{x}{y}}\ s)$.
\label{lem:term-ep}
\end{lemma}
\begin{lemma}[Disjunction Property]
  If $M : (\ftrans{\phi \lor \psi}\ s)$ then there exists an $N$ such that either $N : (\ftrans{\phi}\ s)$ or $N : (\ftrans{\psi}\ s)$.
\end{lemma}
Their proofs follow directly from their counterparts in type theory.
The Disjunction Property considers truth at a \emph{specific state}.
It is \emph{not} the case that validity of $\phi \lor \psi$ implies validity of either $\phi$ or $\psi$.
For example, $x < 1 \lor x > 0$ is valid, but its disjuncts are not.

Function $\kwplay$ below gives a big-step semantics: Angel and Demon strategies $\as$ and $\ds$ for respective goals $\phi$ and $\psi$ in game $\alpha$ suffice to construct a final state $s$ satisifying both.
By parametricity, $s$ was found by playing $\alpha$, because $\kwplay$ cannot inspect $p$ and $q,$ thus can only prove them via $\as$ and $\ds$.
\[\kwplay :
\pity{\alpha}{\textsf{Game}}{}
\pity{P,Q}{(\sfun{\sty}{\alltype})}{}
\pity{s}{\sty}{}
\sfun{\atrans{\alpha}\ P\ s}{\sfun{\dtrans{\alpha}\ Q\ s}{\sity{t}{\sty}{\ftrans{(P \land Q)}\ t}}}\]
Applications of $\kwplay$ are written $\play[\alpha]{\as}{\ds}{s}$ ($P$ and $Q$ implicit).
Game consistency (\rref{cor:consistency}) is by $\kwplay$ and consistency of type theory.
Note that  $\pdual{\alpha}$ is played by swapping the Angel and Demon strategies in $\alpha$.
{\small\begin{align*}
\play[\humod{x}{f}]{\as}{\ds}{s}              &= (\slet{t}{\lset{s}{x}{(f\ s)}} (t, (\as\ t, \ds\ t)))\\
\play[\prandom{x}]{\as}{\ds}{s}               &= \slet{t}{\lset{s}{x}{\sprojL{\as}}} (t, (\sprojR{\as}, \ds\ \sprojL{\as}))\\
\play[\pevolvein{\D{x}=f}{\ivr}]{\as}{\ds}{s} &=
\slet{(d,\text{sol},\text{solves},c,p)}{\as\ s} \\
&\phantom{=\ }(\lset{s}{x}{(\text{sol}\ d)}, (p,\ds\ d\ \text{sol}\ \text{solves}\ c))\\
\play[\ptest{\phi}]{\as}{\ds}{s}              &=  (s, (\sprojR{\as},\ds\ (\sprojL{\as})))\\
\play[\alpha\cup\beta]{\as}{\ds}{s}           &=
  \scase{(\as\ s)}\\
   \slbranch{\as'&}{\play[\alpha]{\as'}{(\sprojL{\ds})}{s}}\\
   \srbranch{\as'&}{\play[\beta]{\as'}{(\sprojR{\ds})}{s}}\\
\play[\alpha;\beta]{\as}{\ds}{s}              &= (\slet{(t,(\as',\ds'))}{\play[\alpha]{\as}{\ds}{s}}  \play[\beta]{\as'}{\ds'}{t})\\
\play[\prepeat{\alpha}]{\as}{\ds}{s}          &=
  \scase{(\as\ s)}\\
   \slbranch{\as'&}{(s, (\as', \sprojL{\ds}))}\\
   \srbranch{\as'&}{\slet{(t,(\as'',\ds''))}{\play[\alpha]{\as'}{(\sprojR{\ds})}{s}}}\\
   &\ \ \ \ \:\play[\prepeat{\alpha}]{\as''}{\ds''}{t}\\
\play[\pdual{\alpha}]{\as}{\ds}{s}            &= \play[\alpha]{\ds}{\as}{s}
\end{align*}}
\begin{corollary}[Consistency]
  It is never the case that both $\ftrans{\ddiamond{\alpha}{\phi}}\ s$ and $\ftrans{\dbox{\alpha}{\lnot \phi}}\ s$ are inhabited.
\label{cor:consistency}
\end{corollary}
\begin{proof}
Suppose $\as :\ftrans{\ddiamond{\alpha}{\phi}}\ s$ and $\ds : \ftrans{\dbox{\alpha}{\lnot \phi}}\ s,$
then $\sprojR{(\play[\alpha]{\as}{\ds}{s})} : \bot,$ contradicting consistency of type theory.
\end{proof}
The  $\kwplay$ semantics show how strategies can be executed.
Consistency is a theorem which ought to hold in any \GL and thus helps validate our semantics.


\section{Conclusion and Future Work}
We extended Constructive Game Logic \CGL to \CdGL for \emph{hybrid} games.
We contributed a new static and dynamic semantics.
We presented a natural deduction proof calculus for \CdGL and used it to prove reach-avoid correctness of 1D driving with adversarial timing.
We showed soundness and constructivity results.

The next step is to implement a proof checker, game interpreter, and synthesis tool for \CdGL.
Function $\kwplay$ is the high-level interpreter algorithm, while synthesis would commit to one Angel strategy and allow black-box Demon implementations.
In practice, Demon strategies represent some physical environment which is not implemented in type theory.
There is good justification to allow black-box treatment of Demon: the Demon connectives are \emph{negative} and thus defined by their observable behaviors.
Any program which behaves like a Demon \emph{is} a Demon.
Angel connectives are defined positively by their introduction forms, thus the task of synthesis is to extract these contents into code form.

\bibliographystyle{splncs04}
\bibliography{constructive-games,platzer}

\appendix
\newpage

\section{Example Proof}
\label{app:example-proofs}
It is understood that reading these appendices is \textbf{optional} for the reviewers.
We include the appendices in the event that a reviewer wishes to read them.
When published, the appendices will be published in an online-only extended version.

\input{hybrid-example-proofs}

\section{Theory Proofs}
\label{app:proofs}
\input{hybrid-proofs}
\end{document}



%% file: defs.tex
\irlabel{brandomI|{$[{:}*]$}}
\irlabel{bchoiceI|{$[\cup]$}}
\irlabel{btestI|{$[?]$}}
\irlabel{broll|{$[*]$r}}
\irlabel{bunroll|{$[*]${E}}}
\irlabel{bloopI|{$[*]$I}}
\irlabel{bsolve|{$[']$}}

\irlabel{di|DI}
\irlabel{dc|DC}
\irlabel{dw|DW}
\irlabel{dg|DG}
\irlabel{dv|DV}
\irlabel{dsolve|{$\langle'\rangle$}}
\irlabel{dloopI|{$\langle*\rangle$}I}
\irlabel{dstop|{$\langle*\rangle$}S}
\irlabel{dgo|{$\langle*\rangle$}G}
\irlabel{dchoiceIL|{$\langle\cup\rangle$}IL}
\irlabel{dchoiceIR|{$\langle\cup\rangle$}IR}
\irlabel{dtestI|{$\langle?\rangle$}I}
\irlabel{drandomI|{$\langle{:}*\rangle$}I}
\irlabel{seqI|{$\lstrike{;}\rstrike$}I}
\irlabel{seqE|{$\lstrike{;}\rstrike$}E}
\irlabel{asgnI|{$\lstrike{{:}=}\rstrike$}I}
\irlabel{asgnE|{$\lstrike{{:}=}\rstrike$}E}
\irlabel{drcase|{\m{\langle{\cup}\rangle}}E}
\irlabel{hyp|hyp}
\irlabel{mon|M}

\newcommand{\iheq}{\stackrel{=}{\text{\scriptsize{IH}}}}

\newcommand{\bebecomes}{\mathrel{::=}}
\newcommand{\alternative}{~|~}
\newcommand{\ivr}{\psi}

\newcommand{\dvar}{g}
\newcommand{\durvar}{d}

\newcommand{\ctrl}{\textsf{ctrl}\xspace}

\newcommand{\plant}{\textsf{plant}\xspace}

\newcommand{\eren}[3]{\urename[{#1}]{#2}{#3}}
\newcommand{\earen}[2]{\urename[{#1}]{\boundvars{#2}}{\vec{y}}}

\usepackage{prettyref}
\newcommand{\rref}[2][]{\prettyref{#2}}
\newrefformat{ch}{Chapter\,\ref{#1}}
\newrefformat{sec}{Section\,\ref{#1}}
\newrefformat{app}{Appendix\,\ref{#1}}
\newrefformat{def}{Def.\,\ref{#1}}
\newrefformat{thm}{Theorem\,\ref{#1}}
\newrefformat{prop}{Proposition\,\ref{#1}}
\newrefformat{lem}{Lemma\,\ref{#1}}
\newrefformat{rem}{Remark\,\ref{#1}}
\newrefformat{cor}{Corollary\,\ref{#1}}
\newrefformat{ex}{Example\,\ref{#1}}
\newrefformat{tab}{Table\,\ref{#1}}
\newrefformat{fig}{Fig.\,\ref{#1}}
\newrefformat{case}{case\,\ref{#1}}
\newrefformat{fn}{\ref{#1}}

\newcommand{\allstate}{\sty}

\renewcommand*{\freevars}[1]{\mathop{\text{FV}}(#1)}
\renewcommand*{\boundvars}[1]{\mathop{\text{BV}}(#1)}
\renewcommand*{\mustboundvars}[1]{\mathop{\text{MBV}}(#1)}
\newcommand*{\sfreevars}[1]{\mathop{\mathrel{\textsf{\upshape F}}\joinrel\mathrel{\textsf{\upshape V}}}(#1)}
\newcommand*{\sboundvars}[1]{\mathop{\mathrel{\textsf{\upshape B}}\joinrel\mathrel{\textsf{\upshape V}}}(#1)}
\newcommand*{\smustboundvars}[1]{\mathop{\mathrel{\textsf{\upshape M}}\joinrel\mathrel{\textsf{\upshape B}}\joinrel\mathrel{\textsf{\upshape V}}}(#1)}

\newcommand{\logname}{\text{\upshape\textsf{dH{\kern-0.05em}L}}\xspace}

\newcommand{\lequiv}{\leftrightarrow}

\newcommand{\DL}{\textsf{DL}\xspace}
\newcommand{\GL}{\textsf{GL}\xspace}
\newcommand{\CGL}{\textsf{CGL}\xspace}
\newcommand{\CdGL}{\textsf{CdGL}\xspace}

\newcommand{\ModelPlex}{ModelPlex\xspace}

\newcommand{\veps}{\varepsilon}

\newcommand{\seq}[2]{#1 \vdash #2}
\newcommand{\proves}[3]{#1\allowbreak\vdash #2 \allowbreak \mathop{:} #3}
\newcommand{\mproves}[3]{#1\allowbreak\vdash #2 \allowbreak \mathop{:} #3}

\newcommand{\edprojL}[1]{\langle\pi_1{#1}\rangle}
\newcommand{\edprojR}[1]{\langle\pi_2{#1}\rangle}
\newcommand{\ebprojL}[1]{[\pi_1{#1}]}
\newcommand{\ebprojR}[1]{[\pi_2{#1}]}

\newcommand{\edOpencons}{\langle}
\newcommand{\edClosecons}{\rangle}
\newcommand{\edSepcons}{,}
\newcommand{\edcons}[2]{\edOpencons{#1}\edSepcons{#2}\edClosecons}
\newcommand{\ebOpencons}{[}
\newcommand{\ebSepcons}{,}
\newcommand{\ebClosecons}{]}
\newcommand{\ebcons}[2]{\ebOpencons#1\ebSepcons#2\ebClosecons}

\newcommand{\pmodality}[2]{\llensehat{#1}\rlensehat{#2}}

\newcommand{\kwcase}{\textsf{case}}
\newcommand{\ecaseHead}[1]{\langle\kwcase\textrm{ }#1{\text{\textrm{ of }}}}
\newcommand{\ecaseLeft}[2]{#1\Rightarrow~#2}
\newcommand{\ecaseRight}[2]{~|~#1\Rightarrow~#2}
\newcommand{\ecaseEnd}{\rangle}
\newcommand{\ecasegen}[5]{\ecaseHead{#1}\allowbreak\ecaseLeft{#2}{#3}\allowbreak\ecaseRight{#4}{#5}\ecaseEnd}
\newcommand{\ecase}[3]{\ecasegen{#1}{\ell}{#2}{r}{#3}}
\newcommand{\edcase}[3]{\ecase{#1}{#2}{#3}}
\newcommand{\ercase}[3]{\langle\textsf{case}_*\ #1\text{ of }\pvs\Rightarrow~#2~|~\pvg\Rightarrow#3\rangle}

\newcommand{\edinjL}[1]{\langle\ell \cdot #1\rangle}
\newcommand{\edinjR}[1]{\langle r \cdot #1\rangle}

\newcommand{\kwrep}{\textsf{rep}}
\newcommand{\erep}[3]{{#1}\text{ }\kwrep\text{ }#3.~{#2}}
\newcommand{\eapp}[2]{#1\ #2}

\newcommand{\elam}[3]{\lambda #1:#2 \: #3}
\newcommand{\eplam}[2]{\elam{\pvx}{#1}{#2}}
\newcommand{\etlam}[2]{\elam{x}{#1}{#2}}

\newcommand{\eSeq}[1]{\lstrike\iota~#1\rstrike}

\newcommand{\eSwap}[1]{\lstrike\textsf{yield }#1\rstrike}

\newcommand{\emonInfix}[1]{{\circ_{#1}}}
\newcommand{\emon}[3]{{#1} \emonInfix{#3} {#2}}
\newcommand{\eQE}[2]{\textsf{FO}[#1](#2)}

\newcommand{\kwstop}{\textsf{stop}}
\newcommand{\kwgo}{\textsf{go}}
\newcommand{\estop}[1]{\langle\kwstop\ #1\rangle}
\newcommand{\ego}[1]{\langle\kwgo\ #1\rangle}
\newcommand{\kwroll}{\textsf{roll}}
\newcommand{\kwunroll}{\textsf{unroll}}
\newcommand{\ebroll}[1]{[\kwroll\ {#1}]}
\newcommand{\ebunroll}[1]{[\kwunroll\ #1]}

\newcommand{\eghost}[4]{\textsf{Ghost}[#1=#2](#3.~#4)}
\newcommand{\eAsgn}[4]{\lstrike\humod{#2}{\eren{f}{#2}{#1}}\text{ in }#3.~#4\rstrike}
\newcommand{\eAsgneq}[4]{\eAsgn{#1}{#2}{#3}{#4}}
\newcommand{\etconsgen}[5]{\langle{\eren{#4}{#1}{#2}}~{{:}{*}}~#3.~#5\rangle}
\newcommand{\etcons}[2]{\etconsgen{x}{y}{\pvx}{#1}{#2}}

\newcommand{\eunpack}[2]{\textsf{unpack}(#1,\pvx y.~#2)}

\newcommand{\efpgen}[5]{\textit{FP}(#1, #2.~#3, #4.~#5)}
\newcommand{\efp}[3]{\efpgen{#1}{\pvs}{#2}{\pvg}{#3}}

\newcommand{\met}{\ensuremath{\mathcal{M}}}

\newcommand{\metz}{{\boldsymbol{0}}}
\newcommand{\metgeq}{\boldsymbol{\succcurlyeq}}
\newcommand{\metgr}{\boldsymbol{\succ}}
\newcommand{\metleq}{\boldsymbol{\preccurlyeq}}
\newcommand{\metle}{\boldsymbol{\prec}}
\newcommand{\conv}{\varphi\xspace}
\newcommand{\G}{\Gamma}
\newcommand{\Gemp}{\cdot}

\newcommand{\eforHead}[4]{\textsf{for}(#1:\conv(\met)={#2};#3;{#4})} 
\newcommand{\eforBody}[1]{\{#1\}}
\newcommand{\eforgen}[5]{\eforHead{#1}{#2}{#3}{#4}\eforBody{#5}}
\newcommand{\efor}[2]{\eforgen{\pvx}{#1}{\pvy}{#2}{\alpha}}

\newcommand{\esub}[3]{[{#3}/{#2}]{#1}}
\newcommand{\tsub}[3]{\subst[#1]{#2}{#3}}

\newcommand{\mycase}{\textbf{Case}\xspace}

\newcommand{\ediv}[2]{#1 / #2}

\newcommand{\pvx}{p}
\newcommand{\pvy}{q}

\newcommand{\pvs}{s}
\newcommand{\pvg}{g}

\newcommand{\btt}{\text{\rm\tt{tt}}}
\newcommand{\bff}{\text{\rm\tt{ff}}}

\newcommand{\squarequotes}[2][0]{{%
  {\vphantom{#2}}^{\ulcorner}\kern-\scriptspace
  \mspace{-#1mu}%
  {{}#2}^{\urcorner}%
}}
\newcommand{\ftrans}[1]{\squarequotes[-0.5]{#1}}
\newcommand{\atrans}[1]{\langle\!\langle{#1}\rangle\!\rangle}
\newcommand{\dtrans}[1]{[[{#1}]]}

\newcommand{\alltype}{\mathbb{T}}

\newcommand{\lindty}[2]{\mu{#1}.\,{#2}}
\newcommand{\lcoty}[2]{\rho{#1}.\,{#2}}

\newcommand{\pity}[3]{\Pi #1\mathrel{:}#2.\,#3}
\newcommand{\sity}[3]{\Sigma #1\mathrel{:}#2.\,#3}
\newcommand{\xty}{\reals\xspace}
\newcommand{\sty}{\m{\mathfrak{S}}\xspace}
\newcommand{\lget}[2]{{#1}\ #2} 
\newcommand{\lset}[3]{\textsf{set}\ {#1}\ #2\ #3}
\newcommand{\kwprod}{\texttt{*}}
\newcommand{\kwsum}{\texttt{+}}
\newcommand{\sprod}[2]{#1\,\kwprod\,#2}
\newcommand{\sfun}[2]{#1 \to #2}
\newcommand{\vdot}{{\boldsymbol{\cdot}}}

\newcommand{\solves}[4]{#1,#2,#3 \vDash #4}

\newcommand{\kwplay}{\textsf{play}}
\newcommand{\play}[4][{}]{\kwplay_{#1}\ #4\ #2\ #3}

\newcommand{\sprojL}[1]{\pi_{L}#1}
\newcommand{\sprojR}[1]{\pi_{R}#1}

\newcommand{\as}{\textsf{as}}
\newcommand{\ds}{\textsf{ds}}

\newcommand{\slet}[2]{\textsf{let}\,#1\,=\,#2\,\textsf{in}\,}
\newcommand{\scase}[1]{\textsf{case}\ #1\ \textsf{of}}
\newcommand{\slbranch}[2]{~ #1 \Rightarrow #2}
\newcommand{\srbranch}[2]{|~ #1 \Rightarrow #2}

\makeatletter
\let\save@mathaccent\mathaccent
\newcommand*\if@single[3]{%
  \setbox0\hbox{${\mathaccent"0362{#1}}^H$}%
  \setbox2\hbox{${\mathaccent"0362{\kern0pt#1}}^H$}%
  \ifdim\ht0=\ht2 #3\else #2\fi
  }
\newcommand*\rel@kern[1]{\kern#1\dimexpr\macc@kerna}
\newcommand*\widebar[1]{\@ifnextchar^{{\wide@bar{#1}{0}}}{\wide@bar{#1}{1}}}
\newcommand*\wide@bar[2]{\if@single{#1}{\wide@bar@{#1}{#2}{1}}{\wide@bar@{#1}{#2}{2}}}
\newcommand*\wide@bar@[3]{%
  \begingroup
  \def\mathaccent##1##2{%
    \let\mathaccent\save@mathaccent
    \if#32 \let\macc@nucleus\first@char \fi
    \setbox\z@\hbox{$\macc@style{\macc@nucleus}_{}$}%
    \setbox\tw@\hbox{$\macc@style{\macc@nucleus}{}_{}$}%
    \dimen@\wd\tw@
    \advance\dimen@-\wd\z@
    \divide\dimen@ 3
    \@tempdima\wd\tw@
    \advance\@tempdima-\scriptspace
    \divide\@tempdima 10
    \advance\dimen@-\@tempdima
    \ifdim\dimen@>\z@ \dimen@0pt\fi
    \rel@kern{0.6}\kern-\dimen@
    \if#31
      \overline{\rel@kern{-0.6}\kern\dimen@\macc@nucleus\rel@kern{0.4}\kern\dimen@}%
      \advance\dimen@0.4\dimexpr\macc@kerna
      \let\final@kern#2%
      \ifdim\dimen@<\z@ \let\final@kern1\fi
      \if\final@kern1 \kern-\dimen@\fi
    \else
      \overline{\rel@kern{-0.6}\kern\dimen@#1}%
    \fi
  }%
  \macc@depth\@ne
  \let\math@bgroup\@empty \let\math@egroup\macc@set@skewchar
  \mathsurround\z@ \frozen@everymath{\mathgroup\macc@group\relax}%
  \macc@set@skewchar\relax
  \let\mathaccentV\macc@nested@a
  \if#31
    \macc@nested@a\relax111{#1}%
  \else
    \def\gobble@till@marker##1\endmarker{}%
    \futurelet\first@char\gobble@till@marker#1\endmarker
    \ifcat\noexpand\first@char A\else
      \def\first@char{}%
    \fi
    \macc@nested@a\relax111{\first@char}%
  \fi
  \endgroup
}
\makeatother


%% file: hybrid-example-proofs.tex
\newcommand{\SB}{\text{SB}}
\newcommand{\LS}{\text{LS}}
\newcommand{\RS}{\text{RS}}
\newcommand{\SLS}{\text{SLS}}
\newcommand{\SRS}{\text{SRS}}
\newcommand{\UV}{\text{UV}}
\newcommand{\LV}{\text{LV}}
\newcommand{\BM}{\text{BM}}
\newcommand{\AM}{\text{AM}}
\newcommand{\LT}{\text{LT}}
\newcommand{\ST}{\text{ST}}
\newcommand{\UL}{\text{UL}}
\newcommand{\UR}{\text{UR}}
\newcommand{\acc}{\text{acc}}

We restate the definition of the 1D driving game and its reach-avoid specification here:
\begin{align*}
  \ctrl &\equiv \prandom{a}; \ptest{-B \leq a \leq A}; \humod{t}{0} \\
  \plant &\equiv \pdual{\{\pevolvein{\D{t}=1, \D{x}=v,\D{v}=a}{t \leq T \land v \geq 0}\}} \\
  \pre &\equiv T > 0 \land A > 0 \land B > 0 \land v=0 \land x=0\\
  \reachavoid &\equiv \pre \limply \ddiamond{\prepeat{(\ctrl;\plant;\ptest{x \leq g};\{\pdual{\ptest{t > T/2}}\})}}{(g=x \land v=0)}
\end{align*}

We first give an overview of our proof approach, then give the main algebraic derivations, then finally the complete natural deduction proof.

\subsection{Proof overview}
While safety of 1D driving is a thoroughly studied introductory example, adversarial 1D reach-avoid is more challenging due to the combination of adversarial timing and liveness.

To simplify our arithmetic, the proof uses the same acceleration magnitude $C = \min(A,B)$ to both accelerate and brake.
The resulting strategy is conservative, but still satisifies reach-avoid correctness.
Our proof proceeds by convergence: we establish a minimum distance $\Delta{x}$ which is traversed in each iteration, guaranteeing that the goal is eventually reached.
The minimum distance is determined by appealing to a velocity envelope which is invariant throughout the loop, given as a function of the position.
Much of \rref{sec:derivations} is devoted to identifying a velocity envelope which is invariant while also strong enough to ensure liveness.
The car's velocity goes to $0$ during braking, so the key is to show that velocity decreases slowly enough to ensure progress.
We depict the safe driving envelope in \rref{fig:full-envelope}.

\begin{figure}
  \centering
  \includegraphics[width=3in]{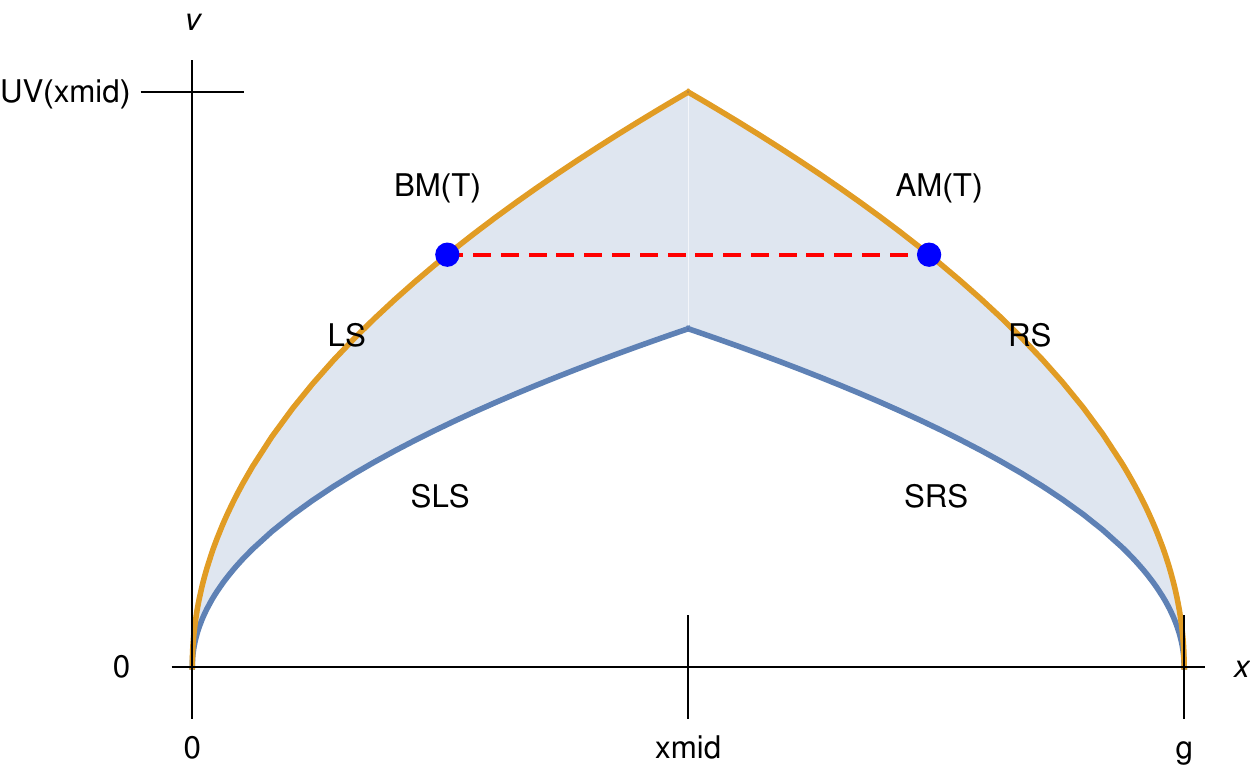}
  \caption{Safe driving envelope}
  \label{fig:full-envelope}
\end{figure}

A major task of the controller is to detect events within an adversarial environment.
We detect one event: when are we close enough to the goal that we must brake?
Because our acceleration and braking rates are the same, it suffices to begin braking by the midpoint $x = g/2$.
Care is still required because the controller is \emph{time-triggered}: we must determine whether it is possible to cross the midpoint within the coming timestep, and react right \emph{before} we actually cross the midpoint.
In \rref{fig:full-envelope}, the blue point $BM(T)$ is the point at which we would start to react.

The other  major task of the controller is to choose an acceleration value.
Until we approach the midpoint, the acceleration is at the maximum value $C$.
Once the midpoint is detected, acceleration is computed with a predictive method: compute the state at the end of the timestep, and solve for the greatest $C$ that maintains safety.


Recall that \CdGL features Type-2 effective computations on reals.
We note that in proofs which use inexact comparisons for constructive reals, only approximate ``reach'' properties might be provable.
The reason our proof obtains an exact result is subtle: in Type-2 effectivity, the functions $\min(f,g)$ and $\max(f,g)$ are exact comparisons, in contrast to inexact formula-level comparisons.
Because  $\min(f,g)$ and $\max(f,g)$ are \emph{terms}, Type-2 effectivity demands only that when real numbers are represented as lazy streams of bits, the binary representation of $\min(f,g)$ or $\max(f,g)$ could be computed \emph{lazily} from the bits of $f$ and $g$.
Exact $\min$ and $\max$ \emph{are} implementable lazily: each bit of $f$ is compared with the corresponding bit of $g$.
In the case that $f$ and $g$ have identical bit representations, this process will lazily return their bit representation.
In the case that the bit representations of $f$ and $g$ differ, the extremum will start by returning identical bits, then commit to a choice of $f$ or $g$ once a differing bit is found.
While there exist numbers with multiple binary representations (${1.0}_{2} = {0.111\ldots}_{2}),$ this simply means $\min$ and $\max$ are free to return either representation.

\subsection{Algebraic derivations}
\label{sec:derivations}
We now algebraically derive the main equations of the proof, e.g., for invariant regions, termination metrics, and acceleration control.
In this section, ``monotonicity'' does not refer to the monotonicity rule for game modalities, but to monotone functions in the arithmetic sense, e.g.,
\[x \geq y \limply f(x) \geq f(y)\]

The following employs the well-known Newtonian motion equations:
\begin{align*}
  v(t) &= v(0) + at\\
  x(t) &= x(0) + v(0)t + \frac{at^2}{2}
\end{align*}
Because our initial conditions are $v(0) = x(0) = 0,$ we can often eliminate the first terms.
We write $x(k)$ and $v(k)$ for the values of $x$ and $v$ at the beginning of a given iteration of the game loop, as opposed to the beginning of the game.
From the Newton equations we derive the safe-braking ($\SB$) inequality, which says braking rate $-C$ suffices to stop the car by the time $x$ reaches $g$:
\[\SB \equiv  \left(\frac{v^2}{2C} \leq g-x\right)\]
When the equality is strict, then $x$ reaches $g$ exactly as the car stops.

We write $x_{mid}$ for the midpoint, i.e., $g/2$.
We write $v_{mid}$ for the maximum velocity which may ever be attained, i.e., the maximum velocity one might have at the midpoint and still brake safely.
We write $\UV(x)$ for the \textbf{u}pper bound of permissible \textbf{v}elocity at a position $x$.
The shape of $\UV$ will be a convex ``triangle'' which is 0 at $x=0$ and $x=g$ and attains value $v_{mid}$ at $x = x_{mid}$.
The shape of $\UV$ is convex rather than a true triangle because its edges are defined by square-root expressions defining the relationship between position and velocity.

Its \textbf{l}eft \textbf{s}ide ($\LS$) is derived by setting the maximum acceleration $a = C$ and solving the Newton equations for $v$ as a function of $x$:
\[\LS(x) = \sqrt{2Cx}\]
The \textbf{r}ight \textbf{s}ide ($RS$) is derived by setting the braking acceleration to $C,$ assuming $SB$ holds as a strict equality, and solving for $v$ as a function of $x$:
\[\RS(x) = \sqrt{2C(g-x)}\]
The upper velocity envelope is their minimum:
\[\UV(x) = \min(\LS(x),\RS(x))\]
We refer to the set of points bounded by these curves as the Large Triangle (\LT), with the understanding that its sides are actually convex curves.


Choosing a \textbf{l}ower \textbf{v}elocity bound $\LV(x)$ at each point $x$ is more challenging, because velocity necessarily decays as $x$ approaches $g$.
To prove that position $x=d$ is eventually reached, our proof must rule out so-called ``Zeno'' behaviors, such as those where distanced traversed in each timestep decreases exponentially.
Ruling out Zeno behaviors requires a somewhat strong lower bound $\LV$.
Strong bounds are of course desirable in and of themselves: the higher the bound, the faster $x$ is guaranteed to reach $g$.

Our control scheme predicts future motion: a control choice is safe if for every duration $t \in [0,T],$ the motion is safe.
By monotonicity, it suffices to show the case  $t = T$.
The lower bound likewise predicts motion, we introduce several helper functions which predict motion.
We write $\BM(T)$ (before midpoint) for ``the position $x$ from which the car will reach the midpoint in time $T$ at maximum acceleration.''
Under time-triggered control, $\BM(T)$ is the ``point of no return'' by which Angel must react, and we will safety detect the approaching midpoint by comparing our position to $\BM(T)$.  
We also write $\AM(T)$ (after midpoint) for the conjugate point of $\BM(T)$ opposite the midpoint.
To derive $\BM(T)$ we set $x(k) + \UV(x(k))T + \frac{CT^2}{2} = g/2$ with the simplification that $\UV(x(k)) = \LS(x(k))$ before the midpoint.
We solve (by computer):
\[x + \sqrt{2Cx}T  + \frac{CT^2}{2} = g/2\]
which yields
\[\BM(T) = \frac{1}{2} (-2 T\sqrt{Cg} + C T^2 + g)\]
and its conjugate
\[\AM(T) = \frac{1}{2} (2 T\sqrt{Cg} - C T^2 + g)\]

It is clear that control must be more conservative past $\BM(T)$, the only question is \emph{how} conservative.
For example, the minimum safe acceleration from $\BM(T/2)$ is $0$, which takes us to the boundary at $\AM(T/2)$ if Demon chooses $t=T$.
We might wonder if it is sufficient to exclude all states above the line connecting $\BM(T/2)$ and $\AM(T/2),$ indicated by a red dashed line in   \rref{fig:full-envelope}.
It is not, under adversarial timing. Consider blue point $\BM(T/2)$ again.
Demon may choose $t = T/2$ so that by definition Angel regains control at $x = x_{mid}$ but $v < v_{mid}$.
Demon now has a strategy to keep Angel strictly in the interior of the Large Triangle indefinitely, so that the lower bound is eventually violated, probably after $\AM(T/2)$.

Our envelope can only hope to be an invariant if a \emph{strict} inequality $\LV(x) < \UV(x)$ holds for \emph{all} $x \in  (x_{mid},g)$.
We believe that the optimal lower bound is particularly nontrivial, so we do not aim to show our bound is perfectly tight.
We do note that our bound is not exceptionally loose either.
For example, if we were to permit Demon to elapse physics up to time $2T,$ then any strategy more aggressive than ours becomes clearly unsafe.
Regardless, we believe the controller used in this proof is tight modulo our simplifying assumption $A = B = C$, we simply use this slightly looser bound for the sake of \emph{proving liveness}.

The starting observation is that Angel might need to decrease acceleration as early as $\BM(T)$.
The simplest live decision would be to construct a braking rate $a_{end} > 0$ which reaches $g$ exactly when $v=0,$ and simply brake at rate $ a = -a_{end}$ until stopped.
The braking curve of rate $a_{end}$ form the small right side (SRS) of the triangle, while its mirror forms the small left side (SLS).
Perhaps we could reuse $\LS$ for the left side, but we preserve symmetry in hopes of simplifying the proof.

Consider the upper-left point  $\UL = (\BM(T), \UV(\BM(T)))$ and upper-right point $\UR = (\AM(T), \UV(\AM(T)))$.
The conservative braking rate $a_{end}$ 
is defined by setting $\SB$ as an equality and solving for $a$ as a function of $x$ and $v$.


\begin{align*}
           &\left(\frac{\UV(\BM(T))^2}{2a_{end}} = g - \BM(T)\right)\\
\text{iff} &\left(a_{end} = \frac{\LS(\BM(T))^2}{2(g- \BM(T))}\right)\\
= & \frac{\LS(\frac{1}{2} (-2 T\sqrt{C g}  + C T^2 + g))^2}{2(g-\frac{1}{2} (-2 T\sqrt{Cg} + C T^2 + g))}\\
= & \frac{\LS(\frac{1}{2} (-2 T\sqrt{C g}  + C T^2 + g))^2}{2g + 2 T\sqrt{Cg}  - C T^2 - g}\\
= & \frac{\LS(\frac{1}{2} (-2 T\sqrt{C g}  + C T^2 + g))^2}{g  + 2 T\sqrt{Cg}  - C T^2}\\
= & \frac{\LS(- T\sqrt{C g} + CT^2/2 + g/2))^2}{g + 2 T \sqrt{Cg} - C T^2}\\
= & \frac{\left(\sqrt{2C(-T\sqrt{C g} + C/2 T^2 + g/2)}\right)^2}{g + 2 T\sqrt{Cg}  - C T^2}\\
= & \frac{2C(-T\sqrt{C g} + CT^2/2 + g/2)}{g + 2 T \sqrt{Cg} - C T^2}\\
= &\frac{C (\sqrt{g} - T \sqrt{C})^2}{g + 2 T \sqrt{Cg} - C T^2}
\end{align*}

Then we define the Small Triangle (\ST) as the set of points bounded by:
\[\SLS(x) = \sqrt{2 a_{end} x}\]
\[\SRS(x) = \sqrt{2 a_{end}(g-x)}\]
The lower velocity envelope is their minimum:
\[\LV(x) = \min(\SLS(x),\SRS(x))\]

We are finally ready to give the variant formula $J$ and the progress term $\Delta{x}$ which induces the termination metric.
The variant simply says the velocity is between the envelopes and gives the signs of state variables:
\[J \lequiv \LV(x) \leq v \leq \UV(x) \land x \geq 0 \land v \geq 0\]

The minimum progress is found by taking the shortest distance traversed along any path of duration $T/2$ contained within the envelope between triangles \LT and \ST:
\[\Delta{x}= \inf_{(x,v)\in \LT \setminus \ST, t \in [T/2,T]} \inf_a (vt/2 + at^2/2)\]
where $a$ ranges over accelerations which remain within the envelope for time $t$.
We observe that distance traversed is monotone in $x, v,$ and $t$.
Thus it suffices to consider the worst case where $t = T/2, x=0, v=0, a = -a_{end}$.
The worst-case acceleration is $a = a_{end}$ because $a_{end}$ defines the lower bound of the velocity envelope.
The worst case is then \[\frac{a_{end} T^2}{8}\]
We construct our termination argument.
We use $\met \equiv (g-x)$ as our termination metric  with terminal value $\metz = 0$ and define
$a \metle b \lequiv a + \Delta{x} \leq b$.
As usual, the convergence proof will need only prove a disjunction: either $\met$ has decreased \emph{or} it is equal to zero.
It sometimes happens that the penultimate step actually reaches $x=g$ but can only prove that $(g-x)$ has decreased, in which case the final step will observe $x=g$ and terminate.
We are simply observing that such behavior is permissible: when we observe the goal has already been reached, it is irrelevantly how much progress the final (i.e., previous) step had made.

The invariant and metric are major components of the proof.
The last major component is the strategy for choosing $a,$ which we have only alluded to thus far.
We wish to set the highest acceleration that is guaranteed to remain within the velocity envelope.

To find this acceleration, recall the motion equations:
\begin{align*}
x(k+t) &= x(k) + v(k)t + a \frac{t^2}{2} \\
v(k+t) &= v(k) + at
\end{align*}
where $v(k)$ and $x(k)$ are the values of state variables $v$ and $x$ at the start of the current iteration of the game loop.

The most aggressive safe acceleration is that which satisfies $SB$ as a strict equality after the pessimal time interval $T$, so we set
\[\frac{(v + aT)^2}{2C} = \left(g - (x(k) + v(k)t + a \frac{t^2}{2})\right)\]
and solve for $a$.
Wolfram Alpha gives two conjugate solutions (assuming $T \neq 0$ and $C \neq 0,$ which are true):
\[a = -\frac{\sqrt{C} \sqrt{C T^2 + 8 g - 4 T v - 8 x} + C T + 2 v}{2 T}\]
\[a = \frac{\sqrt{C} \sqrt{C T^2 + 8 g - 4 T v - 8 x} + C T + 2 v}{2 T}\],
the latter of which is positive.

We take this second solution as a candidate for the acceleration:
\[a_{cand} = \frac{\sqrt{C} \sqrt{C T^2 + 8 g - 4 T v - 8 x} - C T - 2 v}{2 T}\]
Recall that accelerations are required to fall within the range $[-B,A]$ and that for simplicitly we show the stronger condition $a \in [-C,C]$.
The lower bound $a_{cand} \geq -C$ holds by construction: $a_{cand}$ is the \emph{greatest} acceleration which remains within the safe envelope.
By construction of the envelope, an acceleration $-C$ always remains below the upper limit, so $a_{cand}$ must be at least $-C$.
The upper bound $a_{cand} \leq C$ does not hold in general: we computed $a_{cand}$ as the acceleration required to reach the maximum velocity \emph{in this timestep}, when reaching maximum velocity usually takes multiple timesteps.
Thus the final acceleration ($\acc$) is computed by bounding $a_{cand}$ against the upper limit $C$:
\[\acc \equiv \min\left(C, a_{cand}\right)\]


\newcommand{\ndd}[1]{\m{\mathcal{D}_{\text{#1}}}}
\irlabel{implyR|implyR}
\irlabel{FO|FO}
\irlabel{pf|{}}
\subsection{Natural Deduction Proof}
We give a formal proof in the natural deduction calculus.
We first give derived rules and lemmas used in the proof, and we split the deduction into small pieces for the sake of formatting.

\subsubsection{Derived Rules}
The vacuity axiom for constant propositions (indicated $p()$) is not sound for games~\cite{DBLP:journals/tocl/Platzer15}:
\[p() \not\limply \dbox{\alpha}{p()}\]
The following rule is sound for games, however, and can be derived from \rref{lem:bound-effect} and \rref{lem:monotone}:
\[
\cinferenceRule[gV|GV]{}
{\linferenceRule[formula]
  {\proves{\G}{}{p()} & \proves{\G}{}{\ddiamond{\alpha}{Q}}}
  {\proves{\G}{}{\ddiamond{\alpha}{p()}}}
}{}
\]
The formula $Q$ is arbitrary: as soon as Angel has \emph{any} winning strategy, vacuity becomes sound.
$Q = \btt$ is usually chosen in practice.

As discussed in \rref{sec:proof-calculus}, we do not axiomatize first-order reasoning in this paper, but assume it has been implemented in ``the host logic''
Thus we label first-order steps ``FO'' but do give arithmetic proofs in full axiomatic detail.
To be precise, the following (non-effective!) rule is sound:
\[\cinferenceRule[FO|{$\text{FO}_C$}]{}
{\linferenceRule[formula]
{\lclose}
{\proves{\G}{\eQE{\phi}{M}}{\phi}}}{\text{exists }M : (\pity{s}{\sty}{\ftrans{\G\limply\phi}\ s})\text{ and }\G,\phi \text{ F.O.}}
\]

\subsubsection{Lemmas}
We use several arithmetic facts throughout the proof.
\begin{lemma}[Safe Upper Bound]
Braking is safe when the upper velocity bound is satisified.
Formula $v \leq \UV(x) \limply (g-x) \geq \SB(x,v)$ is provable in context $C > 0$.
\label{lem:safe-upper-bound}
\end{lemma}
\begin{proof}
By first-order arithmetic.
\end{proof}

\begin{lemma}[Acceleration in Bounds]
Our control algorithm only proposed accelerations which are feasible.
Formula $\SB \limply   -C \leq \acc \leq C \limply -B \leq \acc \leq A]$ is provable in context
$A > 0 \land B > 0 \land C = \min(A,B) \land (g-x) \geq \SB$.
\label{lem:accel-in-bounds}
\end{lemma}
\begin{proof}
The lower bound holds by construction: as discussed in the last section, $\acc \geq -C$ when $(g-x) \geq \SB$.
The upper bound holds trivially because $\acc$ is computed by bounding $a_{cand}$ to $C$.
\end{proof}

\subsubsection{Main Proof}

The main proof begins by applying the convergence rule \irref{dloopI}, which is parameterized by a well-order $(\met, \metz, <_{\met})$ and an invariant $\varphi$.
The convergence metric is the remaining distance $\met = (g-x)$ and the zero element is $\metz = 0$.
The ordering $<_{\met}$ defines $0$ less than all other states and otherwise defines
$(s \leq t) \equiv ((g-x)\ s + \Delta{x} \leq (g-x)\ t)$.
The invariant $\conv$ says variables' signs are preserved and that velocity remains within its envelope:
\[\conv \lequiv (\LV(x) \leq v \leq \UV(x) \land x \geq 0 \land v \geq 0)\]
Note we do not include the sign conditions on $T, A, B, C$ in the invariant because they are constants.
By \irref{gV}, for any game $\alpha$ and constant proposition (indicated $p()$), we have $p() \limply \ddiamond{\alpha}{p()}$ whenever
$\ddiamond{\alpha}{\psi}$ holds for \emph{any} postcondition $\psi$.
Loop convergence contains such a proof, thus vacuity can always applied to constants of a convergence proof in the inductive step.

\begin{sequentdeduction}[array]
\linfer[implyR]
{
  \linfer[dloopI]{
      \linfer[pf]{\ndd{pre}}{\lsequent{\pre}{\conv}}
    ! \linfer[pf]{\ndd{body}}{\lsequent{\conv,(\metz \metleq \met \land \met_0 = \met)}{(\conv\land \met_0 \metgr \met)}}
    ! \linfer[pf]{\ndd{post}}{\lsequent{\conv \land \metz \metgeq \met}{\phi}}
  }
  {\lsequent{\pre}{\ddiamond{\prepeat{\alpha}}{\post}}}
}
{\lsequent{}{\pre \limply \ddiamond{\prepeat{\alpha}}{\post}}}
\end{sequentdeduction}

We first dispatch the precondition and postcondition steps, which are purely arithmetic.

\begin{sequentdeduction}[array]
\linfer[FO]
  {\lclose}
  {\lsequent{T > 0 \land A > 0 \land B > 0 \land v=0 \land x=0}{(\LV(x) \leq v \leq \UV(x) \land x \geq 0 \land v \geq 0)}}
\end{sequentdeduction}
By construction, when $x=v=0$ then $\LV(x) = \UV(x) = 0$ and the first two conjuncts are trivially satisfied.
The latter two conjuncts follow directly from $v=0$ and $x=0$.

\begin{sequentdeduction}[array]
\linfer[FO]
  {\lclose}
  {\lsequent{(\LV(x) \leq v \leq \UV(x) \land x \geq 0 \land v \geq 0) \land 0 \geq (g-x)}{(g=x \land v=0)}}
\end{sequentdeduction}
Since $v \leq \UV(x)$ then (SB) $\frac{v^2}{2C} \leq (g-x)$. The LHS is always nonnegative, so $x \leq g$.
Since $0 \geq (g-x)$ then $x \geq g$ so $g=x$ as desired.
Moreover $\frac{v^2}{2C} \leq (g-x) = 0$ so $v=0$ as desired.

We proceed to the proof of the loop body.
We abbreviate $\G \equiv \LV(x) \leq v \leq \UV(x) \land x \geq 0 \land v \geq 0, 0 \leq (g-x) \land \met_0 = (g-x)$
and $\phi \equiv \LV(x) \leq v \leq \UV(x) \land x \geq 0 \land v \geq 0\land \met_0 \geq \Delta{x} + (g - x)$.

The first \irref{dtestI} step applies the lemma $\SB \limply \acc \in [-B,A]$.
The second \irref{dtestI} step applies the lemma $\ndd{arith1}$.
In the \irref{dsolve} step, we write $X(t)$ and $V(t)$ for the solutions of $x$ and $v$, where
$X(t) = x + vt + acc \frac{t^2}{2}$ and $V(t) = v + \acc\,t$.
The domain constraint assumption in the \irref{dsolve} step has been simplified monotonely.
\begin{sequentdeduction}[array]
\linfer[seqI] {
  \linfer[seqI] {
    \linfer[drandomI] {
      \linfer[seqI] { %
        \linfer[dtestI] { %
          \linfer[seqI] { %
            \linfer[asgnI] { %
              \linfer[seqI] { %
                \linfer[dualI] { %
                  \linfer[dsolve] { %
                    \linfer[seqI] {
                      \linfer[dtestI] { %
                        \linfer[dualI] { %
                          \linfer[btestI] {
                            \linfer[pf]
                              {\ndd{arith2}}
                              {\lsequent{\G,a=\acc,T/2 \leq t \leq T, V(t) \geq 0}{\esub{\phi}{(x,v)}{(X(t),V(t))}}}
                            } 
                            {\lsequent{\G,a=\acc,0 \leq t \leq T, V(t) \geq 0}{\esub{\dbox{\ptest{t > T/2}}{\phi}}{(x,v)}{(X(t),V(t))}}}} 
                          {\lsequent{\G,a=\acc,0 \leq t \leq T, V(t) \geq 0}{\esub{\ddiamond{\pdual{\ptest{t > T/2}}}{\phi}}{(x,v)}{(X(t),V(t))}}}} 
                        {\lsequent{\G,a=\acc,0 \leq t \leq T, V(t) \geq 0}{\esub{\ddiamond{\ptest{\safe}}{\ddiamond{\pdual{\ptest{t > T/2}}}{\phi}}}{(x,v)}{(X(t),V(t))}}}} 
                      {\lsequent{\G,a=\acc,0 \leq t \leq T, V(t) \geq 0}{\esub{\ddiamond{\ptest{\safe};\pdual{\ptest{t > T/2}}}{\phi}}{(x,v)}{(X(t),V(t))}}}} 
                    {\lsequent{\G,a=\acc,t=0}{\dbox{\pevolvein{\D{x}=v,\D{v}=a}{t \leq T \land v \geq 0}}{\ddiamond{\ptest{\safe};\pdual{\ptest{t > T/2}}}{\phi}}}}} 
                  {\lsequent{\G,a=\acc,t=0}{\ddiamond{\plant}}{\ddiamond{\ptest{\safe};\pdual{t \geq T/2}}{\phi}}}} 
                {\lsequent{\G,a=\acc,t=0}{\ddiamond{\plant;\ptest{\safe};\pdual{t \geq T/2}}{\phi}}}} 
              {\lsequent{\G,a=\acc}{\ddiamond{\humod{t}{0}}{\ddiamond{\plant;\ptest{\safe};\pdual{t \geq T/2}}{\phi}}}}} 
            {\lsequent{\G,a=\acc}{\ddiamond{\humod{t}{0};\plant;\ptest{\safe};\pdual{t \geq T/2}}{\phi}}}} 
          {\lsequent{\G,a=\acc}{\ddiamond{\ptest{-B \leq a \leq A}}{\ddiamond{\humod{t}{0};\plant;\ptest{\safe};\pdual{t \geq T/2}}{\phi}}}}} 
        {\lsequent{\G,a=\acc}{\ddiamond{\ptest{-B \leq a \leq A}; \humod{t}{0};\plant;\ptest{\safe};\pdual{t \geq T/2}}{\phi}}}}
      {\lsequent{\G}{\ddiamond{\prandom{a}}{\ddiamond{\ptest{-B \leq a \leq A}; \humod{t}{0}}{\ddiamond{\plant;\ptest{\safe};\pdual{t \geq T/2}}{\phi}}}}}} 
    {\lsequent{\G}{\ddiamond{\ctrl}{\ddiamond{\plant;\ptest{\safe};\pdual{t \geq T/2}}{\phi}}}}} 
  {\lsequent{\G}{\ddiamond{\alpha}{\phi}}}
\end{sequentdeduction}

The remainder of the proof is \ndd{arith1} and \ndd{arith2}.
To prove them, we first prove a lemma \ndd{arith3}.
$\lsequent{\G,a=acc,0 \leq t \leq T, V(t) \geq 0}{LV(X(t)) \leq V(t) \leq \UV(X(t))}$.

To prove \ndd{arith3} prove the upper bound $V(T) \leq \UV(X(T)),$ then the lower bound $\LV(X(T)) \leq V(T)$.
The upper bound holds by construction since $\acc$ is specifically chosen to remain within $\LV(T)$.
To show the lower bound,  consider two regions which partition the safe envelope.
Let Region 1 be bounded by SLS, LS, and SRS, while Region 2 is bounded by SRS, RS, and LS.
Note that in our strategy, case analysis on Region 1 vs.\ Region 2 is implicit in the comparisons $\min$ and $\max$.
We make this case analysis explicit in our proof for the sake of clarity.
It suffices in Region 1 to show $\acc \geq a_{min}$ and in Region 2 to show $\acc \geq -C$.

First consider an initial state $(X(0),V(0)) \in \text{Region 1}$.
The acceleration can be determined by first determining the distance $\delta{X} = X(T) - X(0)$ traversed in time $T$, not to be confused with the \emph{global minimum} traversed distance $\Delta{x}$.
The greatest elapsed distance is attained at upper-left point $\UL = \LS \land \SRS$ where, by definition of $\BM(T)$ we have
$\delta{X} = (x_{mid} - \BM(T))$.
At this point, the acceleration $\acc = C$ is clearly live.
The next extremum is the upper-right point $\UR = \SLS \land \SRS$.
Because $a_{min}$ defines the curve $\SLS,$ any acceleration $\acc \geq a_{min}$ is live from any point on $\SRS$.
Monotonicity shows that all other points are live because $\delta{X}$ decreases with distance from $\SRS$.
That is, $\acc \leq C \limply \delta{X} \leq (x_{mid} - \BM(T)) = (\AM(T) - x_{mid}),$ so that $X(T) \leq \AM(T)$.
Then $\acc$ increases as $\delta{X}$ decreases.
Since $a_{min}$ is defined to yield $V(T) = \UV(X(T))$ only at $X(T) = \AM(T),$ then by monotonicity also $V(T) \leq \UV(T)$.

From the UL and UR points, correctness of the entire Region 1 follows by additional monotonicity and continuity arguments.
Any point between  $\UL$ and $\UR$ has $\acc \in [a_{min},C]$ because the restriction of $\acc$ to this line is monotone.
As we move toward the lower-left corner (LL), then $\acc$ can only increase: decreasing $V(0)$ or $X(0)$ frees us to be \emph{less} conservative.
Thus $\acc \geq a_{min}$ everywhere in Region 1 as desired.

From every point in Region 2, consider the simplistic braking rate $a_{simp}$ which, if followed indefinitely, achieves $x=g$ exactly when $v=0$.
The trajectory of $a_{simp}$ is clearly within Region 2 for any initial point in Region 2.
The value $a_{simp}$ is always in $[-C, -a_{min}]$ and so is not only physically achievable but is also live.
Thus concludes the proof of \ndd{arith3}.

For \ndd{arith1} we prove
$\lsequent{\G,a=\acc,0 \leq t \leq T, V(t) \geq 0}{\esub{\safe}{(x,v)}{(X(t),V(t))}},$ i.e., we prove
$\lsequent{\G,a=\acc,0 \leq t \leq T, V(t) \geq 0}{X(t) \leq g}$, assuming \ndd{arith3}.
Since $V(T) \leq \UV(X(T)) \leq \RS(X(T))$ then (SB) $\frac{V(T)^2}{2C} \leq (g - X(T))$.
Since the LHS is trivially nonnegative, the RHS is nonnegative, i.e. $X(T) \leq g$ as desired.

We now prove \ndd{arith2}, i.e., we prove
$\G,a=\acc,T/2 \leq t \leq T, V(t) \geq 0 \allowbreak \vdash \esub{\phi}{(x,v)}{(X(t),V(t))}$
where we abbreviate:
$\rho_1 = \LV(x) \leq v \leq \UV(x),
\rho_2 = x \geq 0,
\rho_3 = v \geq 0,
\rho_4 = (\met_0 = (g-x)),
\rho_5 = a=\acc,
\rho_4 = t \in [T/2,T],
\rho_7 = V(t) \geq 0$
and
$\phi_1 = \LV(X(t)) \leq V(t) \leq \LV(X(t)),
\phi_2 = X(t) \geq 0,
\phi_3 = V(t) \geq 0,
\phi_4 = (\met_0 \geq \Delta{x} + g-X(0) \lor g = X(t))$
so that the context
$(\G,a=\acc,0 \leq t \leq T, V(t) \geq 0) \lequiv \bigwedge_{i} \rho_i$
and $\esub{\phi}{(x,v)}{(X(t),V(t))} = \bigwedge_{i} \phi_i$.
We prove each conjunct $\phi_i$.
Conjunct $\phi_1$ is already proven by \ndd{arith3}.
We prove $\phi_3$ next to prove $\phi_2$ as a corollary.
We prove $\phi_4$ last.

In each case except $\phi_4,$ it suffices to consider the case $t = T$ by monotonicity.

We prove $\rho_3$ by hypothesis rule: assumption $\rho_7$ is the desired result.

We prove $\rho_2$.
We can prove it by the ODE solution or even more obviously by \irref{di}.
From the domain constraint, $v \geq 0$ is an invariant.
Then $\rho_2$ says $X(0) \geq 0$ and since $\D{x} = v \geq 0$ then by \irref{di} we have $X(T) \geq 0$.

We prove $\rho_4$.
To do so, we must test whether we are in the ``final'' iteration of the loop.
We perform an inexact (formula-level) comparison of  $g - x$ against  $\frac{a_{min} T^2}{16}$
with tolerance $\frac{a_{min} T^2}{16}$
so that we have constructively
$g - x \leq \frac{a_{min} T^2}{8} \lor g - x \geq \frac{a_{min} T^2}{16}$.
In the former case since $t \geq T/2$ we derive from construction of $a$ and the initial velocity envelope
that  $\frac{V(0)^2}{2a} = (g-x)$ so by definition of $X(t)$ have $X(t) = g$ which satisfies the right disjunct.

In the second case, not only does the test yield $g - X(0) \geq \frac{a_{min} T^2}{16}$, but the test $t \geq T/2$ implies a stronger condition:
\[X(t) - X(0) \geq \frac{a_{min} T^2}{8} \]
which combined with safe braking entails
$g - X(0) \geq \frac{a_{min} T^2}{8}$.
Then
\begin{align*}
           &\met_0 \geq \Delta{x} + (g-X(T))\\
\text{iff}~& g - X(0) \geq \Delta{x} + (g - X(T))\\
\text{iff}~& X(T) \geq \Delta{x} + X(0)
\end{align*}
We argue by monotonicity.
The elapsed distance $\delta{X}$ is minimized (i.e., $\delta{X} = \Delta{x}$) when velocity and acceleration are minimized, that is when $a = a_{min}$ and $x = 0$.
In the worst case Demon chooses $t = T/2,$ and Demon is responsible for satisfying the domain constraint $V(t) \geq 0$ and test $t \geq T/2$ simultaneously.
Then by the Newton equations, $(X(T) - X(0)) \geq \delta{X} = a_{\min} \frac{T^2}{8}$, which exactly the definition of $\Delta{x}$ as required.

This completes the last case of the arithmetic lemma, which in turn closes the final goal of the reach-avoid proof.


%% file: hybrid-proofs.tex
\renewcommand{\d}{\mathrm{d}}
\renewcommand{\mproves}[3]{#1\allowbreak\vdash #2 \allowbreak \mathop{:} #3}
We prove the stated meta-theorems of \CdGL such as monotonicity, soundness, the Existential Property, and the Disjunction Property.

\subsection{Preliminaries and Assumptions}
We first state preliminaries from the literature and assumptions.

\paragraph{Constructive ODEs.}
A difference between our soundness proof and that of \dGL is that we draw on results of constructive analysis rather than classical analysis.
The major results on which we rely have been proven in the literature, but we restate them here because the theorem statements are otherwise difficult to locate.
The main catch in applying these results is that they are proven for time-derivatives, whereas our differentials are spatial.
For this reason, we will prove \rref{lem:app-differential-lemma} equating time and space differentials within the context of an ODE, which justifies applying these existing results.

\begin{theorem}[Constructive Picard-Lindel\"{o}f~\cite{DBLP:conf/mkm/Cruz-FilipeGW04}]
Picard-Lindel\"{o}f has been formalized in Coq.
We restate it from CoRN\footnote{
The statement of Picard-Lindel\"{o}f is in file \texttt{ode/Picard.v} of the CoRN repository.
See \cite{DBLP:conf/mkm/Cruz-FilipeGW04} for the URL and commit numbers on which our statements are based.
}.
The functions and theorems referenced in our proof summary are also from the CoRN repository.
Let $\tau_1$ and $\tau_2$ be metric spaces and let $f_0 : \tau_1 \to \tau_2$ be uniformly continuous on some region $X \subseteq \tau_1$.
Consider the initial-value problem where $f(0,y) = f_0(y)$ and $\der{f}(x,y) = v(x,y)$ and where $v$ is Lipschitz on $X$.

Then there constructively exists a function $f: \tau_1 \to \tau_2$ that solves the initial value problem, i.e.,
\begin{itemize}
\item $f(0,y) = f_0(y)$
\item $\der{f}(x,y) = v(x,y)$
\end{itemize}
\begin{proof}[Summary]
  The proof relies on the existence for each $v$ of the well-known \emph{\textsf{Picard}} operator $\textsf{picard}_v : (\tau_1 \to \tau_2) \to (\tau_1 \to \tau_2)$ and the fact that this operator is contractive.
  When contracted iteratively, the limit is the solution of the ODE $\der{f}(x,y) = v(x,y)$.
  The proof relies on the Banach fixed point operator $\textsf{fp}$ such that $\textsf{fp}\ g\ g_0$ is a fixed point of $g,$ computed as the limit of the sequence $g_{i+1} = g\ g_i$ starting from the given $g_0$.
Specifically, define $g_0(t)(y) = f(0,y)$.
  \begin{enumerate}
  \item By the Banach fixed point theorem, then $\textsf{fp}\ \textsf{picard}_v\ g_0$ is a fixed point such that $\textsf{picard}_v\ (\textsf{fp}\ \textsf{picard}_v\ g_0) = (\textsf{fp}\ \textsf{picard}_v\ g_0)$.
  \item $\textsf{fp}\ \textsf{picard}_v\ g_0$ is a solution of the IVP.
  \item $\textsf{fp}\ \textsf{picard}_v\ g_0$ is constructive: its exact value is arbitrarily approximated by iterating the $\textsf{picard}$ operator. \qedhere
  \end{enumerate}
\end{proof}
\label{thm:app-pl}
\end{theorem}
\begin{theorem}[Constructive Differential Induction (DI) Lemmas]
The following statements are restated from CoRN\footnote{See files \texttt{ftc/CalculusTheorems.v} and \texttt{ftc/Rolle.v} of CoRN~\cite{DBLP:conf/mkm/Cruz-FilipeGW04}.} and are corollaries of the constructive Mean Value Theorem.

Let $a, b : \reals$ with $a < b,$ let $I = [a,b]$.
Let $\veps : \reals > 0$.
Let $F, \der{F}, G, \der{G} : I \to \reals$.
\begin{itemize}
\item
Lemma \emph{(Feq-criterium)} supports equational DI.
If $\der{F} = \der{G}$ on $I$ and there exists $x \in I$ such that $F\ x = G\ x,$ then $F=G$ on $I$.

\item
Lemma \emph{(Derivative-imp-resp-less)} supports strict inequational DI.
If $0 < \der{F}$ on $[a,b]$ then $F\ a < F\ b$.

\item
Lemma \emph{(Derivative-imp-resp-leEq)} supports nonstrict inequational DI.
If $0 \leq \der{F}$ on $I,$ then $F\ a \leq F\ b$.
\end{itemize}
The cases for $>$ and $\geq,$ which are also proven by CoRN, are symmetric.

\label{thm:app-mvt}
\end{theorem}

\begin{theorem}[Constructive DV Lemma]
If  $\der{G} \geq d$ on $[a,b]$ for some constant $d > 0,$ then
$G(b) - G(a) \geq d (b-a)$.
\label{lem:app-dv}
\end{theorem}
\begin{proof}
CoRN features a lemma \emph{(Law-of-the-Mean-Abs-ineq)} which is almost our desired lemma for DV:
If $\der{F} \leq c$ on $[a,b]$ for some constant $c,$ then
 $F(b) - F(a) \leq c (b-a)$.
Assume (0) $\der{G} \geq d$ on $[a,b]$ and (1) $d > 0$.
Because (the full statement of) \emph{(Law-of-the-Mean-Abs-ineq)} supports $c < 0$ and $b < a$ as well, it suffices to
  let $F = -G$ and $c = -d,$ then from (0) and (1) have
(2) $\der{F} \leq c$ on $[a,b]$ so by \emph{(Law-of-the-Mean-Abs-ineq)}
have $F(b) - F(a) \leq c (b-a),$ thus $G(b) - G(a) \geq d(b-a)$ by definition of $F$ and $c$ as desired.
\end{proof}

\paragraph{Static Semantics.}
The proof calculus and soundness proofs rely on standard notions of free variables $\freevars{e}$, bound variables $\boundvars{\alpha},$ and must-bound variables $\mustboundvars{\alpha}$.
The design decision must be made whether to characterize these functions implicitly (semantically) or define them explicitly (syntactically).
For example the semantic free variables of an expression are the smallest set of variables which determine its meaning, while the syntactic free variables are all those which appear in free position.
The semantic free variables are never more than the syntactic free variables, but sometimes are a strict subset.
For a game $\alpha,$ the syntactic bound variables $\boundvars{\alpha},$ are those which are assigned on at least one execution path of $\alpha,$ while the syntactic must-bound variables $\mustboundvars{\alpha}$ are those which are assigned on every execution path of $\alpha$.

Because we leave the language of terms $f,g$ open (any well-typed term from the meta logic is permitted) we have no choice but to characterize free term variables implicitly or assume a correct syntactic computation exists.
For formulas and games we do have a choice: the language of games is closed and easily admits a syntactic definition.
In this work we use a closed formula language, but certainly one might \emph{wish} to use an open formula language;
it could be useful to use arbitrary type families $\tau : \sty \to \alltype$ as the postcondition of a game modality: $\ddiamond{\alpha}{\tau}$.
This is easy for a semantic treatment of variables but not a syntactic one: any new connectives would need new syntactic variable computations.
Yet, a syntactic characterization of free variables is required to show that our proof rules are effective.

The cases for systems are those from~\cite{DBLP:conf/cade/Platzer15}, while the duality cases $\pdual{\alpha}$ are homomorphic~\cite{DBLP:conf/cade/Platzer18}.
However, for comparison, we briefly discuss their semantic counterparts (note the different font) $\sfreevars{e}, \sboundvars{\alpha},$ and $\smustboundvars{\alpha}$ based directly on the coincidence and bound effect properties.
For an expression $e$ (term $f,$ formula $\phi,$ or game $\alpha$), the semantic free variables $\sfreevars{e}$ are those which can influence the meaning of $e$.
In these definitions, $V = s\ V$ is an abbreviation for $\bigwedge_{x \in V} x = s\ x$ where $s\ x$ is a constant equal to the value of $x$ in state $s$.
We write $S^\complement$ for the complement of set $S$.

\begin{align*}
\sfreevars{f} &= \bigcap \{V~|~ \text{for all }s,t : \sty\text{ if } s = t\text{ on } V\text{ then } f\ s = f\ t\}\\
\sfreevars{\phi} &= \bigcap \{V~|~ \text{for all }s,t : \sty\text{ if } s = t\text{ on } V\text{ then } \ftrans{\phi}\ s = \ftrans{\phi}\ t\}\\
\sfreevars{\alpha} &= \left(\bigcap \{V~|~ \text{for all }P : (\sty\to\alltype)\text{ and } s,t : \sty\text{ if } s = t \text{ on } V \text{ then } \atrans{\alpha}\ P\ s = \atrans{\alpha}\ P\ t\}\right)\\
                   &\cap \left(\bigcap \{V~|~ \text{for all }P : (\sty\to\alltype)\text{ and } s,t : \sty\text{ if } s = t \text{ on } V \text{ then } \dtrans{\alpha}\ P\ s = \dtrans{\alpha}\ P\ t\}\right)\\
\sboundvars{\alpha}  &= \big(\left(\bigcup\{V~|~\text{for all }P : (\sty\to\alltype)\text{ and } s : \sty\text{ if } \atrans{\alpha}\ P\ s \text{ then } \atrans{\alpha}\ (P \land (V = s\ V))\ s\right)\\
               &\cap  \left(\bigcup\{V~|~\text{for all }P : (\sty\to\alltype)\text{ and } s : \sty\text{ if } \dtrans{\alpha}\ P\ s \text{ then } \dtrans{\alpha}\ (P \land (V = s\ V))\ s\right)\big)^\complement\\
\end{align*}
The semantic bound variables of games are the complement of the set of preserved variables, which are provably equal to their initial values regardless of the postcondition.
A semantic treatment does not strictly need a definition of must-bound variables, because they are primarily used to provide a less conservative definition of syntactic free variables.
Regardless, a semantic definition can be given:
\begin{align*}
\smustboundvars{\alpha} &= \bigcup\left(\{V~|~\text{for all } s,t : \sty\text{ if } s=t\text{ on } \sfreevars{\alpha}\text{ then }\atrans{\alpha}\ (\lambda ss, \atrans{\alpha}\ (\lambda tt, ss = tt\text{ on }V)\ t)\ s\}\right)\\
                  &\cap  \bigcup\left(\{V~|~\text{for all } s,t : \sty\text{ if } s=t\text{ on } \sfreevars{\alpha}\text{ then }\dtrans{\alpha}\ (\lambda ss, \dtrans{\alpha}\ (\lambda tt, ss = tt\text{ on }V)\ t)\ s\}\right)
\end{align*}

We now recall the syntactic definitions of free, bound, and must-bound variables for the sake of contrast and the sake of being self-contained.

\begin{align*}
  \freevars{f \sim g} &= \freevars{f} \sim \freevars{g}\\
  \freevars{\ddiamond{\alpha}{\phi}} &= \freevars{\alpha} \cup (\freevars{\phi} \setminus \mustboundvars{\alpha})\\
  \freevars{\dbox{\alpha}{\phi}} &= \freevars{\alpha} \cup (\freevars{\phi} \setminus \mustboundvars{\alpha})\\\hline
  \freevars{\ptest{\phi}} &= \freevars{\phi}\\
  \freevars{\humod{x}{f}} &= \freevars{f}\\
  \freevars{\prandom{x}} &= \emptyset\\
  \freevars{\pevolvein{\D{x}=f}{\ivr}} &= \{x\} \cup \freevars{f} \cup \freevars{\ivr}\\
  \freevars{\alpha;\beta} &= \freevars{\alpha} \cup (\freevars{\beta} \setminus \mustboundvars{\alpha})\\
  \freevars{\alpha\cup\beta} &= \freevars{\alpha} \cup \freevars{\beta}\\
  \freevars{\prepeat{\alpha}} &= \freevars{\alpha}\\
  \freevars{\pdual{\alpha}} &= \freevars{\alpha}
\end{align*}

\begin{align*}
  \boundvars{\ptest{\phi}} &= \emptyset\\
  \boundvars{\humod{x}{f}} &= \{x\}\\
  \boundvars{\prandom{x}} &= \{x\}\\
  \boundvars{\pevolvein{\D{x}=f}{\ivr}} &= \{x,\D{x}\}\\
  \boundvars{\alpha;\beta} &= \boundvars{\alpha} \cup \boundvars{\beta}\\
  \boundvars{\alpha\cup\beta} &= \boundvars{\alpha} \cup \boundvars{\beta}\\
  \boundvars{\prepeat{\alpha}} &= \boundvars{\alpha}\\
  \boundvars{\pdual{\alpha}} &= \boundvars{\alpha}
\end{align*}

\begin{align*}
  \mustboundvars{\ptest{\phi}} &= \emptyset\\
  \mustboundvars{\humod{x}{f}} &= \{x\}\\
  \mustboundvars{\prandom{x}} &= \{x\}\\
  \mustboundvars{\pevolvein{\D{x}=f}{\ivr}} &= \{x,\D{x}\}\\
  \mustboundvars{\alpha;\beta} &= \mustboundvars{\alpha} \cup \mustboundvars{\beta}\\
  \mustboundvars{\alpha\cup\beta} &= \mustboundvars{\alpha} \cap \mustboundvars{\beta}\\
  \mustboundvars{\prepeat{\alpha}} &= \emptyset\\
  \mustboundvars{\pdual{\alpha}} &= \mustboundvars{\alpha}
\end{align*}

While we reuse existing definitions of $\freevars{\cdot}, \boundvars{\cdot},$ and $\mustboundvars{\cdot},$ we necessarily offer new proofs of the coincidence and bound effect properties: our semantics are entirely different from those of prior work.
As discussed in the next paragraph, $\freevars{f}$ is not defined here, rather we assume there exists $\freevars{f}$ function which satisfies the coincidence lemma.

\paragraph{Term Language.}
Because our term language reuses terms of the host type theory, we must assume basic lemmas about the term language in order to prove the corresponding lemmas for formulas and games.
The following lemmas should hold in any reasonable type theory.
Coincidence, renaming, and substitution for state variables $s$ are fundamental operations in \emph{any} $\lambda$-calculus, and we simply require that these properties must be generalizable to program variables $x,$ which are simply projections of $s$.

\begin{lemma}[Term coincidence]
  If $s = t$ on $\freevars{f},$ then $f\ s = f\ t$.
\label{lem:app-term-coincide}
\end{lemma}
\begin{proof}[Justification]
Just as definitions of $\freevars{\phi}$ and $\freevars{\alpha}$ have been given in the literature, we are simply assuming that the host theory supports a notion of free variables $\freevars{f}$ for functions over states.
We do not define it ourselves because its exact definition depends on what operations are supported on $\xty$ and $\sty$.
Once $\freevars{f}$ has been defined, coincidence is simply the fundamental correctness theorem for the $\freevars{\cdot}$ function.
\end{proof}

\begin{lemma}[Term renaming]
  $(\eren{f}{x}{y})\ s = f (\eren{s}{x}{y})$
\label{lem:app-term-rename}
\end{lemma}
\begin{proof}[Justification]
We are assuming that variables can be renamed in terms.
This is only a modest generalization of the $\alpha$-renaming rule for variables $s$.
\end{proof}

\begin{lemma}[Term substitution]
  $(\tsub{f}{x}{g})\ s = f (\lset{s}{x}{(g\ s)})$
\label{lem:app-term-subst}
\end{lemma}
\begin{proof}[Justification.]
  Substitution for program variables is a modest generalization of substitution for state variables.
\end{proof}

\paragraph{Notations and abbreviations.}
Some notations are useful for brevity in the proofs which may not have been mentioned in the main paper.

It is sometimes useful to talk of ODE solutions as yielding an entire state rather than the value for one variable.
Thus we abbreviate 
\[Sol(t) = \lset{s}{(x,\D{x})}{(sol\ t, f\ (\lset{s}{x}{(sol\ t)}))}\]
 for a state-valued solution.

\subsection{Proofs of Stated Results}
For semantic proofs about the inhabitance of types, we do not explicitly write out the proof terms which inhabit them, since the proof terms are obvious from our proofs-by-type-rewriting.

\begin{lemma}[Differential Lemma]
Assume (A1) $(\solves{sol}{s}{d}{\D{x}=f})$ and (A2) $t \in [0,d]$ and (A3) $\freevars{g} \subseteq \{x\}$.
Recall that we abbreviate $Sol(t) = \lset{s}{(x,\D{x})}{(sol\ t, f\ (\lset{s}{x}{(sol\ t)}))}$.
Then
$\der{g}\ (Sol\ t) = \frac{\d}{\d r}\left[g\ (Sol\ r)\right](t)$
\label{lem:app-differential-lemma}
\end{lemma}
\begin{proof}[Sketch]
First observe the left-hand-side is the directional derivative of $g$ at $x = sol\ t$ in direction $\D{x} = f\ (\lset{s}{x}{(sol\ t)})$ by the semantics of $\der{g}$.
By assumptions $(\solves{sol}{s}{d}{\D{x}=f})$ and $t \in [0,d]$ then $\lget{Sol(t)}{\D{x}} = \frac{\d x}{\d t}$, i.e., the directional derivative and time derivative agree.
Assumption (A3) is essential: differential variables $\D{y}$ are not bound in $Sol$ for $y \neq x,$ so $g$ must not depend on them.
\end{proof}

In practice, however, (A3) is not a limitation.
Rather, before applying a rule which relies on this differential lemma, one would apply a step that locally transforms any additional game variables $y$ into constants, which ensures their derivatives are 0 as intended, also fulfilling the requirement of this lemma.

\begin{lemma}[Monotonicity]
  Let $P, Q : \sty \to \alltype$.
  Note that in this lemma, $P$ and $Q$ need not be of form $\ftrans{\phi}$.
  If $P\ s \limply Q\ s$ is inhabited for all $s$
  then $\atrans{\alpha}\ P\ s \limply \atrans{\alpha}\ Q\ s$
  and  $\dtrans{\alpha}\ P\ s \limply \dtrans{\alpha}\ Q\ s$  are inhabited for all $s$.
\label{lem:app-monotone}
\end{lemma}
\begin{proof}

In each case, assume (0) $P\ s \limply Q\ s$ for all $s : \sty$.
Then fix some such $s : \sty,$ for which (0) trivially also holds.
Then assume (1) $\atrans{\alpha}\ P\ s$ or $\dtrans{\alpha}\ P\ s$ to show
  $\atrans{\alpha}\ Q\ s$ or $\dtrans{\alpha}\ Q\ s$ accordingly.
We annotate a step with subscript 0 when its justification is fact (0), likewise for other facts.

The Angel and Demon cases are proven by simultaneous induction, of which we list the Angel cases first.

\mycase $\humod{x}{f},$ have
$ \atrans{\humod{x}{f}}\ P\ s
= P\ (\lset{s}{x}{(f\ s)})
\limply_0 Q\ (\lset{s}{x}{(f\ s)})
= \atrans{\humod{x}{f}}\ Q\ s$

\mycase $\prandom{x},$ have
$\atrans{\prandom{x}}\ P\ s
= \sity{v}{\xty}{(P\ (\lset{s}{x}{v}))}$.
Let $v$ such that $(P\ (\lset{s}{x}{v}))$.
Then by (0), $Q\ (\lset{s}{x}{v}),$ and picking the same $v,$
$\sity{v}{\xty}{(Q\ (\lset{s}{x}{v}))}
= \atrans{\prandom{x}}\ Q\ s$.

\mycase $\ptest{\phi},$ have
$ \atrans{\ptest{\phi}}\ P\ s
= \ftrans{\phi}\ s \mathop{\kwprod} P\ s
\limply_0 \ftrans{\phi}\ s \mathop{\kwprod} Q\ s
= \atrans{\ptest{\phi}}\ Q\ s$

\mycase $\pevolvein{\D{x}=f}{\ivr},$ have
\begin{align*}
 &\atrans{\pevolvein{\D{x}=f}{\ivr}}\ P\ s\\
=~&\sity{d}{{\reals_{\geq0}}}{\sity{sol}{[0,d]\to\xty}{}}\\
 &(\solves{sol}{s}{d}{\D{x}=f})\\
 &\mathop{\kwprod} (\pity{t}{{[0,d]}}{\ftrans{\ivr}\ {(\lset{s}{x}{(sol\ t)})}})\\
 &\mathop{\kwprod} P\ (\lset{s}{(x,\D{x})}{(sol\ d, f\ (\lset{s}{x}{(sol\ d)}))})
\end{align*}
Then unpack $d$ and $sol$ such that $(\solves{sol}{s}{d}{\D{x}=f})$
  and $(\pity{t}{{[0,d]}}{P\ {(\lset{s}{x}{(sol\ t)})}})$
so that $P\ (\lset{s}{(x,\D{x})}{(sol\ d, f\ (\lset{s}{x}{(sol\ d)}))})$ and by (1) have
\[Q\ (\lset{s}{(x,\D{x})}{(sol\ d, f\ (\lset{s}{x}{(sol\ d)}))})\]
so
\begin{align*}
 &\sity{d}{{\reals_{\geq0}}}{\sity{sol}{[0,d]\to\xty}{}}\\
 &(\solves{sol}{s}{d}{\D{x}=f})\\
 &\mathop{\kwprod} (\pity{t}{{[0,d]}}{\ftrans{\ivr}\ (\lset{s}{x}{(sol\ t)})})\\
 &\mathop{\kwprod} Q\ (\lset{s}{(x,\D{x})}{(sol\ d, f\ (\lset{s}{x}{(sol\ d)}))})\\
=~&\atrans{\pevolvein{\D{x}=f}{\ivr}}\ Q\ s
\end{align*}

\mycase $\alpha;\beta,$ have
$ \atrans{\alpha;\beta}\ P\ s
= \atrans{\alpha}\ (\atrans{\beta}\ P)\ s$ (2).
Note by IH on $\beta$ that (3) for all $s,$ have $(\atrans{\beta}\ P\ s) \limply (\atrans{\beta}\ Q\ s)$.
Then (2) and (3) suffice to apply the IH on $\alpha,$ giving
$\atrans{\alpha}\ (\atrans{\beta}\ Q)\ s
= \atrans{\alpha;\beta}\ Q\ s$

\mycase $\alpha\cup\beta,$ have
$ \atrans{\alpha \cup \beta}\ P\ s
= \atrans{\alpha}\ P\ s \mathop{\kwsum} \atrans{\beta}\ P\ s
= \atrans{\alpha}\ Q\ s \mathop{\kwsum} \atrans{\beta}\ Q\ s
= \atrans{\alpha \cup \beta}\ Q\ s$

\mycase $\prepeat{\alpha},$ have
\begin{align*}
  &\atrans{\prepeat{\alpha}}\ P\ s\\
= &(\lindty{\tau'\mathrel{:}(\sty \to \alltype)}{} \lambda{t:\sty}\,\\
  & (P\ t \to \tau'\ t)\\
& \quad \mathop{\kwsum}  (\atrans{\alpha}\ \tau'\ t \to \tau'\ t))\ s
\end{align*}
Then note for all $t$ have $P\ t \limply Q\ t$ so that
\begin{align*}
 & (\lindty{\tau'\mathrel{:}(\sty \to \alltype)}{} \lambda{t:\sty}\,(P\ t \to \tau'\ t)\\
& \quad\mathop{\kwsum} (\atrans{\alpha}\ \tau'\ t \to \tau'\ t))\ s\\
\limply &(\lindty{\tau'\mathrel{:}(\sty \to \alltype)}{} \lambda{t:\sty}\,(Q\ t \to \tau'\ t)\\
& \quad\mathop{\kwsum} (\atrans{\alpha}\ \tau'\ t \to \tau'\ t))\ s\\
= &\atrans{\prepeat{\alpha}}\ Q\ s
\end{align*}

\mycase $\pdual{\alpha},$ have
$ \atrans{\pdual{\alpha}}\ P\ s
= \dtrans{\alpha}\ P\ s
\limply_{IH} \dtrans{\alpha}\ Q\ s
= \atrans{\pdual{\alpha}}\ Q\ s
$ where the step marked IH employs the IH from the simultaneous IH on Demonic games, which applies because $\alpha$ is structurally smaller than $\pdual{\alpha}$.

We give the Demon cases.

\mycase $\humod{x}{f},$ have
$ \dtrans{\humod{x}{f}}\ P\ s
= P\ (\lset{s}{x}{(f\ s)})
\limply_0 Q\ (\lset{s}{x}{(f\ s)})
= \dtrans{\humod{x}{f}}\ Q\ s$

\mycase $\prandom{x},$ have
$\dtrans{\prandom{x}}\ P\ s
= \pity{v}{\xty}{(P\ (\lset{s}{x}{v}))},$ so
(2) $(P\ (\lset{s}{x}{v}))$ for all $v:\xty$.
Apply (2) for arbitrary $v$ so by (1),
(2) $(Q\ (\lset{s}{x}{v}))$ for all $v:\xty$, thus
 $\pity{v}{\xty}{(Q\ (\lset{s}{x}{v}))}
= \dtrans{\prandom{x}}\ Q\ s$.

\mycase $\ptest{\rho},$ have
$ \dtrans{\ptest{\rho}}\ P\ s
= (\ftrans{\rho}\ s \limply P\ s)$ (2).
Assume $\ftrans{\rho}\ s$ so by (2), $P\ s$ and by (1), $Q\ s$,
thus
$(\ftrans{\rho}\ s \limply Q\ s)
= \dtrans{\ptest{\rho}}\ Q\ s$.

\mycase $\pevolvein{\D{x}=f}{\rho},$ have
\begin{align*}
 &\dtrans{\pevolvein{\D{x}=f}{\rho}}\ P\ s\\
=&\pity{d}{{\reals_{\geq0}}}{\sity{sol}{[0,d]\to\xty}{}}\\
 &(\solves{sol}{s}{d}{\D{x}=f})\\
 &\to (\pity{t}{{[0,d]}}{\ftrans{\rho}\ (\lset{s}{x}{(sol\ t)})})\\
 &\to P\ (\lset{s}{(x,\D{x})}{(sol\ d, f\ (\lset{s}{x}{(sol\ d)}))})
\end{align*}
Then for arbitrary  $d$ and $sol$ and assume $(\solves{sol}{s}{d}{\D{x}=f})$
  and $(\pity{t}{{[0,d]}}{P\ {(\lset{s}{x}{(sol\ t)})}}),$
so that $P\ (\lset{s}{(x,\D{x})}{(sol\ d, f\ (\lset{s}{x}{(sol\ d)}))})$ and by (0) have
\[Q\ (\lset{s}{(x,\D{x})}{(sol\ d, f\ (\lset{s}{x}{(sol\ d)}))})\] giving
\begin{align*}
&\pity{d}{{\reals_{\geq0}}}{\sity{sol}{[0,d]\to\xty}{}}\\
&(\solves{sol}{s}{d}{\D{x}=f})\\
&\to (\pity{t}{{[0,d]}}{\ftrans{\rho}\ (\lset{s}{x}{(sol\ t)})})\\
&\to Q\ (\lset{s}{(x,\D{x})}{(sol\ d, f\ (\lset{s}{x}{(sol\ d)}))})\\
= &\dtrans{\pevolvein{\D{x}=f}{\rho}}\ Q\ s
\end{align*}

\mycase $\alpha;\beta,$ have
$ \dtrans{\alpha;\beta}\ \ftrans{P\ s}
= \dtrans{\alpha}\ (\dtrans{\beta}\ P)\ s,$ call this fact (2).
Note by IH on $\beta$ that (3) for all $s,$ have $(\dtrans{\beta}\ P\ s) \limply (\dtrans{\beta}\ Q\ s)$.
Then (2) and (3) suffice to apply the IH on $\alpha,$ giving
$\dtrans{\alpha}\ (\dtrans{\beta}\ Q)\ s
= \dtrans{\alpha;\beta}\ Q\ s$

\mycase $\alpha\cup\beta,$ have
$ \dtrans{\alpha \cup \beta}\ P\ s
= \dtrans{\alpha}\ P\ s  \mathop{\kwprod} \dtrans{\beta}\ P\ s
\iheq \dtrans{\alpha}\ Q\ s  \mathop{\kwprod} \dtrans{\beta}\ Q\ s
= \dtrans{\alpha \cup \beta}\ Q\ s$

\mycase $\prepeat{\alpha},$ have
\begin{align*}
  &\dtrans{\prepeat{\alpha}}\ P\ s\\
= & (\lcoty{\tau'\mathrel{:}(\sty \to \alltype)}{}\lambda{t:\sty}\, (\tau'\ t \to \dtrans{\alpha}\ \tau'\ t)\\
&\quad \mathop{\kwprod} (\tau'\ t \to P\ t))\ s
\end{align*}
then since for all $t$ have $P\ t \limply Q\ t$ then have
\begin{align*}
  &(\lcoty{\tau'\mathrel{:}(\sty \to \alltype)}{} \lambda{s:'\sty}\,\\
  & \quad(\tau'\ t \to \dtrans{\alpha}\ \tau'\ t)\\
  & \quad\mathop{\kwprod} (\tau'\ t \to P\ t))\ s\\
\limply &(\lcoty{\tau'\mathrel{:}(\sty \to \alltype)}{} \lambda{t:\sty}\,\\
 & \quad(\tau'\ t \to \dtrans{\alpha}\ \tau'\ t)\\
 &\quad \mathop{\kwprod} (\tau'\ t \to Q\ t))\ s\\
= &\dtrans{\prepeat{\alpha}}\ Q\ s
\end{align*}

\mycase $\pdual{\alpha},$ have
$ \dtrans{\pdual{\alpha}}\ P\ s
= \atrans{\alpha}\ P\ s
\limply_{IH} \atrans{\alpha}\ Q\ s
= \dtrans{\pdual{\alpha}}\ Q\ s$
where the IH is from the simultaneous induction on Angelic games.
\end{proof}

The static semantics results are stated informally in the main paper for the sake of brevity.
We give full formal statements here with proof.
The coincidence lemmas for formulas, Angelic games, and Demonic games are proven by simultaneous induction.
We also include results for contexts which are simply finite conjunctions of formulas, and write $\ftrans{\G}$ to mean the product of $\ftrans{\phi}$ for $\phi \in \G$.
The same holds of the renaming and substitution lemmas.

Note that we state coincidence and bound effect for games differently from prior work~\cite{DBLP:conf/cade/Platzer18} simply to avoid introducing some extra notations used in prior work.

\begin{lemma}[Formula coincidence]
If $s=t$ on $\freevars{\G} \cup \freevars{\phi}$ then given $M$ such that $\mproves{\ftrans{\G}\ s}{M}{(\ftrans{\phi}\ s)}$
there exists $N$ such that $\mproves{\ftrans{\G}\ t}{N}{(\ftrans{\phi}\ t)}$.
Coincidence for contexts also holds: If $s=t$ on $\freevars{\G}$ and $\ftrans{\G}\ s$ is inhabited then $\ftrans{\G}\ t$ is inhabited.
Coincidence for the construct $(\solves{sol}{s}{d}{\D{x}=f})$ also holds:
If $s=t$ on $\freevars{f} \cup \{x\}$ then $(\solves{sol}{s}{d}{\D{x}=f}) = (\solves{sol}{t}{d}{\D{x}=f})$
\label{lem:app-formula-coincide}
\end{lemma}
\begin{proof}
The formula and context cases are proven by simultaneous induction with one another and with \rref{lem:app-game-coincide}.

\mycase $\Gemp$ holds trivially because the unit tuple witnesses $\ftrans{\Gemp}\ t$ for all $t$.

\mycase $\G,\psi:$
Then (A0) $\ftrans{\G,\psi}\ s = \ftrans{\G}\ s \mathop{\kwprod} \ftrans{\psi}\ s$.
Note $\freevars{(\G,\psi)} = \freevars{\G} \cup \freevars{\psi}$ as required to apply the IH.
Apply the IH on smaller context $\G$ to get $\ftrans{\G}\ s \limply \ftrans{\G}\ t$ and apply the formula IH to get (1) $\ftrans{\psi}\ s \limply \ftrans{\psi}\ t$.
Then from (1) the right conjunct of (A0) have $\ftrans{\psi}\ t,$ then with the left conjunct of (A0) have $\ftrans{\G}\ t \mathop{\kwprod} \ftrans{\psi}\ t = \ftrans{(\G,\psi)}$ as desired.

We note a simplification for the formula cases:
In each case the proof of the conclusion begins by assuming
(A2) $\ftrans{\G}\ t$ is inhabited, which by the context cases \rref{lem:app-formula-coincide}
and because $V \supseteq \freevars{\G}$ gives an inhabitant of (A4) $\ftrans{\G}\ s$.
Then by modus ponens on (A1) gives
(A3) $\ftrans{\phi}\ s$ in each case.
In short, in every case we are entitled to simply consider formulas as in (A3) rather than sequents (A1).

\mycase $\ddiamond{\alpha}{\phi},$ have
$\ftrans{\ddiamond{\alpha}{\phi}}\ s
= \atrans{\alpha}\ \ftrans{\phi}\ s
\iheq \atrans{\alpha}\ \ftrans{\phi}\ t
= \ftrans{\ddiamond{\alpha}{\phi}}\ t$
where the IH applies because $\freevars{\ddiamond{\alpha}{\phi}} = \freevars{\alpha} \cup (\freevars{\phi} \setminus \mustboundvars{\alpha})$.

\mycase $\dbox{\alpha}{\phi},$ have
$\ftrans{\dbox{\alpha}{\phi}}\ s
= \dtrans{\alpha}\ \ftrans{\phi}\ s
\iheq \dtrans{\alpha}\ \ftrans{\phi}\ t
= \ftrans{\dbox{\alpha}{\phi}}\ t$
where the IH applies because $\freevars{\dbox{\alpha}{\phi}} = \freevars{\alpha} \cup (\freevars{\phi} \setminus \mustboundvars{\alpha})$.

\mycase $f \sim g:$
Because $\freevars{f \sim g} = \freevars{f} \cup \freevars{g}$ then by \rref{lem:app-term-coincide} have $f\ s = f\ t$ and $g\ s = g\ t$ from which we derive
$\ftrans{f \sim g}\ s = (f\ s \sim g\ s) = (f\ t \sim g\ t) = \ftrans{f \sim g}\ t$.

\mycase $(\solves{sol}{s}{d}{\D{x}=f}):$
Note (0) $(\lget{s}{x} = sol\ 0)$ iff $(\lget{t}{x} = sol\ 0)$ since $s = t$ on $\{x\}$.
Also, $s = t$ on $\freevars{f} \setminus \{x\}$ so by \rref{lem:app-term-coincide} have
(1) $(\der{y}\ r = f\ (\lset{s}{x}{(sol\ r)}))$ iff $(\der{sol}\ r = f\ (\lset{t}{x}{(sol\ r)}))$ for all $r \in [0,d]$.
Next, by \rref{lem:app-term-coincide} have (2) $f\ (\lset{s}{x}{(sol\ r)}) = f\ (\lset{t}{x}{(sol\ r)})$ since $s = t$ on $\freevars{f}$.

Then by (0), (1), and (2) have
$ (\solves{sol}{s}{d}{\D{x}=f})
= (\sprod{(\lget{s}{x} = sol\ 0)}{(\pity{r}{[0,d]}{(\der{sol}\ r = f\ (\lset{s}{x}{(sol\ r)}))})})
= (\sprod{(\lget{t}{x} = sol\ 0)}{(\pity{r}{[0,d]}{(\der{sol}\ r = f\ (\lset{t}{x}{(sol\ r)}))})})
= (\solves{sol}{t}{d}{\D{x}=f})
$
\end{proof}

\begin{lemma}[Game coincidence]
If $s=t$ on $V \supseteq \freevars{\G} \cup \freevars{\alpha} \cup (\freevars{\phi} \setminus \mustboundvars{\alpha})$ then:
\begin{itemize}
\item  given $M$ such that $\mproves{\ftrans{\G}\ s}{M}{(\atrans{\alpha}\ \ftrans{\phi}\ s)}$
there exists $N$ such that $\mproves{\ftrans{\G}\ t}{N}{(\atrans{\alpha}\ \ftrans{\phi}\ t)}$
\item  given $M$ such that $\mproves{\ftrans{\G}\ s}{M}{(\dtrans{\alpha}\ \ftrans{\phi}\ s)}$
there exists $N$ such that $\mproves{\ftrans{\G}\ t}{N}{(\dtrans{\alpha}\ \ftrans{\phi}\ t)}$
\end{itemize}
\label{lem:app-game-coincide}
\end{lemma}
\begin{proof}
Proven by induction simultaneously with \rref{lem:app-formula-coincide}.

In each case assume (A0) $s=t$ on
$V \supseteq \freevars{\G} \cup \freevars{\alpha} \cup (\freevars{\phi} \setminus \mustboundvars{\alpha})$
and (A1) $\mproves{\ftrans{\G}\ s}{M}{(\atrans{\alpha}\ \ftrans{\phi}\ s)}$ or
    $\mproves{\ftrans{\G}\ s}{M}{(\dtrans{\alpha}\ \ftrans{\phi}\ s)}$ as appropriate.
We note a simplification:
In each case the proof of the conclusion begins by assuming
(A2) $\ftrans{\G}\ t$ is inhabited, which by the context case of \rref{lem:app-formula-coincide}
and because $V \supseteq \freevars{\G}$ gives an inhabitant of (A4) $\ftrans{\G}\ s$.
Then by modus ponens on (A1) gives
(A3) $\atrans{\alpha}\ \ftrans{\phi}\ s$ or $\dtrans{\alpha}\ \ftrans{\phi}\ s$ in each case.
In short, in every case we are entitled to simply consider formulas as in (A3) rather than sequents as in(A1).

We give the Angel cases.

\mycase $\humod{x}{f}$:
Since $\freevars{\humod{x}{f}} = \freevars{f}$ note by \rref{lem:app-term-coincide} that (0) $f\ s = f\ t$.
Note that $\mustboundvars{\humod{x}{f}} = \{x\}$ so (1) $s = t$ on $V \supseteq \freevars{\phi} \setminus \{x\}$.
Then
$ \atrans{\humod{x}{f}}\ \ftrans{\phi}\ s
= \ftrans{\phi} (\lset{s}{x}{(f\ s)})
=_0 \ftrans{\phi} (\lset{s}{x}{(f\ t)})
=_1 \ftrans{\phi} (\lset{t}{x}{(f\ t)})
= \atrans{\humod{x}{f}}\ \ftrans{\phi}\ t$

\mycase $\prandom{x}$:
Note that $\mustboundvars{\prandom{x}} = \{x\}$ so (0) $s = t$ on $V \supseteq \freevars{\phi} \setminus \{x\}$.
Then
$ \atrans{\prandom{x}}\ \ftrans{\phi}\ s
= (\sity{v}{\xty}{\ftrans{\phi}\ (\lset{s}{x}{v})})
=_0 (\sity{v}{\xty}{\ftrans{\phi}\ (\lset{t}{x}{v})})
= \atrans{\prandom{x}}\ \ftrans{\phi}\ t
$

\mycase $\ptest{\psi},$ have
$ \atrans{\ptest{\psi}}\ \ftrans{\phi}\ s
= \ftrans{\psi}\ s \mathop{\kwprod} \ftrans{\phi}\ s
= \ftrans{\psi}\ t \mathop{\kwprod} \ftrans{\phi}\ t
= \atrans{\ptest{\psi}}\ \ftrans{\phi}\ t
$

\mycase $\pevolvein{\D{x}=f}{\rho}$:
By the case for $(\solves{sol}{s}{d}{\D{x}=f})$ and

have (0) $(\solves{sol}{s}{d}{\D{x}=f}) = (\solves{sol}{t}{d}{\D{x}=f})$.
Note $s = t$ on $\freevars{\rho} \setminus \{x\}$ since
$\freevars{\rho} \subseteq \freevars{\pevolvein{\D{x}=f}{\rho}}$ and
$\{x,\D{x}\} = \mustboundvars{\pevolvein{\D{x}=f}{\rho}}$ and $\D{x} \notin \freevars{\rho}$ by syntactic constraints, thus the IH
on $\rho$ applies to give (1) $(\pity{r}{{[0,d]}}{\ftrans{\rho}\ {(\lset{s}{x}{(sol\ r)})}}) = (\pity{r}{{[0,d]}}{\ftrans{\rho}\ {(\lset{t}{x}{(sol\ r)})}})$.
Likewise $s = t$ on $\freevars{\phi} \setminus \{x,\D{x}\}$ so the IH on $\phi$ gives 
(2) $\ftrans{\phi}\ (\lset{s}{(x,\D{x})}{(sol\ d, f\ (\lset{s}{x}{(sol\ d)}))})
 = \ftrans{\phi}\ (\lset{t}{(x,\D{x})}{(sol\ d, f\ (\lset{t}{x}{(sol\ d)}))})$.

Then applying (1), (2), and (3) we have

\begin{align*}
  &\atrans{\pevolvein{\D{x}=f}{\rho}}\ \ftrans{\phi}\ s\\
= &\sity{d}{{\reals_{\geq0}}}{\sity{sol}{[0,d]\to\xty}{}}\\
  &(\solves{sol}{s}{d}{\D{x}=f})\\
  &\mathop{\kwprod} (\pity{r}{{[0,d]}}{\ftrans{\rho}\ (\lset{s}{x}{(sol\ r)})})\\
  &\mathop{\kwprod} (\ftrans{\phi}\ (\lset{s}{(x,\D{x})}{(sol\ d, f\ (\lset{s}{x}{(sol\ d)}))}))\\
= &\sity{d}{{\reals_{\geq0}}}{\sity{sol}{[0,d]\to\xty}{}}\\
  &(\solves{sol}{t}{d}{\D{x}=f})\\
  &\mathop{\kwprod} (\pity{r}{{[0,d]}}{\ftrans{\rho}\ (\lset{t}{x}{(sol\ r)})})\\
  &\mathop{\kwprod} (\ftrans{\phi}\ (\lset{t}{(x,\D{x})}{(sol\ d, f\ (\lset{s}{x}{(sol\ d)}))}))\\
= &\atrans{\pevolvein{\D{x}=f}{\rho}}\ \ftrans{\phi}\ t
\end{align*}
 as desired.

\mycase $\alpha;\beta$:
Note that $
  \freevars{\alpha;\beta} \cup (\freevars{\phi} \setminus \mustboundvars{\alpha;\beta})
= \freevars{\alpha}
  \cup (\freevars{\beta} \setminus \mustboundvars{\alpha})
  \cup (\freevars{\phi} \setminus \mustboundvars{\alpha;\beta})
= \freevars{\alpha}
  \cup ((\freevars{\beta} \cup (\freevars{\phi} \setminus \mustboundvars{\beta}))
  \setminus \mustboundvars{\alpha})
= \freevars{\alpha} \cup (\freevars{\ddiamond{\beta}{\phi}} \setminus \mustboundvars{\alpha})$
as required in the IH application.
Then
$ \atrans{\alpha;\beta}\ \ftrans{\phi}\ s
= \atrans{\alpha}\ (\atrans{\beta}\ \ftrans{\phi})\ s
= \atrans{\alpha}\ \ftrans{\ddiamond{\beta}{\phi}}\ s
\iheq \atrans{\alpha}\ \ftrans{\ddiamond{\beta}{\phi}}\ t
= \atrans{\alpha}\ (\atrans{\beta}\ \ftrans{\phi})\ t
= \atrans{\alpha;\beta}\ \ftrans{\phi}\ t
$

\mycase $\alpha\cup\beta$:
$ \atrans{\alpha\cup\beta}\ \ftrans{\phi}\ s
= \atrans{\alpha}\ \ftrans{\phi}\ s \mathop{\kwsum} \atrans{\beta}\ \ftrans{\phi}\ s
= \atrans{\alpha}\ \ftrans{\phi}\ t \mathop{\kwsum} \atrans{\beta}\ \ftrans{\phi}\ t
= \atrans{\alpha\cup\beta}\ \ftrans{\phi}\ t$

\mycase $\prepeat{\alpha}$:
Note that $\freevars{\alpha} \cup \freevars{\phi}$ are the free variables of the fixed point
$(\lindty{\tau'\mathrel{:}(\sty \to \alltype)}{\lambda{t:\sty}\,{
(\phi\ t \to \tau'\ t)
\mathop{\kwprod}
(\atrans{\alpha}\ \tau'\ t \to \tau'\ t)
}})$.
Note $s = t$ on $\freevars{\alpha} \cup \freevars{\phi}$ since $\freevars{\prepeat{\alpha}} = \freevars{\alpha}$ and $\mustboundvars{\prepeat{\alpha}} = \emptyset$.
This suffices to ensure
$ \atrans{\prepeat{\alpha}}\ \ftrans{\phi}\ s
= \atrans{\prepeat{\alpha}}\ \ftrans{\phi}\ t$.
Formally the proof follows by induction on the proof that $s$ belongs to the fixed point.

\mycase $\pdual{\alpha},$ have
$ \atrans{\pdual{\alpha}}\ \ftrans{\phi}\ s
= \dtrans{\alpha}\ \ftrans{\phi}\ s
= \dtrans{\alpha}\ \ftrans{\phi}\ t
= \atrans{\pdual{\alpha}}\ \ftrans{\phi}\ t
$

We give the Demon cases.

\mycase $\humod{x}{f},$
since $\freevars{\humod{x}{f}} = \freevars{f}$ note by \rref{lem:app-term-coincide} that (0) $f\ s = f\ t$.
Note that $\mustboundvars{\humod{x}{f}} = \{x\}$ so (1) $s = t$ on $V \supseteq \freevars{\phi} \setminus \{x\}$.
Then
$ \dtrans{\humod{x}{f}}\ \ftrans{\phi}\ s
= \ftrans{\phi} (\lset{s}{x}{(f\ s)})
=_0 \ftrans{\phi} (\lset{s}{x}{(f\ t)})
=_1 \ftrans{\phi} (\lset{t}{x}{(f\ t)})
= \dtrans{\humod{x}{f}}\ \ftrans{\phi}\ t$

\mycase $\prandom{x}$:
Note that $\mustboundvars{\prandom{x}} = \{x\}$ so (0) $s = t$ on $V \supseteq \freevars{\phi} \setminus \{x\}$.
Then
$ \dtrans{\prandom{x}}\ \ftrans{\phi}\ s
=   (\pity{v}{\xty}{\ftrans{\phi}\ (\lset{s}{x}{v})})
=_0 (\pity{v}{\xty}{\ftrans{\phi}\ (\lset{t}{x}{v})})
= \dtrans{\prandom{x}}\ \ftrans{\phi}\ t
$

\mycase $\ptest{\phi},$ have
$ \dtrans{\ptest{\psi}}\ \ftrans{\phi}\ s
= (\ftrans{\psi}\ s \limply \ftrans{\phi}\ s)
\iheq (\ftrans{\psi}\ t \limply \ftrans{\phi}\ t)
= \dtrans{\ptest{\psi}}\ \ftrans{\phi}\ t
$

\mycase $\pevolvein{\D{x}=f}{\rho}$:
By the case for $(\solves{sol}{s}{d}{\D{x}=f})$ have (0) $(\solves{sol}{s}{d}{\D{x}=f}) = (\solves{sol}{t}{d}{\D{x}=f})$.
Note $s = t$ on $\freevars{\rho} \setminus \{x\}$ since $\freevars{\rho} \subseteq \freevars{\pevolvein{\D{x}=f}{\rho}}$ and $\{x,\D{x}\} = \mustboundvars{\pevolvein{\D{x}=f}{\rho}}$ and $\D{x} \notin \freevars{\rho}$ by syntactic constraints, thus the IH
on $\rho$ applies to give (1) $(\pity{r}{{[0,d]}}{\ftrans{\rho}\ {(\lset{s}{x}{(sol\ r)})}}) = (\pity{r}{{[0,d]}}{\ftrans{\rho}\ {(\lset{t}{x}{(sol\ r)})}})$.
Likewise $s = t$ on $\freevars{\phi} \setminus \{x,\D{x}\}$ so the IH on $\phi$ gives (2) 
\begin{align*}
  &\ftrans{\phi}\ (\lset{s}{(x,\D{x})}{(sol\ d, f\ (\lset{s}{x}{(sol\ d)}))})\\
= &\ftrans{\phi}\ (\lset{t}{(x,\D{x})}{(sol\ d, f\ (\lset{t}{x}{(sol\ d)}))})
\end{align*}

Then applying (1), (2), and (3) we have

\begin{align*}
 &\dtrans{\pevolvein{\D{x}=f}{\rho}}\ \ftrans{\phi}\ s\\
= &\pity{d}{{\reals_{\geq0}}}{\sity{sol}{[0,d]\to\xty}{}}\\
 &(\solves{sol}{s}{d}{\D{x}=f})\\
 &\to (\pity{r}{{[0,d]}}{\ftrans{\rho}\ (\lset{s}{x}{(sol\ r)})})\\
 &\to \ftrans{\phi}\ (\lset{s}{(x,\D{x})}{(sol\ d, f\ (\lset{s}{x}{(sol\ d)}))})\\
= &\pity{d}{{\reals_{\geq0}}}{\sity{sol}{[0,d]\to\xty}{}}\\
 &(\solves{sol}{t}{d}{\D{x}=f})\\
 &\to (\pity{sol}{{[0,d]}}{\ftrans{\rho}\ (\lset{t}{x}{(sol\ r)})})\\
 &\to (\ftrans{\phi}\ (\lset{t}{(x,\D{x})}{(sol\ d, f\ (\lset{s}{x}{(sol\ d)}))}))\\
= &\ftrans{\phi}\ (\lset{t}{(x,\D{x})}{(sol\ d, f\ (\lset{t}{x}{(sol\ d)}))})
\end{align*}

\mycase $\alpha;\beta$
Note that $
  \freevars{\alpha;\beta} \cup (\freevars{\phi} \setminus \mustboundvars{\alpha;\beta})
= \freevars{\alpha}
  \cup (\freevars{\beta} \setminus \mustboundvars{\alpha})
  \cup (\freevars{\phi} \setminus \mustboundvars{\alpha;\beta})
= \freevars{\alpha}
  \cup ((\freevars{\beta} \cup (\freevars{\phi} \setminus \mustboundvars{\beta}))
  \setminus \mustboundvars{\alpha})
= \freevars{\alpha} \cup (\freevars{\dbox{\beta}{\phi}} \setminus \mustboundvars{\alpha})$
as required in the IH application.
Then
$ \dtrans{\alpha;\beta}\ \ftrans{\phi}\ s
= \dtrans{\alpha}\ (\atrans{\beta}\ \ftrans{\phi})\ s
= \dtrans{\alpha}\ \ftrans{\dbox{\beta}{\phi}}\ s
\iheq \dtrans{\alpha}\ \ftrans{\dbox{\beta}{\phi}}\ t
= \dtrans{\alpha}\ (\atrans{\beta}\ \ftrans{\phi})\ t
= \dtrans{\alpha;\beta}\ \ftrans{\phi}\ t
$

\mycase $\alpha\cup\beta,$ have
$ \dtrans{\alpha\cup\beta}\ \ftrans{\phi}\ s
= \dtrans{\alpha}\ \ftrans{\phi}\ s \mathop{\kwprod} \dtrans{\beta}\ \ftrans{\phi}\ s
= \dtrans{\alpha}\ \ftrans{\phi}\ t \mathop{\kwprod} \dtrans{\beta}\ \ftrans{\phi}\ t
= \dtrans{\alpha\cup\beta}\ \ftrans{\phi}\ t$

\mycase $\prepeat{\alpha}$:
Note that $\freevars{\alpha} \cup \freevars{\phi}$ are the free variables of the fixed point
$(\lcoty{\tau'\mathrel{:}(\sty \to \alltype)}{\lambda{t:\sty}\,{
(\tau'\ t \to \dtrans{\alpha}\ \tau'\ t)
\mathop{\kwprod}
(\tau'\ t \to \ftrans{\phi}\ t)
}})$.
Note $s = t$ on $\freevars{\alpha} \cup \freevars{\phi}$ since $\freevars{\prepeat{\alpha}} = \freevars{\alpha}$ and $\mustboundvars{\prepeat{\alpha}} = \emptyset$.
This suffices to ensure
$ \dtrans{\prepeat{\alpha}}\ \ftrans{\phi}\ s
= \dtrans{\prepeat{\alpha}}\ \ftrans{\phi}\ t$.
Formally the proof follows by coinduction on the proof that $s$ belongs to the fixed point.

\mycase $\pdual{\alpha},$ have
$ \dtrans{\pdual{\alpha}}\ \ftrans{\phi}\ s
= \atrans{\alpha}\ \ftrans{\phi}\ s
= \atrans{\alpha}\ \ftrans{\phi}\ t
= \dtrans{\pdual{\alpha}}\ \ftrans{\phi}\ t$

\end{proof}

\begin{lemma}[Bound effect]
Let $V$ be a finite set of program variables. 
Let the defined \CdGL formula $V = s\ V$ be shorthand for $\bigwedge_{x\in V}{(\pity{t}{\sty}{(\lget{t}{x} = \lget{s}{x})})}$ meaning that the values of every $x \in V$ at the present state match those of state $s$.
Formally, every $s\ V$ in this formula is a real-valued literal.
Recall that $S^\complement$ denotes the complement of set $S$.

If $V \subseteq \boundvars{\alpha}^\complement$ then
\begin{itemize}
\item
there exists $M$ such that $\mproves{\ftrans{\G}\ s}{M}{\atrans{\alpha}\ \ftrans{\phi}\ s}$ iff
there exists $N$ such that $\mproves{\ftrans{\G}\ s}{M}{\atrans{\alpha}\ \ftrans{\phi \land V = s\ V}\ s}$.
\item
there exists $M$ such that $\mproves{\ftrans{\G}\ s}{M}{\dtrans{\alpha}\ \ftrans{\phi}\ s}$ iff
there exists $N$ such that $\mproves{\ftrans{\G}\ s}{M}{\dtrans{\alpha}\ \ftrans{\phi \land V = s\ V}\ s}$.
\end{itemize}
\label{lem:app-bound-effect}
\end{lemma}
\begin{proof}
By the same argument as in the coincidence lemma, sequent-style bound effect follows trivially from formula-style bound effect.
We first give a uniform argument for the converse direction, then prove the forward direction.
The forward direction performs an outer induction on the size of $V,$ then an inner simultaneous induction on Angelic and Demonic games.

Converse direction:
By left projection, $\ftrans{\phi \land V = s\ V}\ s$ implies $\ftrans{\phi}\ s$ for all $s$.
Then by \rref{lem:app-monotone} have
that $\atrans{\alpha}\ \ftrans{\phi \land V = s\ V}\ s$
implies $\atrans{\alpha}\ \ftrans{\phi}\ s$ and likewise for $\dtrans{\alpha}$.

Forward direction:
By induction on $\abs{V}$, generalizing $\phi$.
The case $\abs{V} = 0$ is trivial since $\atrans{\alpha}\ \ftrans{\phi \land (V = s\ V)} = \atrans{\alpha}\ \ftrans{\phi \land \btt} = \atrans{\alpha}\ \ftrans{\phi}$.
In the case $\abs{V \cup \{x\}} = k + 1$ then apply the IH to $\ftrans{\phi \land x = s\ x}$ and set $V$ to get
that $\atrans{\alpha}\ \ftrans{\phi \land x = s\ x}\ s$ iff $\atrans{\alpha}\ \ftrans{(\phi \land x = s\ x) \land (V = s\ V)} \ s$.
Since we also have that $(\ftrans{(\phi \land x = s\ x) \land (V = s\ V)} \ s
= (\ftrans{\phi \land ((V\cup \{x\}) = s\ (V\cup\{x\}))}\ s)$,
then by transitivity it suffices to show that $\atrans{\alpha}\ \ftrans{\phi \land x = s\ x}\ s$ follows $\atrans{\alpha}\ \ftrans{\phi}\ s$.
We proceed by inner induction on games, simultaneously for Angel and Demon.
In each case we assume (A) $\atrans{\alpha}\ \ftrans{\phi}\ s$ and show $\atrans{\alpha}\ \ftrans{\phi \land (x = s\ x)}\ s$ or likewise for $\dtrans{\alpha}$.
We do so by inner induction on games, simultaneously for Angel and Demon.

We give the cases for Angel.

\mycase $\humod{y}{f}$:
Since $x \notin \boundvars{\humod{y}{f}}$ then (0) $x \neq y$.
$\atrans{\humod{y}{f}}\ \ftrans{\phi \land (x = s\ x)}\ s
=   (\ftrans{\phi}\ (\lset{s}{x}{(f\ s)}) \mathop{\kwprod} \lget{(\lset{s}{y}{(f\ s)})}{x} = \lget{s}{x})
=_0 (\ftrans{\phi}\ (\lset{s}{x}{(f\ s)}) \mathop{\kwprod} \lget{s}{x} = \lget{s}{x})$
where the left holds by (A) and the right holds reflexively.

\mycase $\prandom{x}$:
From (A) have some $v$ such that (A1) $\ftrans{\phi}\ (\lset{s}{y}{v})$.
Since $x \notin \boundvars{\prandom{y}}$ then (0) $x \neq y$.
\begin{align*}
  &\atrans{\prandom{y}}\ (\ftrans{\phi \land x = s\ x})\ s\\
= &(\sity{v}{\xty}{\ftrans{\phi}\ (\lset{s}{y}{v}) \mathop{\kwprod} (\lget{(\lset{s}{y}{v})}{x} = \lget{s}{x})})\\
= &(\sity{v}{\xty}{\ftrans{\phi}\ (\lset{s}{y}{v}) \mathop{\kwprod} (\lget{s}{x} = \lget{s}{x})})
\end{align*}
where the left holds from (A1) 
and the right holds for \emph{any} $v$ since $(\lget{s}{x} = \lget{s}{x})$ reflexively.

\mycase $\ptest{\psi},$ have
$\atrans{\ptest{\psi}}\ \ftrans{\phi \land (x = s\ x)}\ s
= \ftrans{\psi}\ s \mathop{\kwprod} \ftrans{\phi}\ s \mathop{\kwprod} (\lget{s}{x} = \lget{s}{x})
=_A (\lget{s}{x} = \lget{s}{x})$
which holds reflexively.

\mycase $\pevolvein{\D{y}=f}{\ivr}$:
Note (0) $x \notin \{y,\D{y}\} = \boundvars{\pevolvein{\D{y}=f}{\ivr}}$.
From (A) unpack $d$ and $sol$ such that
(1)$(\solves{sol}{s}{d}{\D{y}=f})$
and (2) $(\pity{t}{{[0,d]}}{\ftrans{\phi}\ (\lset{s}{y}{(sol\ t)})})$
and (3) $\ftrans{\psi}\ (\lset{s}{y}{(sol\ d)})$.
Then by (0) have (4) $\lget{(\lset{s}{y}{(sol\ d)})}{x} = \lget{s}{x}$.
Then using the same $d$ and $sol$ then using (1) (2) (3) and (4) have
$\atrans{\pevolvein{\D{x}=f}{\ivr}}\ \ftrans{\phi \land (x = s\ x)}\ s$ as desired.

\mycase $\alpha;\beta$:
Note (0) $\atrans{\beta}\ \ftrans{\phi \land (x = t\ x)}\ t$ for all $t$ by IH on $\beta$.
Note (1) that the truth value of $(\lambda r.~t\ x = s\ x)$ is constant with respect to  $r$ so it suffices to show $\atrans{\beta} \ftrans{\phi}$, which follows from (A).
\begin{align*}
  &\atrans{\alpha;\beta}\ \ftrans{\phi \land (x = s\ x)}\ s\\
= &\atrans{\alpha}\ (\atrans{\beta}\ \ftrans{\phi \land (x = s\ x)})\ s\\
= &\atrans{\alpha}\ (\ftrans{x = s\ x} \mathop{\kwprod} (\atrans{\beta}\ \ftrans{\phi \land (x = s\ x)}))\ s\\
= &\atrans{\alpha}\ (\lambda t.~ t\ x = s\ x \mathop{\kwprod} (\atrans{\beta}\ \ftrans{\phi \land (x = s\ x)\ t}))\ s\\
\leftarrow_0 &\atrans{\alpha}\ (\lambda t.~ t\ x = s\ x \mathop{\kwprod} (\atrans{\beta}\ \ftrans{\phi \land (t\ x = s\ x)}\ t))\ s\\
\leftarrow &\atrans{\alpha}\ (\atrans{\beta} \ftrans{\phi \land (x = s\ x)})\ s\\
\leftarrow &\atrans{\alpha;\beta}\  \ftrans{\phi \land (x = s\ x)}\ s
\end{align*}
 which follows from the IH on $\alpha$.

\mycase $\alpha\cup\beta$:
Have either $\atrans{\alpha}\ \ftrans{\phi}\ s$ or $\atrans{\beta}\ \ftrans{\phi}\ s$.
In each case, the IH applies by $\boundvars{\alpha\cup\beta}^\complement = \boundvars{\alpha}^\complement \cup \boundvars{\beta}^\complement.$
In the first case,
$\atrans{\alpha} \ftrans{\phi \land (x = s\ x)}\ s
\limply \atrans{\alpha\cup\beta} \ftrans{\phi \land (x = s\ x)}\ s$.
In the second case,
\[\atrans{\beta} \ftrans{\phi \land (x = s\ x)}\ s
\limply \atrans{\alpha\cup\beta} \ftrans{\phi \land (x = s\ x)}\ s\]

\mycase $\prepeat{\alpha}$:
Proceed by induction on membership of $s$ in the fixed point in
$\atrans{\prepeat{\alpha}}\ \ftrans{\phi}\ s$.

In the base case, $\atrans{\prepeat{\alpha}}\ \ftrans{\phi \land (x = s\ x)}\ s$ trivially since $\lget{s}{x} = \lget{s}{x}$.
In the inductive case wish to show $\ftrans{\phi \land (x = s\ x)}\ t \limply \atrans{\alpha}\ \ftrans{\phi \land (x = s\ x)} t$ which follows from the inner IH on membership and outer IH on $\alpha$ by transitivity and monotonicity over $\alpha$.
The IH on $\alpha$ applies because $\boundvars{\prepeat{\alpha}} = \boundvars{\alpha}$.

\mycase $\pdual{\alpha}$:
From (A) have $\dtrans{\alpha}\ \ftrans{\phi}\ s$ so by IH have $\dtrans{\alpha}\ \ftrans{\phi \land (x = s\ x)}\ s$ which gives $\atrans{\pdual{\phi}}\ \ftrans{\phi \land (x = s\ x)}\ s$.

We give the cases for Demon.

\mycase $\humod{x}{f}$:
Since $x \notin \boundvars{\humod{y}{f}}$ then (0) $x \neq y$.
Then $\dtrans{\humod{y}{f}}\ \ftrans{\phi \land (x = s\ x)}\ s
= \ftrans{\phi}\ (\lset{s}{y}{(f\ s)}) \mathop{\kwprod} (\lget{(\lset{s}{y}{(f\ s)})}{x} = \lget{s}{x})
=_A (\lget{(\lset{s}{y}{(f\ s)})}{x} = \lget{s}{x})
=_0 (\lget{s}{x} = \lget{s}{x})$
which holds reflexively.

\mycase $\prandom{x}$:
Since $x \notin \boundvars{\prandom{y}}$ then (0) $x \neq y$.
Then $\dtrans{\prandom{y}}\ \ftrans{\phi \land (x = s\ x)}\ s
=   (\pity{v}{\xty}{(\ftrans{\phi}\ (\lset{s}{y}{v}) \mathop{\kwprod} \lget{(\lset{s}{y}{v})}{x} = \lget{s}{x})})
=_A (\pity{v}{\xty}{(\lget{(\lset{s}{y}{v})}{x} = \lget{s}{x})})
=   (\pity{v}{\xty}{(\lget{s}{x} = \lget{s}{x})})$
which holds for all $v$ since
$(\lget{s}{x} = \lget{s}{x})$ reflexively.

\mycase $\ptest{\psi},$ have
$\dtrans{\ptest{\psi}}\ \ftrans{\phi \land (x = s\ x)}\ s
= (\ftrans{\psi}\ s \limply \ftrans{\phi}\ s \mathop{\kwprod} (\lget{s}{x} = \lget{s}{x})\ s)
=_A (\ftrans{\psi}\ s \limply (\lget{s}{x} = \lget{s}{x}))
= (\ftrans{\psi}\ s \limply \btt)$
since $(\lget{s}{x} = \lget{s}{x})$ holds reflexively.

\mycase $\pevolvein{\D{x}=f}{\ivr}$:
Note (0) $x \notin \{y,\D{y}\} = \boundvars{\pevolvein{\D{y}=f}{\ivr}}$.
Assume arbitrary $d$ and $sol$ such that
(1)$(\solves{sol}{s}{d}{\D{y}=f})$
and (2) $(\pity{t}{{[0,d]}}{\ftrans{\phi}\ {(\lset{s}{y}{(sol\ t)})}})$
and (3) $\ftrans{\psi}\ (\lset{s}{y}{(sol\ d)})$.
Then by (0) have (4) $\lget{(\lset{s}{y}{(sol\ d)})}{x} = \lget{s}{x}$.
Then using (1) (2) (3) and (3) have
$\atrans{\pevolvein{\D{x}=f}{\ivr}}\ \ftrans{\phi \land (x = s\ x)}\ s$ as desired.

\mycase $\alpha;\beta$:
Note (0) $\dtrans{\beta}\ \ftrans{\phi \land (x = t\ x)}\ t$ for all $t$ by IH on $\beta$.
Note (1) that the truth value of $(\lambda r.~t\ x = s\ x)$ is constant as a function of $r$ and it suffices to show $\dtrans{\beta}\ (\ftrans{\phi})\ r$.
\begin{align*}
  &\dtrans{\alpha;\beta}\ \ftrans{\phi \land (x = s\ x)}\ s\\
= &\dtrans{\alpha}\ (\dtrans{\beta}\ \ftrans{\phi \land (x = s\ x)})\ s\\
= &\dtrans{\alpha}\ (\lambda t.~ \ftrans{x = s\ x}\ t \mathop{\kwprod} (\dtrans{\beta}\ \ftrans{\phi \land (x = s\ x)}\ t))\ s\\
= &\dtrans{\alpha}\ (\lambda t.~ t\ x = s\ x \mathop{\kwprod} (\dtrans{\beta}\ \ftrans{\phi \land (x = s\ x)}\ t))\ s\\
\leftarrow_0 &\dtrans{\alpha}\ (\lambda t.~ t\ x = s\ x \mathop{\kwprod} (\dtrans{\beta}\ \ftrans{\phi \land (t\ x = s\ x)}\ t))\ s\\
\leftarrow_1 &\dtrans{\alpha}\ (\lambda t.~ t\ x = s\ x \mathop{\kwprod} (\dtrans{\beta}\ (\ftrans{\phi})\ t))\ s\\
\leftarrow_1 &\dtrans{\alpha}\ (\dtrans{\beta}\ \ftrans{\phi})\ s\\
\leftarrow_1 &\dtrans{\alpha;\beta}\ \ftrans{\phi}\ s
\end{align*}
which is (A).

\mycase $\alpha\cup\beta$:
Have both $\dtrans{\alpha}\ \ftrans{\phi}\ s$ and $\dtrans{\beta}\ \ftrans{\phi}\ s$.
Each IH applies by $\boundvars{\alpha\cup\beta}^\complement = \boundvars{\alpha}^\complement \cup \boundvars{\beta}^\complement.$
The first IH gives $\dtrans{\alpha}\ \ftrans{\phi \land (x = s\ x)}\ s$ and the second gives $\dtrans{\beta} \ftrans{\phi \land (x = s\ x)}\ s$, then by conjunction introduction have $\dtrans{\alpha\cup\beta}\ \ftrans{\phi \land (x = s\ x)}\ s$.

\mycase $\prepeat{\alpha}$:
By inversion on the proof (A) of $\dtrans{\prepeat{\alpha}}\ \ftrans{\phi}\ s$, we have
some $J$ such that 
(1) $J\ s$,
(2) $J\ t \to \dtrans{\alpha} J\ t$ for all $t$, and
(3) $J\ t \to \ftrans{\phi}\ t$ for all $t$.
We show that  $J\ s \land x = s\ x$ is sufficient for an invariant proof of postcondition $x = s\ x$.

(1a) $\ftrans{J \land x = s\ x}\ s$ from (1) and reflexivity.
(2a) $\ftrans{J \land x = s\ x \limply \dbox{\alpha}{(J \land x = s\ x)}}\ t$ because (2) gives a proof of $\dtrans{\alpha}\ \ftrans{J}\ t$ from which the IH on $\alpha$ gives $\dtrans{\alpha}\ \ftrans{J \land x = s\ x}\ t$ as desired.
(3a) $\ftrans{(J \land x = s\ x) \limply (\phi \land x = s\ x)}\ t$ from (3) and hypothesis rule.

Then the coinductive generated by (1a), (2a), and (3a) is a proof of $\dtrans{\prepeat{\alpha}}\ \ftrans{\phi \land x = s\ x}\ s$ as desired.

\mycase $\pdual{\alpha}$:
From (A) have $\atrans{\alpha}\ \ftrans{\phi}\ s$ so by IH have $\atrans{\alpha}\ \ftrans{\phi \land (x = s\ x)}\ s$ which gives $\dtrans{\pdual{\phi}}\ \ftrans{\phi \land (x = s\ x)}\ s$.
\end{proof}

\begin{lemma}[Transposition]
$\eren{\eren{e}{x}{y}}{x}{y} = e,$ for any \CdGL or CIC term, formula, game, or type $e$.
\label{lem:app-transpose}
\end{lemma}
\begin{proof}
  Trivial induction because we define $\eren{e}{x}{y}$ to be \emph{transposition renaming} which renames $x$ to $y$ but also renames $y$ to $x$.
\end{proof}

\begin{lemma}[Formula uniform renaming]
  If $\mproves{\G}{M}{(\ftrans{\phi}\ s)}$ then $\mproves{\eren{\G}{x}{y}}{\eren{M}{x}{y}}{\eren{\phi}{x}{y}}$.
  If $M : \ftrans{\G}\ s$ then $\eren{M}{x}{f} : \ftrans{\eren{\G}{x}{y}}\ (\eren{\phi}{x}{y})$.
  Also, $(\solves{sol}{s}{d}{\D{z}=f}) = (\solves{sol}{\eren{s}{x}{y}}{d}{\D{\eren{z}{x}{y}}=\eren{f}{x}{y}})$
\label{lem:app-formula-rename}
\end{lemma}
\begin{lemma}[Game uniform renaming]
  If $\mproves{\G}{M}{(\atrans{\alpha}\ \ftrans{\phi}\ s)}$ then $\mproves{\eren{\G}{x}{y}}{\eren{M}{x}{y}}{(\atrans{\eren{\alpha}{x}{y}}\ \ftrans{\eren{\phi}{x}{y}}\ \eren{s}{x}{y})}$.
  If $\mproves{\G}{M}{(\dtrans{\alpha}\ \ftrans{\phi}\ s)}$ then $\mproves{\eren{\G}{x}{y}}{\eren{M}{x}{y}}{(\dtrans{\eren{\alpha}{x}{y}}\ \ftrans{\eren{\phi}{x}{y}}\ \eren{s}{x}{y})}$.
\label{lem:app-game-rename}
\end{lemma}
\begin{proof}[Proof of \rref{lem:app-formula-rename} and \rref{lem:app-game-rename}]
The cases for formulas, contexts, solutions, and games are all proven by simultaneous induction.
We give the cases for contexts first.

\mycase $\Gemp,$ have
$\eren{\Gemp}{x}{y} = \Gemp$ and $\ftrans{\Gemp}\ s$ holds trivially for all $s$.

\mycase $G,\psi$:
Assume $\ftrans{(\G,\psi)}\ s = \ftrans{\G}\ s \mathop{\kwprod} \ftrans{\phi}\ s,$ so that by IH on $\G$ have (0) $ \eren{\sprojL{M}}{x}{y} : \ftrans{\eren{\G}{x}{y}}\ (\eren{s}{x}{y})$ and by the IH on $\psi$ have 
(1)  $\eren{\sprojR{M}}{x}{y} : \ftrans{\eren{\psi}{x}{y}}\ (\eren{\G}{x}{y})$.
Then by conjunction of (0) and (1) have $\eren{M}{x}{y} : \ftrans{\eren{\G,\psi}{x}{y}}\ (\eren{s}{x}{y})$ as desired.

The formula and game cases employ the following simplification using the context case.
Each case first assumes (A1) a sequent of shape $\ftrans{\G}\ s \to \ftrans{\phi}\ s$ then exhibits a sequent of shape $\ftrans{\eren{\G}{x}{y}}\ (\eren{s}{x}{y}) \to \ftrans{Q}\ (\eren{s}{x}{y})$, the first step of which is to assume (A2) $\ftrans{\eren{\G}{x}{y}}\ (\eren{s}{x}{y})$.
From (A2), uniform renaming on contexts yields (by \rref{lem:app-transpose}) $\ftrans{\G}\ s,$ then by modus ponens on (A1) have $\ftrans{\phi}\ s$.
That is, the remaining cases are free to ignore the context $\G$.

Also, we write $\eren{z}{x}{y}$ as shorthand for a variable which is $z$ in the case $z \notin\{x,y\},$ or $y$ when $z=x$, or $x$ when $z=y$.

We give the formula cases.

\mycase $\ddiamond{\alpha}{\phi}$:
Assume
$ \ftrans{\ddiamond{\alpha}{\phi}}\ s
= \atrans{\alpha}\ \ftrans{\phi}\ s
= \atrans{\eren{\alpha}{x}{y}}\ \eren{\ftrans{\phi}}{x}{y}\ \eren{s}{x}{y}
= \atrans{\eren{\alpha}{x}{y}}\ \ftrans{\eren{\phi}{x}{y}}\ \eren{s}{x}{y}
= \ftrans{\ddiamond{\alpha}{\phi}}\ \eren{s}{x}{y}$

\mycase $\dbox{\alpha}{\phi},$ have
$ \ftrans{\dbox{\alpha}{\phi}}\ s
= \dtrans{\alpha}\ \ftrans{\phi}\ s
= \dtrans{\eren{\alpha}{x}{y}}\ \eren{\ftrans{\phi}}{x}{y}\ \eren{s}{x}{y}
= \dtrans{\eren{\alpha}{x}{y}}\ \ftrans{\eren{\phi}{x}{y}}\ \eren{s}{x}{y}
= \ftrans{\dbox{\alpha}{\phi}}\ \eren{s}{x}{y}$

\mycase $f \sim g$:
follows from \rref{lem:app-term-rename}:
$ \ftrans{f \sim g}\ s
= f\ s \sim g\ s
= (\eren{f}{x}{y})\ (\eren{s}{x}{y}) \sim (\eren{g}{x}{y})\ (\eren{s}{x}{y})
= \ftrans{\eren{(f \sim g)}{x}{y}}\ (\eren{s}{x}{y})$

We give the Angel cases.

\mycase $\humod{z}{f},$ have
$\atrans{\humod{z}{f}}\ \ftrans{\phi}\ s
= \ftrans{\phi} (\lset{s}{z}{(f\ s)})
= \ftrans{\eren{\phi}{x}{y}}\ \eren{(\lset{s}{z}{(f\ s)})}{x}{y}
= \ftrans{\eren{\phi}{x}{y}}\ (\lset{(\eren{s}{x}{y})}{(\eren{z}{x}{y})}{(f\ s)})
= \ftrans{\eren{\phi}{x}{y}}\ (\lset{(\eren{s}{x}{y})}{(\eren{z}{x}{y})}{(\eren{f}{x}{y}\ \eren{s}{x}{y})})
= \atrans{\humod{\eren{z}{x}{y}}{\eren{f}{x}{y}}}\ \ftrans{\eren{\phi}{x}{y}}\ \eren{s}{x}{y}
= \atrans{\eren{\humod{z}{f}}{x}{y}}\ \ftrans{\eren{\phi}{x}{y}}\ \eren{s}{x}{y}$

\mycase $\prandom{x},$ have
\begin{align*}
  &\atrans{\prandom{z}}\ \ftrans{\phi}\ s\\
= &\sity{v}{\sty}{\ftrans{\phi} (\lset{s}{z}{v})}\\
= &\sity{v}{\sty}{\ftrans{\eren{\phi}{x}{y}}\ \eren{(\lset{s}{z}{v})}{x}{y}}\\
= &\sity{v}{\sty}{\ftrans{\eren{\phi}{x}{y}}\ (\lset{(\eren{s}{x}{y})}{(\eren{z}{x}{y})}{v})}\\
= &\atrans{\prandom{\eren{z}{x}{y}}}\ \ftrans{\eren{\phi}{x}{y}}\ \eren{s}{x}{y}\\
= &\atrans{\eren{\{\prandom{z}\}}{x}{y}}\ \ftrans{\eren{\phi}{x}{y}}\ \eren{s}{x}{y}
\end{align*}

\mycase $\ptest{\psi},$ have
$\atrans{\ptest{\psi}}\ \ftrans{\phi}\ s
= \ftrans{\psi}\ s \mathop{\kwprod} \ftrans{\phi}\ s
= \ftrans{\eren{\psi}{x}{y}}\ \eren{s}{x}{y} \mathop{\kwprod} \ftrans{\eren{\phi}{}{y}}\ \eren{s}{x}{y}
= \atrans{\eren{\ptest{\psi}}{x}{y}}\ \ftrans{\eren{\phi}{x}{y}}\ \eren{s}{x}{y}$

\mycase $\solves{sol}{s}{d}{\D{z}=f},$ have
\begin{align*}
  &(\solves{sol}{s}{d}{\D{z}=f})\\
= &\sprod{(\lget{s}{z} = sol\ 0)}{\pity{r}{[0,d]}{(\der{sol}\ r = f\ (\lset{s}{z}{(sol\ r)}))}}\\
= &\sprod{(\lget{\eren{s}{x}{y}}{\eren{z}{x}{y}} = sol\ 0)}{\pity{r}{[0,d]}{(\der{sol}\ r = \eren{f}{x}{y}\ (\eren{(\lset{s}{z}{(sol\ r)})}{x}{y}))}}\\\\
= &\sprod{(\lget{\eren{s}{x}{y}}{\eren{z}{x}{y}} = sol\ 0)}{\pity{r}{[0,d]}{(\der{sol}\ r = \eren{f}{x}{y}\ (\lset{s}{(\eren{z}{x}{y})}{(sol\ r)}))}}
= &(\solves{sol}{\eren{s}{x}{y}}{d}{\D{\eren{z}{x}{y}}=\eren{f}{x}{y}})
\end{align*}

\mycase $\pevolvein{\D{z}=f}{\ivr},$ have

\begin{align*}
  &\atrans{\pevolvein{\D{z}=f}{\ivr}}\ \ftrans{\phi}\ s\\
= &\sity{d}{{\reals_{\geq0}}}{\sity{sol}{[0,d]\to\xty}{}} (\solves{sol}{s}{d}{\D{z}=f})\\
&\mathop{\kwprod} (\pity{t}{{[0,d]}}{\ftrans{\phi}\ {(\lset{s}{z}{(sol\ t)})}})\\
&\mathop{\kwprod} \ftrans{\phi}\ (\lset{s}{(x,\D{x})}{(sol\ d, f\ (\lset{s}{x}{(sol\ d)}))})\\
= &\sity{d}{{\reals_{\geq0}}}{\sity{sol}{[0,d]\to\xty}{}} (\solves{sol}{\eren{s}{x}{y}}{d}{\D{\eren{z}{x}{y}}=\eren{f}{x}{y}})\\
  &\mathop{\kwprod} (\pity{t}{{[0,d]}}{\ftrans{\eren{\phi}{x}{y}}\ {\eren{(\lset{s}{z}{(sol\ t)})}{x}{y}}})\\
  &\mathop{\kwprod} \ftrans{\eren{\phi}{x}{y}}\ (\eren{(\lset{s}{(x,\D{x})}{(sol\ d, f\ (\lset{s}{x}{(sol\ d)}))})}{x}{y})\\
= &\sity{d}{{\reals_{\geq0}}}{\sity{sol}{[0,d]\to\xty}{}} (\solves{sol}{\eren{s}{x}{y}}{d}{\D{\eren{z}{x}{y}}=\eren{f}{x}{y}})\\
  &\mathop{\kwprod} (\pity{t}{{[0,d]}}{\ftrans{\eren{\phi}{x}{y}}\ {(\lset{\eren{s}{x}{y}}{\eren{z}{x}{y}}{(sol\ t)})}})\\
  &\mathop{\kwprod} \ftrans{\eren{\phi}{x}{y}}\ (\lset{\eren{s}{x}{y}}{(z,\D{z})}{(sol\ d, f\ (\lset{s}{z}{(sol\ d)}))})
\end{align*}

\mycase $\alpha;\beta,$ have
$ \atrans{\alpha;\beta}\ \ftrans{\phi}\ s
= \atrans{\alpha}\ (\atrans{\beta}\ \ftrans{\phi})\ s
= \atrans{\alpha}\ (\lambda t.~ \atrans{\beta}\ \ftrans{\phi}\ t)\ s
= \atrans{\eren{\alpha}{x}{y}}\ (\lambda t.~ \atrans{\eren{\beta}{x}{y}}\ \ftrans{\eren{\phi}{x}{y}}\ t)\ \eren{s}{x}{y}
= \atrans{\eren{\alpha}{x}{y}}\ (\eren{(\atrans{\beta}\ \ftrans{\phi})}{x}{y})\ \eren{s}{x}{y}
= \atrans{\eren{\alpha;\beta}{x}{y}}\ \ftrans{\eren{\phi}{x}{y}}\ \eren{s}{x}{y}$

\mycase $\alpha\cup\beta,$ have
$ \atrans{\alpha\cup\beta}\ \ftrans{\phi}\ s
= \atrans{\alpha}\ \ftrans{\phi}\ s + \atrans{\beta}\ \ftrans{\phi}\ s
= \atrans{\eren{\alpha}{x}{y}}\ \ftrans{\eren{\phi}{x}{y}}\ \eren{s}{x}{y}
+ \atrans{\eren{\beta}{x}{y}}\ \ftrans{\eren{\phi}{x}{y}}\ \eren{s}{x}{y}
= \atrans{\eren{\alpha\cup\beta}{x}{y}}\ \ftrans{\eren{\phi}{x}{y}}\ \eren{s}{x}{y}$

\mycase $\prepeat{\alpha}$
Note (0)
$(\lindty{\tau'\mathrel{:}(\sty \to \alltype)}{}\,\lambda{t:\sty}\,
(\phi\ t \to \tau'\ t)
\mathop{\kwprod}
(\atrans{\alpha}\ \tau'\ t \to \tau'\ t)
)\eren{}{x}{y}
= (\lindty{\tau'\mathrel{:}(\sty \to \alltype)}{}\,\lambda{t:\sty}\,
(\eren{\phi}{x}{y}\ t \to \tau'\ \eren{t}{x}{y})
\mathop{\kwprod}
(\atrans{\eren{\alpha}{x}{y}}\ \tau'\ t \to \tau'\ t)
)$ which can be proven by an inner induction.
Likewise
$ \atrans{\prepeat{\alpha}}\ \ftrans{\phi}\ s
= (\lindty{\tau'\mathrel{:}(\sty \to \alltype)}{} \lambda{t:\sty}\,
(\phi\ t \to \tau'\ t)
\mathop{\kwprod}
(\atrans{\alpha}\ \tau'\ t \to \tau'\ t)
)\ s
= (\lindty{\tau'\mathrel{:}(\sty \to \alltype)}{} \lambda{t:\sty}\,
(\eren{\phi}{x}{y}\ \eren{t}{x}{y} \to \tau'\ \eren{t}{x}{y})
\mathop{\kwprod}
(\atrans{\eren{\alpha}{x}{y}}\ \tau'\ t \to \tau'\ t)
)\eren{}{x}{y}\ \eren{s}{x}{y},$
by induction on the membership of $s$ in the fixed point.
This simplifies to
$\atrans{\eren{\prepeat{\alpha}}{x}{y}}\ \ftrans{\eren{\phi}{x}{y}}\ \eren{s}{x}{y}$
as desired.

\mycase $\pdual{\alpha},$ have
$ \atrans{\pdual{\alpha}}\ \ftrans{\phi}\ s
= \dtrans{\alpha}\ \ftrans{\phi}\ s
= \dtrans{\eren{\alpha}{x}{y}}\ \ftrans{\eren{\phi}{x}{y}}\ \eren{s}{x}{y}
= \atrans{\eren{\pdual{\alpha}}{x}{y}}\ \ftrans{\eren{\phi}{x}{y}}\ \eren{s}{x}{y}$

We give the Demon cases.

\mycase $\humod{x}{f},$ have
$\dtrans{\humod{z}{f}}\ \ftrans{\phi}\ s
= \ftrans{\phi} (\lset{s}{z}{(f\ s)})
= \ftrans{\eren{\phi}{x}{y}}\ \eren{(\lset{s}{z}{(f\ s)})}{x}{y}
= \ftrans{\eren{\phi}{x}{y}}\ (\lset{(\eren{s}{x}{y})}{(\eren{z}{x}{y})}{(f\ s)})
= \ftrans{\eren{\phi}{x}{y}}\ (\lset{(\eren{s}{x}{y})}{(\eren{z}{x}{y})}{(\eren{f}{x}{y}\ \eren{s}{x}{y})})
= \dtrans{\humod{\eren{z}{x}{y}}{\eren{f}{x}{y}}}\ \ftrans{\eren{\phi}{x}{y}}\ \eren{s}{x}{y}
= \dtrans{\eren{\{\humod{z}{f}\}}{x}{y}}\ \ftrans{\eren{\phi}{x}{y}}\ \eren{s}{x}{y}
$

\mycase $\prandom{x},$ have
\begin{align*}
  &\dtrans{\prandom{z}}\ \ftrans{\phi}\ s\\
= &\pity{v}{\sty}{\ftrans{\phi} (\lset{s}{z}{v})}\\
= &\pity{v}{\sty}{\ftrans{\eren{\phi}{x}{y}}\ \eren{(\lset{s}{z}{v})}{x}{y}}\\
= &\pity{v}{\sty}{\ftrans{\eren{\phi}{x}{y}}\ (\lset{(\eren{s}{x}{y})}{(\eren{z}{x}{y})}{v})}\\
= &\dtrans{\prandom{\eren{z}{x}{y}}}\ \ftrans{\eren{\phi}{x}{y}}\ \eren{s}{x}{y}\\
= &\dtrans{\eren{\{\prandom{z}\}}{x}{y}}\ \ftrans{\eren{\phi}{x}{y}}\ \eren{s}{x}{y}
\end{align*}

\mycase $\ptest{\phi},$ have
$\dtrans{\ptest{\psi}}\ \ftrans{\phi}\ s
= (\ftrans{\psi}\ s \to \ftrans{\phi}\ s)
= (\ftrans{\eren{\psi}{x}{y}}\ \eren{s}{x}{y} \to \ftrans{\eren{\phi}{}{y}}\ \eren{s}{x}{y})
= \dtrans{\eren{\ptest{\psi}}{x}{y}}\ \ftrans{\eren{\phi}{x}{y}}\ \eren{s}{x}{y}$

\mycase $\pevolvein{\D{y}=g}{\ivr}$:
%
%
From (A2) have  (L1) $x \neq y$ and (L2) $y \notin \freevars{g}$.
Note that this is a stronger admissibility condition than those for $\humod{y}{g}$ and $\prandom{y}$.
Unlike the former constructs, $y$ is always a free variable of $\pevolvein{\D{y}=g}{\ivr},$ thus if we attempted to define a sufficient admissibility condition to support the case $x=y,$ we would find it unsatisfiable. 
Thus we simply say $x$ cannot be substituted into $f$ in any ODE which binds $x$.
So (L1) (L2), then have (S)
$\tsub{(\humod{y}{g})}{x}{f} = \humod{x}{(\tsub{g}{x}{f})}$
also, from (L2) have
(*) $(f\ s) = (f\ (\lset{s}{y}{(\tsub{g}{x}{f}\ s)}))$ by \rref{lem:app-term-coincide} since $s = (\lset{s}{y}{(\tsub{g}{x}{f}\ s)})$ on $\{y\}^\complement \supseteq \freevars{f}^\complement$.
then

Have (0)
$ (\solves{sol}{\tsub{s}{x}{f}}{d}{\D{y}=g})
= (\solves{sol}{s}{d}{\tsub{(\D{y}=g)}{x}{f}})$ by the ``solves'' IH.

Have (1) for all $t \in [0,d],$
$\ftrans{\ivr}\ (\lset{(\tsub{s}{x}{f})}{y}{(sol\ t)})
= \ftrans{\tsub{\ivr}{x}{f}}\ (\lset{s}{y}{(sol\ t)})$
by IH on $\ivr$ and because
$ \lset{(\tsub{s}{x}{f})}{y}{(sol\ t)}
= \lset{(\lset{s}{x}{(f\ s)})}{y}{(sol\ t)}
= \lset{(\lset{s}{y}{(sol\ t)})}{x}{(f\ s)}
= \lset{(\lset{s}{y}{(sol\ t)})}{x}{(f\ (\lset{s}{y}{(sol\ t)}))}$
since by (L2) have $y \notin \freevars{f}$ thus $s = (\lset{s}{y}{(sol\ t)})$ on $\freevars{f}^\complement$ and \rref{lem:app-term-coincide} applies.

Have (2)
\begin{align*}
  &\ftrans{\phi}\ (\lset{(\tsub{s}{x}{f})}{(y,\D{y})}{(sol\ d, f\ (\lset{(\tsub{s}{x}{f})}{y}{(sol\ d)}))})\\
= &\ftrans{\tsub{\phi}{x}{f}}\ (\lset{s}{(y,\D{y})}{(sol\ d, f\ (\lset{(\tsub{s}{x}{f})}{y}{(sol\ d)}))})
\end{align*}
by the IH on $\phi$ and because
\begin{align*}
  &\lset{(\tsub{s}{x}{f})}{(y,\D{y})}{(sol\ d, g\ (\lset{(\tsub{s}{x}{f})}{y}{(sol\ d)}))}\\
= &\lset{(\lset{s}{(y,\D{y})}{(sol\ d, g\ (\lset{(\tsub{s}{x}{f})}{y}{(sol\ d)}))})}{x}{(f\ (\lset{s}{y}{(sol\ d)}))}
\end{align*}
by the same argument as above
and because
$ g\ (\lset{(\tsub{s}{x}{f})}{y}{(sol\ d)})
= (\tsub{g}{x}{f})\ (\lset{s}{y}{(sol\ d)})$
by term IH.

\begin{align*}
   &\dtrans{\pevolvein{\D{y}=g}{\ivr}}\ \ftrans{\phi} \tsub{s}{x}{f}\\
=  &\pity{d}{{\reals_{\geq0}}}{\sity{sol}{[0,d]\to\xty}{}}\\
     &(\solves{sol}{\tsub{s}{x}{f}}{d}{\D{y}=g})\\
 \to &(\pity{t}{{[0,d]}}{\ftrans{\ivr}\ {(\lset{\tsub{s}{x}{f}}{y}{(sol\ t)})}})\\
 \to &\ftrans{\phi}\ (\lset{(\tsub{s}{x}{f})}{(y,\D{y})}{(sol\ d, g\ (\lset{(\tsub{s}{x}{f})}{y}{(sol\ d)}))})\\
=_{0,1,2}
&\pity{d}{{\reals_{\geq0}}}{\sity{sol}{[0,d]\to\xty}{}}\\
 &(\solves{sol}{s}{d}{\tsub{(\D{y}=g)}}{x}{f})\\
 \to &(\pity{t}{{[0,d]}}{\ftrans{\tsub{\ivr}{x}{f}}\ {(\lset{s}{y}{(sol\ t)})}})\\
 \to &\ftrans{\tsub{\phi}{x}{f}}\ (\lset{s}{(y,\D{y})}{(sol\ d, (\tsub{g}{x}{f})\ (\lset{s}{y}{(sol\ d)}))})\\
= &\dtrans{\tsub{(\pevolvein{\D{y}=g}{\ivr})}{x}{f}}\ \ftrans{\tsub{\phi}{x}{f}}\ s
\end{align*}

\mycase $\alpha;\beta,$ have
$ \dtrans{\alpha;\beta}\ \ftrans{\phi}\ s
= \dtrans{\alpha}\ (\atrans{\beta}\ \ftrans{\phi})\ s
= \dtrans{\alpha}\ (\lambda t.~ \dtrans{\beta}\ \ftrans{\phi}\ t)\ s
= \dtrans{\eren{\alpha}{x}{y}}\ (\lambda t.~ \dtrans{\eren{\beta}{x}{y}}\ \ftrans{\eren{\phi}{x}{y}}\ t)\ \eren{s}{x}{y}
= \dtrans{\eren{\alpha}{x}{y}}\ (\eren{(\dtrans{\beta}\ \ftrans{\phi})}{x}{y})\ \eren{s}{x}{y}
  = \dtrans{\eren{\alpha;\beta}{x}{y}}\ \ftrans{\eren{\phi}{x}{y}}\ \eren{s}{x}{y}$

\mycase $\alpha\cup\beta,$ have
\begin{align*}
  &\dtrans{\alpha\cup\beta}\ \ftrans{\phi}\ s\\
= &\dtrans{\alpha}\ \ftrans{\phi}\ s \mathop{\kwprod} \dtrans{\beta}\ \ftrans{\phi}\ s\\
= &\dtrans{\eren{\alpha}{x}{y}}\ \ftrans{\eren{\phi}{x}{y}}\ \eren{s}{x}{y}\\
\mathop{\kwprod} &\dtrans{\eren{\beta}{x}{y}}\ \ftrans{\eren{\phi}{x}{y}}\ \eren{s}{x}{y}\\
= &\dtrans{\eren{\alpha\cup\beta}{x}{y}}\ \ftrans{\eren{\phi}{x}{y}}\ \eren{s}{x}{y}
\end{align*}

\mycase $\prepeat{\alpha}$
Note (0)
$\eren{(\lcoty{\tau'\mathrel{:}(\sty \to \alltype)}{\lambda{t:\sty}\,{
(\tau'\ t \to \dtrans{\alpha}\ \tau'\ t)
\mathop{\kwprod}
(\tau'\ t \to \ftrans{\phi}\ t)
}})}{x}{y}
= (\lcoty{\tau'\mathrel{:}(\sty \to \alltype)}{\lambda{t:\sty}\,{
(\tau'\ t \to \dtrans{\eren{\alpha}{x}{y}}\ \tau'\ t)
\mathop{\kwprod}
(\tau'\ t \to \ftrans{\eren{\phi}{x}{y}}\ t)
}})$ which can be proven by an inner coinduction.
Likewise
$ \dtrans{\prepeat{\alpha}}\ \ftrans{\phi}\ s
= (\lcoty{\tau'\mathrel{:}(\sty \to \alltype)}{} \lambda{t:\sty}\,
(\tau'\ t \to \dtrans{\alpha}\ \tau'\ t)
\mathop{\kwprod}
(\tau'\ t \to \ftrans{\phi}\ t)
)\ s
= (\lcoty{\tau'\mathrel{:}(\sty \to \alltype)}{} \lambda{t:\sty}\,
(\tau'\ t \to \dtrans{\eren{\alpha}{x}{y}}\ \tau'\ t)
\mathop{\kwprod}
(\tau'\ t \to \ftrans{\eren{\phi}{x}{y}}\ t)
) \eren{s}{x}{y},$
by coinduction on the membership of $s$ in the fixed point.
This simplifies to
$\dtrans{\eren{\prepeat{\alpha}}{x}{y}}\ \ftrans{\eren{\phi}{x}{y}}\ \eren{s}{x}{y}$
as desired.

\mycase $\pdual{\alpha},$ have
$ \dtrans{\pdual{\alpha}}\ \ftrans{\phi}\ s
= \atrans{\alpha}\ \ftrans{\phi}\ s
= \atrans{\eren{\alpha}{x}{y}}\ \ftrans{\eren{\phi}{x}{y}}\ \eren{s}{x}{y}
= \dtrans{\eren{\pdual{\alpha}}{x}{y}}\ \ftrans{\eren{\phi}{x}{y}}\ \eren{s}{x}{y}$
\end{proof}

\begin{lemma}[Formula substitution]
 If $\mproves{\ftrans{\G}\ s}{M}{\ftrans{\phi}\ s}$ and the substitutions $\tsub{\G}{x}{f}$ and $\tsub{\phi}{x}{f}$ are admissible, then 
 $\mproves{\tsub{\G}{x}{f}\ s}{\tsub{M}{x}{f}\ s}{\ftrans{\tsub{\phi}{x}{f}}\ s}$.
 Likewise for contexts $\G$ and the predicate $(\solves{sol}{s}{d}{\D{y}=g})$.
\label{lem:app-formula-subst}
\end{lemma}
\begin{lemma}[Game substitution]
 If the substitutions $\tsub{\G}{x}{f}$ and $\tsub{\alpha}{x}{f}$ $\tsub{\phi}{x}{f}$ are admissible, then 
 \begin{itemize}
 \item
$\mproves{\ftrans{\G}\ (\tsub{s}{x}{f})}{M\ (\tsub{s}{x}{f})}{\atrans{\alpha}\ \phi\ (\tsub{s}{x}{f})}$ iff
$\mproves{\ftrans{\tsub{\G}{x}{f}}\ s}{\tsub{M}{x}{f}\ s}{\atrans{\tsub{\alpha}{x}{f}}\ \tsub{\phi}{x}{f}\ s}$.
 \item
$\mproves{\ftrans{\G}\ (\tsub{s}{x}{f})}{M\ (\tsub{s}{x}{f})}{\dtrans{\alpha}\ \phi\ (\tsub{s}{x}{f})}$ iff
$\mproves{\ftrans{\tsub{\G}{x}{f}}\ s}{\tsub{M}{x}{f}\ s}{\dtrans{\tsub{\alpha}{x}{f}}\ \tsub{\phi}{x}{f}\ s}$.
 \end{itemize}
where $\tsub{s}{x}{f}$ is shorthand for $\lset{s}{x}{(f\ s)}$.
\label{lem:app-game-subst}
\end{lemma}
\begin{proof}[Proof of \rref{lem:app-formula-subst} and \rref{lem:app-game-subst}]
In the formula cases, we assume (A0) $\mproves{\ftrans{\G}\ s}{M}{\ftrans{\phi}\ s}$, (A1) admissibility of $\tsub{\G}{x}{f}$ and (A2) admissibility of $\tsub{\phi}{x}{f}$.
Likewise for, contexts, games, and predicate $(\solves{sol}{s}{d}{\D{x}=f})$.
We note that in each case, the admissibility conditions of the IH hold following (A1) and (A2) and unpacking the inductive definition of admissibility.
As usual, we also ignore the contexts in the formula and game cases since the cases with contexts follow easily from those without, combined with IH's on the contexts.

We give the context cases first.

\mycase $\G = \Gemp,$ then
trivially $\ftrans{\Gemp}\ \tsub{s}{x}{f}$.

\mycase $\G,\psi$:
From (A0) have $\ftrans{\G}\ \tsub{s}{x}{f}$ and $\ftrans{\psi}\ \tsub{s}{x}{f},$ then by the IH's have $\ftrans{\tsub{\G}{x}{f}}\ s$ and $\ftrans{\tsub{\psi}{x}{f}}\ s$ giving $\ftrans{\tsub{(\G,\psi)}{x}{f}}\ s$ as desired.

We now give the formula cases.

\mycase $f \sim g$:
From (A0) by \rref{lem:app-term-subst} have
$\ftrans{g\sim h}\ \tsub{s}{x}{f}
= (g\ \tsub{s}{x}{f} \sim h\ \tsub{s}{x}{f})
= (\tsub{g}{x}{f}\ s \sim \tsub{h}{x}{f}\ s)
= (\tsub{(g \sim h)}{x}{f})\ s
$

\mycase $\dbox{\alpha}{\phi}$:
From (A0) have by the IH on $\alpha$ that
$\ftrans{\dbox{\alpha}{\phi}}\ \tsub{s}{x}{f}
= \dtrans{\alpha}\ \ftrans{\phi}\ \tsub{s}{x}{f}
= \dtrans{\tsub{\alpha}{x}{f}}\ \ftrans{\tsub{\phi}{x}{f}}\ s
= \ftrans{\tsub{(\dbox{\alpha}{\phi})}{x}{f}}\ s
$

\mycase $\ddiamond{\alpha}{\phi}$:
From (A0) have by the IH on $\alpha$ that
$\ftrans{\ddiamond{\alpha}{\phi}}\ \tsub{s}{x}{f}
= \atrans{\alpha}\ \ftrans{\phi}\ \tsub{s}{x}{f}
= \atrans{\tsub{\alpha}{x}{f}}\ \ftrans{\tsub{\phi}{x}{f}}\ s
= \ftrans{\tsub{\ddiamond{\alpha}{\phi}}{x}{f}}\ s$

\mycase $(\solves{sol}{s}{d}{\D{y}=g})$:
From (A2) have (L1) $x \neq y$ and (L2) $y \notin \freevars{f}$.
Have (0) $\lget{(\tsub{s}{x}{f})}{y} = \lget{s}{y}$ by (L1).
Have (1) for all $r \in [0,d]$ that 
$ g\ (\lset{(\tsub{s}{x}{f})}{y}{(sol\ r)})
= (\tsub{g}{x}{f})\ (\lset{s}{y}{(sol\ r)})$
by \rref{lem:app-term-coincide} and because
$(\lset{(\tsub{s}{x}{f})}{y}{(sol\ r)})
= \tsub{(\lset{s}{y}{(sol\ r)})}{x}{f}$
since by (L2) $s = \lset{s}{y}{(sol\ r)}$ on $\{y\}^\complement \supseteq \freevars{f}^\complement$ thus \rref{lem:app-term-coincide} applies.
Have
\begin{align*}
       &(\solves{sol}{\tsub{s}{x}{f}}{d}{\D{y}=g})\\
=      &(\lget{(\tsub{s}{x}{f})}{y} = sol\ 0  \mathop{\kwprod}  \pity{r}{[0,d]}{(\der{sol}\ r = g\ (\lset{(\tsub{s}{x}{f})}{y}{(sol\ r)}))})\\
=_{0,1} &(\lget{s}{y} = sol\ 0 \mathop{\kwprod} \pity{r}{[0,d]}{(\der{sol}\ r = (\tsub{g}{x}{f})\ (\lset{s}{y}{(sol\ r)}))})\\
=      &(\solves{sol}{s}{d}{\tsub{(\D{y}=g)}{x}{f}})
\end{align*}

We give the Angel cases.

\mycase $\humod{y}{g},$
 from (A2) have either (L1) $x \neq y$ and (L2) $y \notin \freevars{f}$
or (R1) $x=y$ and (R2) $x \notin \freevars{\phi}$.
In the first case, (L1) (L2), then (S)
$\tsub{(\humod{y}{g})}{x}{f} = \humod{x}{(\tsub{g}{x}{f})}$
also, from (L2) have
(*) $(f\ s) = (f\ (\lset{s}{y}{(\tsub{g}{x}{f}\ s)}))$ by \rref{lem:app-term-coincide} since $s = (\lset{s}{y}{(\tsub{g}{x}{f}\ s)})$ on $\{y\}^\complement \supseteq \freevars{f}^\complement$.
Then 
\begin{align*}
  &\atrans{\humod{y}{g}}\ \ftrans{\phi}\ \tsub{s}{x}{f}\\
= &\ftrans{\phi}\ (\lset{(\tsub{s}{x}{f})}{y}{(g\ (\tsub{s}{x}{f}))})\\
= &\ftrans{\phi}\ (\lset{(\lset{s}{y}{(g\ (\tsub{s}{x}{f}))})}{x}{(f\ s)})\\
= &\ftrans{\phi}\ (\lset{(\lset{s}{y}{(\tsub{g}{x}{f}\ s)})}{x}{(f\ s)})\\
=_{*} &\ftrans{\phi}\ (\lset{(\lset{s}{y}{(\tsub{g}{x}{f}\ s)})}{x}{(f\ (\lset{s}{y}{(\tsub{g}{x}{f}\ s)}))})\\
\iheq &\ftrans{\tsub{\phi}{x}{f}}\ (\lset{s}{y}{(\tsub{g}{x}{f}\ s)})\\
= &\atrans{\tsub{\humod{y}{g}}{x}{f}}\ \ftrans{\tsub{\phi}{x}{f}}\ s
\end{align*}
In the second case, (R1) (R2), note (*) for every $t$ have $t = \tsub{t}{x}{f}$ on $\freevars{\phi}$ (since $x \notin \freevars{\phi}$ by R2) by \rref{lem:app-term-coincide}.

Then
$ \atrans{\humod{y}{g}}\ \ftrans{\phi}\ \tsub{s}{x}{f}
= \atrans{\humod{x}{g}}\ \ftrans{\phi}\ \tsub{s}{x}{f}
= \ftrans{\phi}\ (\lset{(\tsub{s}{x}{f})}{x}{(g\ (\tsub{s}{x}{f}))})
= \ftrans{\phi}\ (\lset{s}{x}{(g\ (\tsub{s}{x}{f}))})
=_* \ftrans{\phi}\ \tsub{(\lset{s}{x}{(g\ \tsub{s}{x}{f})})}{x}{f}
\iheq \ftrans{\tsub{\phi}{x}{f}}\ (\lset{s}{x}{(g\ \tsub{s}{x}{f})})
= \ftrans{(\tsub{\phi}{x}{f})}\ (\lset{s}{x}{(g\ (\lset{s}{x}{(f\ s)}))})
= \ftrans{(\tsub{\phi}{x}{f})}\ (\lset{s}{x}{((\tsub{g}{x}{f})\ s)})
= \atrans{\humod{x}{(\tsub{g}{x}{f})}}\ \tsub{\phi}{x}{f}\ s
= \atrans{\tsub{\{\humod{x}{g}\}}{x}{f}}\ \tsub{\phi}{x}{f}\ s
= \atrans{\tsub{\{\humod{y}{g}\}}{x}{f}}\ \tsub{\phi}{x}{f}\ s
$

\mycase $\prandom{y}$:
From (A2) have either (L1) $x \neq y$ and (L2) $y \notin \freevars{f}$
or (R1) $x=y$ and (R2) $x \notin \freevars{\phi}$.
In both cases, then (S) $\tsub{\{\prandom{y}\}}{x}{f} = \prandom{y}$.
In case (L) then
$ \atrans{\prandom{y}}\ \ftrans{\phi}\ \tsub{s}{x}{f}
=  \sity{v}{\xty}{\ftrans{\phi}\ (\lset{(\tsub{s}{x}{f})}{y}{v})}
=  \sity{v}{\xty}{\ftrans{\phi}\ (\lset{(\lset{s}{x}{(f\ s)})}{y}{v})}
=_{L1}  \sity{v}{\xty}{\ftrans{\phi}\ (\lset{(\lset{s}{y}{v})}{x}{(f\ s)}))}
=_{L2}  \sity{v}{\xty}{\ftrans{\phi}\ (\lset{(\lset{s}{y}{v})}{x}{(f\ (\lset{s}{y}{v}))}))}
=  \sity{v}{\xty}{\ftrans{\tsub{\phi}{x}{f}}\ (\lset{s}{y}{v})}
= \atrans{\prandom{y}}\ \ftrans{\tsub{\phi}{x}{f}}\ s
=_S \atrans{\tsub{\{\prandom{y}\}}{x}{f}}\ \ftrans{\tsub{\phi}{x}{f}}\ s$

In case (R1) (R2) then
note (*) for every $t$ have $t = \tsub{t}{x}{f}$ on $\freevars{\phi}$ (since $x \notin \freevars{\phi}$ by R2) by \rref{lem:app-term-coincide}.
then
$ \atrans{\prandom{y}}\ \ftrans{\phi}\ \tsub{s}{x}{f}
= \atrans{\prandom{x}}\ \ftrans{\phi}\ \tsub{s}{x}{f}
= \sity{v}{\xty}{\ftrans{\phi}\ (\lset{(\tsub{s}{x}{f})}{x}{v})}
= \sity{v}{\xty}{\ftrans{\phi}\ (\lset{s}{x}{v})}
=_* \sity{v}{\xty}{\ftrans{\phi}\ \tsub{(\lset{s}{x}{v})}{x}{f}}
\iheq \sity{v}{\xty}{\ftrans{\tsub{\phi}{x}{f}}\ (\lset{s}{x}{v})}
= \atrans{\prandom{x}}\ \ftrans{\tsub{\phi}{x}{f}}\ s
=_{R1} \atrans{\prandom{y}}\ \ftrans{\tsub{\phi}{x}{f}}\ s
=_S \atrans{\tsub{\prandom{y}}{x}{f}}\ \ftrans{\tsub{\phi}{x}{f}}\ s$.

\mycase $\ptest{\psi},$ have
$ \atrans{\ptest{\psi}}\ \ftrans{\phi}\ \tsub{s}{x}{f}
= \ftrans{\psi}\ \tsub{s}{x}{f} \mathop{\kwprod} \ftrans{\phi} \tsub{s}{x}{f}
= \ftrans{\tsub{\psi}{x}{f}}\ s \mathop{\kwprod} \ftrans{\tsub{\phi}{x}{f}}\ s
= \atrans{\tsub{\ptest{\psi}}{x}{f}} \ftrans{\tsub{\phi}{x}{f}}\ s$

\mycase $\pevolvein{\D{y}=g}{\ivr}$:
From (A2) have  (L1) $x \neq y$ and (L2) $y \notin \freevars{f}$.
Note that this is a stronger admissibility condition than those for $\humod{y}{f}$ and $\prandom{y}$.
Unlike the former constructs, $y$ is always a free variable of $\pevolvein{\D{y}=g}{\ivr},$ thus if we attempted to define a sufficient admissibility condition to support the case $x=y,$ we would find it unsatisfiable.
Thus we simply say $x$ cannot be substituted into $f$ in any ODE which binds $x$.
So (L1) (L2), then have (S)
$\tsub{(\humod{y}{g})}{x}{f} = \humod{x}{(\tsub{g}{x}{f})}$
also, from (L2) have
(*) $(f\ s) = (f\ (\lset{s}{y}{(\tsub{g}{x}{f}\ s)}))$ by \rref{lem:app-term-coincide} since $s = (\lset{s}{y}{(\tsub{g}{x}{f}\ s)})$ on $\{y\}^\complement \supseteq \freevars{f}^\complement$.
then

Have (0) $(\solves{sol}{\tsub{s}{x}{f}}{d}{\D{y}=g}) = (\solves{sol}{s}{d}{\tsub{(\D{y}=g)}{x}{f}})$ by the ``solves'' IH.

Have (1) for all $t \in [0,d],$
$\ftrans{\ivr}\ (\lset{(\tsub{s}{x}{f})}{y}{(sol\ t)})
= \ftrans{\tsub{\ivr}{x}{f}}\ (\lset{s}{y}{(sol\ t)})$
by IH on $\ivr$ and because
$ \lset{(\tsub{s}{x}{f})}{y}{(sol\ t)}
= \lset{(\lset{s}{x}{(f\ s)})}{y}{(sol\ t)}
= \lset{(\lset{s}{y}{(sol\ t)})}{x}{(f\ s)}
= \lset{(\lset{s}{y}{(sol\ t)})}{x}{(f\ (\lset{s}{y}{(sol\ t)}))}$
since by (L2) have $y \notin \freevars{f}$ thus $s = (\lset{s}{y}{(sol\ t)})$ on $\freevars{f}^\complement$ and \rref{lem:app-term-coincide} applies.

Have (2)
\begin{align*}
  &\ftrans{\phi}\ (\lset{(\tsub{s}{x}{f})}{(y,\D{y})}{(sol\ d, g\ (\lset{(\tsub{s}{x}{f})}{y}{(sol\ d)}))})\\
=\,&\ftrans{\tsub{\phi}{x}{f}}\ (\lset{s}{(y,\D{y})}{(sol\ d, \tsub{g}{x}{f}\ (\lset{s}{y}{(sol\ d)}))})
\end{align*}
by the IH on $\phi$ and because
\begin{align*}
  &\lset{(\tsub{s}{x}{f})}{y}{(sol\ d)}\\
=\,&\lset{(\lset{y}{(sol\ d)})}{x}{(f\ (\lset{s}{y}{(sol\ d)}))}
\end{align*}
by the same argument as above and by term case because
$(\lset{(\tsub{s}{x}{f})}{y}{(sol\ d)})
= \tsub{(\lset{s}{y}{(sol\ d)})}{x}{f}$.

\begin{align*}
   &\atrans{\pevolvein{\D{y}=g}{\ivr}}\ \ftrans{\phi} \tsub{s}{x}{f}\\
=  &\sity{d}{{\reals_{\geq0}}}{\sity{sol}{[0,d]\to\xty}{}}\\
 &(\solves{sol}{\tsub{s}{x}{f}}{d}{\D{y}=g})\\
 &\mathop{\kwprod} (\pity{t}{{[0,d]}}{\ftrans{\ivr}\ {(\lset{\tsub{s}{x}{f}}{y}{(sol\ t)})}})\\
 &\mathop{\kwprod} \ftrans{\phi}\ (\lset{(\tsub{s}{x}{f})}{(y,\D{y})}{(sol\ d, g\ (\lset{(\tsub{s}{x}{f})}{y}{(sol\ d)}))})\\
=_{0,1,2}
&\sity{d}{{\reals_{\geq0}}}{\sity{sol}{[0,d]\to\xty}{}}\\
 &(\solves{sol}{s}{d}{\tsub{(\D{y}=g)}}{x}{f})\\
 &\mathop{\kwprod} (\pity{t}{{[0,d]}}{\ftrans{\tsub{\ivr}{x}{f}}\ {(\lset{s}{y}{(sol\ t)})}})\\
 &\mathop{\kwprod} \ftrans{\tsub{\phi}{x}{f}}\ (\lset{s}{(y,\D{y})}{(sol\ d, \tsub{g}{x}{f}\ (\lset{s}{y}{(sol\ d)}))})\\
= &\atrans{\tsub{(\pevolvein{\D{y}=g}{\ivr})}{x}{f}}\ \ftrans{\tsub{\phi}{x}{f}} s
\end{align*}

\mycase $\alpha;\beta,$ have
$ \atrans{\alpha;\beta}\ \ftrans{\phi}\ \tsub{s}{x}{f}
= \atrans{\alpha}\ (\atrans{\beta}\ \ftrans{\phi})\ \tsub{s}{x}{f}
= \atrans{\alpha}\ (\ftrans{\ddiamond{\beta}{\phi}})\ \tsub{s}{x}{f}
= \atrans{\tsub{\alpha}{x}{f}}\ (\ftrans{\tsub{\ddiamond{\beta}{\phi}}{x}{f}})\ s
= \atrans{\tsub{\alpha}{x}{f}}\ (\atrans{\tsub{\beta}{x}{f}}\ \ftrans{\tsub{\phi}{x}{f}})\ s
= \atrans{\tsub{\{\alpha;\beta\}}{x}{f}}\ \ftrans{\tsub{\phi}{x}{f}}\ s
$

\mycase $\alpha\cup\beta,$ have
$ \atrans{\alpha\cup\beta}\ \ftrans{\phi}\ \tsub{s}{x}{f}
=  \atrans{\alpha}\ \ftrans{\phi}\ \tsub{s}{x}{f} \mathop{\kwsum} \atrans{\beta}\ \ftrans{\phi}\ \tsub{s}{x}{f}
=  \atrans{\tsub{\alpha}{x}{f}}\ \ftrans{\tsub{\phi}{x}{f}} s \mathop{\kwsum} \atrans{\tsub{\beta}{x}{f}}\ \ftrans{\tsub{\phi}{x}{f}} s
= \atrans{\tsub{\{\alpha\cup\beta\}}{x}{f}}\ \ftrans{\tsub{\phi}{x}{f}} s$

\mycase $\prepeat{\alpha}$:
The main step (*) holds by inner induction on the membership of $\tsub{s}{x}{f}$ in the fixed point.
The two fixed point solutions $\tau'$ and $\tau''$ are related by $\tau'\ \tsub{s}{x}{f}$ iff $\tau'' s$ for all $s : \sty$.

In the base case,
$(\psi\ \tsub{s}{x}{f} \to \tau'\ \tsub{s}{x}{f})
= (\psi\ \tsub{s}{x}{f} \to \tau''\ s)
= (\tsub{\psi}{x}{f}\ s \to \tau''\ s)$.

In the inductive case,
$ (\atrans{\alpha}\ \tau'\ \tsub{s}{x}{f} \to \tau'\ \tsub{s}{x}{f})
= (\atrans{\tsub{\alpha}{x}{f}}\ \tau''\ s \to \tau'\ \tsub{s}{x}{f})
= (\atrans{\tsub{\alpha}{x}{f}}\ \tau''\ s \to \tau''\ s)$.

Then the main proof of the case proceeds:

$ \atrans{\prepeat{\alpha}}\ \ftrans{\phi}\ \tsub{s}{x}{f}
=
(\lindty{\tau'\mathrel{:}(\sty \to \alltype)}{\lambda{t:\sty}\,{
(\psi\ t \to \tau'\ t)
\mathop{\kwprod}
(\atrans{\alpha}\ \tau'\ t \to \tau'\ t)
}})\ (\tsub{s}{x}{f})
=_{*}
(\lindty{\tau''\mathrel{:}(\sty \to \alltype)}{\lambda{t:\sty}\,{
(\tsub{\psi}{x}{f}\ t \to \tau''\ t)
\mathop{\kwprod}
(\atrans{\tsub{\alpha}{x}{f}}\ \tau''\ t \to \tau''\ t)
}})\ (\tsub{s}{x}{f})
= \atrans{\tsub{\prepeat{\alpha}}{x}{f}}\ \ftrans{\tsub{\phi}{x}{f}}\ s$.

\mycase $\pdual{\alpha},$ have
$ \atrans{\pdual{\alpha}}\ \ftrans{\phi}\ \tsub{s}{x}{f}
= \dtrans{\alpha}\ \ftrans{\phi}\ \tsub{s}{x}{f}
\iheq \dtrans{\tsub{\alpha}{x}{f}}\ \ftrans{\tsub{\phi}{x}{f}}\ s
= \atrans{\pdual{(\tsub{\alpha}{x}{f})}}\ \ftrans{\tsub{\phi}{x}{f}}\ s
= \atrans{\tsub{\{\pdual{\alpha}\}}{x}{f}}\ \ftrans{\tsub{\phi}{x}{f}}\ s
$

We give the Demon cases.

\mycase $\humod{y}{g},$
From (A2) have either (L1) $x \neq y$ and (L2) $y \notin \freevars{f}$
or (R1) $x=y$ and (R2) $x \notin \freevars{\phi}$.

In the first case, (L1) (L2), then (S)
$\tsub{\{\humod{y}{g}\}}{x}{f} = \{\humod{x}{(\tsub{g}{x}{f})}\}.$
Also, from (L2), have
(*) $(f\ s) = (f\ (\lset{s}{y}{(\tsub{g}{x}{f}\ s)}))$ by \rref{lem:app-term-coincide} since $s = (\lset{s}{y}{(\tsub{g}{x}{f}\ s)})$ on $\{y\}^\complement \supseteq \freevars{f}^\complement$
thus
\begin{align*}
  &\dtrans{\humod{y}{g}}\ \ftrans{\phi}\ \tsub{s}{x}{f}\\
= &\ftrans{\phi}\ (\lset{(\tsub{s}{x}{f})}{y}{(g\ (\tsub{s}{x}{f}))})\\
= &\ftrans{\phi}\ (\lset{(\lset{s}{y}{(g\ (\tsub{s}{x}{f}))})}{x}{(f\ s)})\\
= &\ftrans{\phi}\ (\lset{(\lset{s}{y}{(\tsub{g}{x}{f}\ s)})}{x}{(f\ s)})\\
=_{*} &\ftrans{\phi}\ (\lset{(\lset{s}{y}{(\tsub{g}{x}{f}\ s)})}{x}{(f\ (\lset{s}{y}{(\tsub{g}{x}{f}\ s)}))})\\
\iheq &\ftrans{\tsub{\phi}{x}{f}}\ (\lset{s}{y}{(\tsub{g}{x}{f}\ s)})\\
= &\dtrans{\tsub{\humod{y}{g}}{x}{f}}\ \ftrans{\tsub{\phi}{x}{f}}\ s
\end{align*}
In the second case, (R1) (R2), note (*) for every $t$ have $t = \tsub{t}{x}{f}$ on $\freevars{\phi}$ (since $x \notin \freevars{\phi}$ by R2) by \rref{lem:app-term-coincide}.

Then
$ \dtrans{\humod{y}{g}}\ \ftrans{\phi}\ \tsub{s}{x}{f}
= \dtrans{\humod{x}{g}}\ \ftrans{\phi}\ \tsub{s}{x}{f}
= \ftrans{\phi}\ (\lset{(\tsub{s}{x}{f})}{x}{(g\ (\tsub{s}{x}{f}))})
= \ftrans{\phi}\ (\lset{s}{x}{(g\ (\tsub{s}{x}{f}))})
=_* \ftrans{\phi}\   \tsub{(\lset{s}{x}{(g\ \tsub{s}{x}{f})})}{x}{f}
\iheq \ftrans{\tsub{\phi}{x}{f}}\ (\lset{s}{x}{(g\ \tsub{s}{x}{f})})
= \ftrans{(\tsub{\phi}{x}{f})}\ (\lset{s}{x}{(g\ (\lset{s}{x}{f\ s}))})
= \ftrans{(\tsub{\phi}{x}{f})}\ (\lset{s}{x}{((\tsub{g}{x}{f})\ s)})
= \dtrans{\humod{x}{(\tsub{g}{x}{f})}}\ \tsub{\phi}{x}{f}\ s
= \dtrans{\tsub{(\humod{x}{g})}{x}{f}}\ \tsub{\phi}{x}{f}\ s
= \dtrans{\tsub{(\humod{y}{g})}{x}{f}}\ \tsub{\phi}{x}{f}\ s
$

\mycase $\prandom{x}$:
From (A2) have either (L1) $x \neq y$ and (L2) $y \notin \freevars{f}$
or (R1) $x=y$ and (R2) $x \notin \freevars{\phi}$.
In both cases, then (S) $\tsub{\{\prandom{y}\}}{x}{f} = \{\prandom{y}\}$.
In case (L) then
$ \dtrans{\prandom{y}}\ \ftrans{\phi}\ \tsub{s}{x}{f}
=  (\pity{v}{\xty}{\ftrans{\phi}\ (\lset{(\tsub{s}{x}{f})}{y}{v})})
=  (\pity{v}{\xty}{\ftrans{\phi}\ (\lset{(\lset{s}{x}{(f\ s)})}{y}{v})})
=_{L1}  (\pity{v}{\xty}{\ftrans{\phi}\ (\lset{(\lset{s}{y}{v})}{x}{(f\ s)})})
=_{L2}  (\pity{v}{\xty}{\ftrans{\phi}\ (\lset{(\lset{s}{y}{v})}{x}{(f\ (\lset{s}{y}{v}))})})
=  (\pity{v}{\xty}{\ftrans{\tsub{\phi}{x}{f}}\ (\lset{s}{y}{v})})
= \dtrans{\prandom{y}}\ \ftrans{\tsub{\phi}{x}{f}}\ s
=_S \dtrans{\tsub{\prandom{y}}{x}{f}}\ \ftrans{\tsub{\phi}{x}{f}}\ s$

In case (R1) (R2) then
note (*) for every $t$ have $t = \tsub{t}{x}{f}$ on $\freevars{\phi}$ (since $x \notin \freevars{\phi}$ by R2) by \rref{lem:app-term-coincide}.
then
$ \dtrans{\prandom{y}}\ \ftrans{\phi}\ \tsub{s}{x}{f}
= \dtrans{\prandom{x}}\ \ftrans{\phi}\ \tsub{s}{x}{f}
= (\pity{v}{\xty}{\ftrans{\phi}\ (\lset{(\tsub{s}{x}{f})}{x}{v})})
= (\pity{v}{\xty}{\ftrans{\phi}\ (\lset{s}{x}{v})})
=_* (\pity{v}{\xty}{\ftrans{\phi}\ \tsub{(\lset{s}{x}{v})}{x}{f}})
\iheq (\pity{v}{\xty}{\ftrans{\tsub{\phi}{x}{f}}\ (\lset{s}{x}{v})})
= \dtrans{\prandom{x}}\ \ftrans{\tsub{\phi}{x}{f}}\ s
=_{R1} \dtrans{\prandom{y}}\ \ftrans{\tsub{\phi}{x}{f}}\ s
=_S \dtrans{\tsub{\prandom{y}}{x}{f}}\ \ftrans{\tsub{\phi}{x}{f}}\ s$.

\mycase $\ptest{\phi},$ have
$ \dtrans{\ptest{\psi}}\ \ftrans{\phi}\ \tsub{s}{x}{f}
= (\ftrans{\psi}\ \tsub{s}{x}{f} \to \ftrans{\phi} \tsub{s}{x}{f})
= (\ftrans{\tsub{\psi}{x}{f}}\ s \to \ftrans{\tsub{\phi}{x}{f}}\ s)
= \dtrans{\tsub{\ptest{\psi}}{x}{f}} \ftrans{\tsub{\phi}{x}{f}}\ s$

\mycase $\pevolvein{\D{z}=f}{\ivr}$:
From (A2) have (L1) $x \neq y$ and (L2) $y \notin \freevars{f}$ (there is no $x=y$ case).
In the sole case, (L1) (L2), then (S)
$\tsub{\{\humod{y}{g}\}}{x}{f} = \{\humod{x}{(\tsub{g}{x}{f})}\}$
also, from (L2) have
(*) $(f\ s) = (f\ (\lset{s}{y}{(\tsub{g}{x}{f}\ s)}))$ by \rref{lem:app-term-coincide} since $s = (\lset{s}{y}{(\tsub{g}{x}{f}\ s)})$ on $\{y\}^\complement \supseteq \freevars{f}^\complement$.
then

Have (0) $(\solves{sol}{\tsub{s}{x}{f}}{d}{\D{y}=g}) = (\solves{sol}{s}{d}{\tsub{(\D{y}=g)}{x}{f}})$ by the ``solves'' IH.

Have (1) for all $t \in [0,d],$
$\ftrans{\ivr}\ (\lset{(\tsub{s}{x}{f})}{y}{(sol\ t)})
= \ftrans{\tsub{\ivr}{x}{f}}\ (\lset{s}{y}{(sol\ t)})$
by IH on $\ivr$ and because
$ \lset{(\tsub{s}{x}{f})}{y}{(sol\ t)}
= \lset{(\lset{s}{x}{(f\ s)})}{y}{(sol\ t)}
= \lset{(\lset{s}{y}{(sol\ t)})}{x}{(f\ s)}
= \lset{(\lset{s}{y}{(sol\ t)})}{x}{(f\ (\lset{s}{y}{(sol\ t)}))}$
since by (L2) have $y \notin \freevars{f}$ thus $s = (\lset{s}{y}{(sol\ t)})$ on $\freevars{f}^\complement$ and \rref{lem:app-term-coincide} applies.

Have (2)
\begin{align*}
  &\ftrans{\phi}\ (\lset{(\tsub{s}{x}{f})}{(y,\D{y})}{(sol\ d, g\ (\lset{(\tsub{s}{x}{f})}{y}{(sol\ d)}))})\\
= &\ftrans{\tsub{\phi}{x}{f}}\ (\lset{s}{(y,\D{y})}{(sol\ d, \tsub{g}{x}{f}\ (\lset{s}{y}{(sol\ d)}))})
\end{align*}
by the IH on $\phi$ and because
\begin{align*}
  &\lset{(\tsub{s}{x}{f})}{y}{(sol\ d)}\\
= &\lset{(\lset{s}{y}{(sol\ d)})}{x}{(f\ (\lset{s}{y}{(sol\ d)}))}
\end{align*}
by the same argument as above and by term case because
$(\lset{(\tsub{s}{x}{f})}{y}{(sol\ d)})
= \tsub{(\lset{s}{y}{(sol\ d)})}{x}{f}$.

\begin{align*}
  &\dtrans{\pevolvein{\D{y}=g}{\ivr}}\ \ftrans{\phi} \tsub{s}{x}{f}\\
= &\pity{d}{{\reals_{\geq0}}}{\pity{sol}{[0,d]\to\xty}{}}\\
 &(\solves{sol}{\tsub{s}{x}{f}}{d}{\D{y}=g})\\
 &\to (\pity{t}{{[0,d]}}{\ftrans{\ivr}\ {(\lset{\tsub{s}{x}{f}}{y}{(sol\ t)})}})\\
 &\to \ftrans{\phi}\ (\lset{(\tsub{s}{x}{f})}{(y,\D{y})}{(sol\ d, g\ (\lset{(\tsub{s}{x}{f})}{y}{(sol\ d)}))})\\
=_{0,1,2}& \pity{d}{{\reals_{\geq0}}}{\pity{sol}{[0,d]\to\xty}{}}\\
 &(\solves{sol}{s}{d}{\tsub{(\D{y}=g)}}{x}{f})\\
 &\to (\pity{t}{{[0,d]}}{\ftrans{\tsub{\ivr}{x}{f}}\ {(\lset{s}{y}{(sol\ t)})}})\\
 &\to \ftrans{\tsub{\phi}{x}{f}}\ (\lset{s}{(y,\D{y})}{(sol\ d, \tsub{g}{x}{f}\ (\lset{s}{y}{(sol\ d)}))})\\
= &\dtrans{\tsub{(\pevolvein{\D{y}=g}{\ivr})}{x}{f}}\ \ftrans{\tsub{\phi}{x}{f}} s
\end{align*}

\mycase $\alpha;\beta,$ have
$ \dtrans{\alpha;\beta}\ \ftrans{\phi}\ \tsub{s}{x}{f}
= \dtrans{\alpha}\ (\dtrans{\beta}\ \ftrans{\phi})\ \tsub{s}{x}{f}
= \dtrans{\alpha}\ (\ftrans{\dbox{\beta}{\phi}})\ \tsub{s}{x}{f}
= \dtrans{\tsub{\alpha}{x}{f}}\ (\ftrans{\tsub{\dbox{\beta}{\phi}}{x}{f}})\ s
= \dtrans{\tsub{\alpha}{x}{f}}\ (\dtrans{\tsub{\beta}{x}{f}}\ \ftrans{\tsub{\phi}{x}{f}})\ s
= \dtrans{\tsub{\{\alpha;\beta\}}{x}{f}}\ \ftrans{\tsub{\phi}{x}{f}}\ s
$

\mycase $\alpha\cup\beta,$ have
$ \dtrans{\alpha\cup\beta}\ \ftrans{\phi}\ \tsub{s}{x}{f}
= \dtrans{\alpha}\ \ftrans{\phi}\ \tsub{s}{x}{f} \mathop{\kwprod} \dtrans{\beta}\ \ftrans{\phi}\ \tsub{s}{x}{f}
= \dtrans{\tsub{\alpha}{x}{f}}\ \ftrans{\tsub{\phi}{x}{f}}\ s \mathop{\kwprod} \dtrans{\tsub{\beta}{x}{f}}\ \ftrans{\tsub{\phi}{x}{f}}\ s
= \dtrans{\tsub{\{\alpha\cup\beta\}}{x}{f}}\ \ftrans{\tsub{\phi}{x}{f}}\ s
$

\mycase $\prepeat{\alpha}$:
The main step (*) holds by inner coinduction on the membership of $\tsub{s}{x}{f}$ in the fixed point.
The two fixed point solutions $\tau'$ and $\tau''$ are related by $\tau'\ \tsub{s}{x}{f}$ iff $\tau''\ s$ for all $s : \sty.$

In the base case by IH on $\psi$,
$ (\tau'\ \tsub{s}{x}{f} \to \psi\ \tsub{s}{x}{f})
= (\tau''\ s \to \psi\ \tsub{s}{x}{f})
= (\tau''\ s \to \tsub{\psi}{x}{f}\ s)$.

In the coinductive case, by IH on $\alpha,$ have
$ (\tau'\ \tsub{s}{x}{f} \to \dtrans{\alpha}\ \tau'\ \tsub{s}{x}{f})
= (\tau''\ s \to \dtrans{\alpha}\ \tau'\ \tsub{s}{x}{f})
= (\tau''\ s \to \dtrans{\tsub{\alpha}{x}{f}}\ \tau''\ s)$.

Then the main proof of the case proceeds:

$ \dtrans{\prepeat{\alpha}}\ \ftrans{\phi}\ \tsub{s}{x}{f}
= (\lcoty{\tau'\mathrel{:}(\sty \to \alltype)}{\lambda{t:\sty}\,{
(\tau'\ t \to \dtrans{\alpha}\ \tau'\ t)
\mathop{\kwprod}
(\tau'\ t \to \ftrans{\psi}\ t)
}})\ (\tsub{s}{x}{f})
=_{*}
 (\lcoty{\tau''\mathrel{:}(\sty \to \alltype)}{\lambda{t:\sty}\,{
(\tau''\ t \to \dtrans{\tsub{\alpha}{x}{f}}\ \tau''\ t)
\mathop{\kwprod}
(\tau''\ t \to \ftrans{\tsub{\psi}{x}{f}}\ t)
}})\ s
= \dtrans{\tsub{\prepeat{\alpha}}{x}{f}}\ \ftrans{\tsub{\phi}{x}{f}}\ s$.

\mycase $\pdual{\alpha},$ have
$ \dtrans{\pdual{\alpha}}\ \ftrans{\phi}\ \tsub{s}{x}{f}
= \atrans{\alpha}\ \ftrans{\phi}\ \tsub{s}{x}{f}
= \atrans{\tsub{\alpha}{x}{f}}\ \ftrans{\tsub{\phi}{x}{f}}\ s
= \dtrans{\pdual{(\tsub{\alpha}{x}{f})}}\ \ftrans{\tsub{\phi}{x}{f}}\ s
= \dtrans{\tsub{\{\pdual{\alpha}\}}{x}{f}}\ \ftrans{\tsub{\phi}{x}{f}}\ s
$

\end{proof}

\begin{theorem}[Soundness of Proof Calculus]
  If $\seq{\G}{\phi}$ is provable then $\seq{\G}{\phi}$ is valid.
  As a corollary when $\G = \Gemp,$ if $\phi$ is provable, then $\phi$ is valid.
\label{thm:proof-calculus-sound}
\end{theorem}
\begin{proof}
Each case proceeds by fixing $s : \sty$ and assuming (A) $\ftrans{\G}\ s$.
In cases where premisses include $\G,$ we assume that modus ponens has been applied to all premisses with (A).
Additional antecedents beyond $\G$ will be explicitly discharged in each caes.

\mycase \irref{bchoiceI}:
Assume premisses (0) $\proves{\G}{M}{\dbox{\alpha}{\phi}}$ and (1) $\proves{\G}{N}{\dbox{\beta}{\phi}}$
then by the IH's have
(0A) $\ftrans{\dbox{\alpha}{\phi}}\ s$ and (1A) $\ftrans{\dbox{\beta}{\phi}}\ s$,
which expand to
(0B) $\dtrans{\alpha}\ \ftrans{\phi}\ s$ and
(1B) $\dtrans{\beta}\ \ftrans{\phi}\ s,$ so that by conjunction
$\dtrans{\alpha}\ \ftrans{\phi}\ s \mathop{\kwprod} \dtrans{\beta}\ \ftrans{\phi}\ s,$ i.e.,
 $\dtrans{\alpha\cup\beta}\ \ftrans{\phi}\ s = \ftrans{\dbox{\alpha\cup\beta}{\phi}}\ s$.

\mycase \irref{bchoiceEL}:
From premiss, have (0) $\proves{\G}{M}{\dbox{\alpha\cup\beta}{\phi}}$, by IH have
(0A) $\ftrans{\dbox{\alpha\cup\beta}{\phi}}\ s$, which expands to 
$\dtrans{\alpha \cup \beta}\ \ftrans{\phi}\ s =
\dtrans{\alpha}\ \ftrans{\phi}\ s \mathop{\kwprod} \dtrans{\beta}\ \ftrans{\phi}\ s,$
whose left projection is
$\dtrans{\alpha}\ \ftrans{\phi}\ s = \ftrans{\dbox{\alpha}{\phi}}\ s$.

\mycase \irref{bchoiceER}:
From premiss, have (0) $\proves{\G}{M}{\dbox{\alpha\cup\beta}{\phi}}$, by IH have
(0A) $\ftrans{\dbox{\alpha\cup\beta}{\phi}}\ s$, which expands to 
$\dtrans{\alpha \cup \beta}\ \ftrans{\phi}\ s =
\dtrans{\alpha}\ \ftrans{\phi}\ s \mathop{\kwprod} \dtrans{\beta}\ \ftrans{\phi}\ s,$
whose right projection is
 $\dtrans{\beta}\ \ftrans{\phi}\ s = \ftrans{\dbox{\beta}{\phi}}\ s$.

\mycase \irref{dchoiceE}:
Assume
(0)  $\proves{\G}{A}{\ddiamond{\alpha\cup\beta}{\phi}}$ and
(1)  $\proves{\G,\ddiamond{\alpha}{\phi}}{B}{\psi}$ and
(2)   $\proves{\G,\ddiamond{\beta}{\phi}}{C}{\psi}$.

By the IH's, have

(0A)  $\ftrans{\ddiamond{\alpha\cup\beta}{\phi}}\ s$ and
(1A)  $\ftrans{\ddiamond{\alpha}{\phi} \limply \psi}\ s$ and
(2A)  $\ftrans{\ddiamond{\beta}{\phi} \limply\psi}\ s$,
which expand to 
(0B)  $\atrans{\alpha\cup\beta}\ \ftrans{\phi}\ s =
   \atrans{\alpha}\ \ftrans{\phi}\ s \mathop{\kwsum} \atrans{\beta}\ \ftrans{\phi}\ s$ and
(1B)  $\atrans{\alpha}\ \ftrans{\phi}\ s \limply \ftrans{\psi}\ s$ and
(2B)  $\atrans{\beta}\ \ftrans{\phi}\ s \limply \ftrans{\psi}\ s$,

From (0B) have two cases:
(L) $\atrans{\alpha}\ \ftrans{\phi}\ s$
or (R) $\atrans{\beta}\ \ftrans{\phi}\ s$.
In the first case, apply (1B) to (L), yielding $\ftrans{\psi}\ s,$
or in the second case, apply (2B) to (R), also yielding $\ftrans{\psi}\ s$ in each case as desired.

\mycase \irref{dchoiceIL}:
Assume (0) $\proves{\G}{M}{\ddiamond{\alpha}{\phi}}$.
By IH, (0A) $\ftrans{\ddiamond{\alpha}{\phi}}\ s$, which expands to
(0B) $\atrans{\alpha}\ \ftrans{\phi}\ s$
then by left injection,
$\atrans{\alpha}\ \ftrans{\phi}\ s \mathop{\kwsum} \atrans{\beta}\ \ftrans{\phi}\ s
= \atrans{\alpha\cup\beta}\ \ftrans{\phi}\ s
= \ftrans{\ddiamond{\alpha\cup\beta}{\phi}}\ s$.

\mycase \irref{dchoiceIR}:
Assume (0) $\proves{\G}{M}{\ddiamond{\beta}{\phi}}$.
By IH, (0A) $\ftrans{\ddiamond{\beta}{\phi}}\ s$, which expands to
(0B) $\atrans{\beta}\ \ftrans{\phi}\ s,$
then by right injection,
$ \atrans{\alpha}\ \ftrans{\phi}\ s \mathop{\kwsum} \atrans{\beta}\ \ftrans{\phi}\ s
= \atrans{\alpha\cup\beta}\ \ftrans{\phi}\ s
= \ftrans{\ddiamond{\alpha\cup\beta}{\phi}}\ s$.

\mycase \irref{dtestI}:
Assume (0) $\proves{\G}{M}{\phi}$ and (1) $\proves{\G}{N}{\psi}$.
By IH, (0A) $\ftrans{\phi}\ s$ and (1A) $\ftrans{\psi}\ s$, i.e.,
$\atrans{\ptest{\phi}}\ \ftrans{\psi}\ s
= \ftrans{\ddiamond{\ptest{\phi}}{\psi}}\ s$.

\mycase \irref{dtestEL}:
Assume (0) $\proves{\G}{M}{\ddiamond{\ptest{\phi}}{\psi}},$ so by IH,
(0A) $\ftrans{\ddiamond{\ptest{\phi}}{\psi}}\ s$, which expands to
(0B) $\atrans{\ptest{\phi}}\ \ftrans{\psi}\ s
= \ftrans{\phi}\ s \kwprod \ftrans{\psi} \ s$
whose left projection is $\ftrans{\phi}\ s$ as desired.

\mycase \irref{dtestER}:
Assume (0) $\proves{\G}{M}{\ddiamond{\ptest{\phi}}{\psi}},$ so by IH,
(0A) $\ftrans{\ddiamond{\ptest{\phi}}{\psi}}\ s$, which expands to
(0B) $\atrans{\ptest{\phi}}\ \ftrans{\psi}\ s
= \ftrans{\phi}\ s \kwprod \ftrans{\psi}\ s$
whose right projection is likewise $\ftrans{\psi}\ s$ as desired.

\mycase \irref{btestI}:
Assume (0) $\proves{\G,\phi}{M}{\psi}$.
By IH, (0A) $\ftrans{\phi}\ s \limply \ftrans{\psi}\ s$, i.e.,
$\dtrans{\ptest{\phi}}\ \ftrans{\psi}\ s$ which is 
$\ftrans{\dbox{\ptest{\phi}}{\psi}}\ s$.

\mycase \irref{btestE}:
Assume (0) $\proves{\G}{M}{\dbox{\ptest{\phi}}{\psi}}$ and (1) $\proves{\G}{N}{\phi}$.
By the IH's, have (0A)
$\ftrans{\dbox{\ptest{\phi}}{\psi}}
= \dtrans{\ptest{\phi}}\ \ftrans{\psi}\ s
= \ftrans{\phi}\ s \to \ftrans{\psi}\ s$
and (1A) $\ftrans{\phi}\ s$.
By modus ponens have $\ftrans{\psi}\ s$.

\mycase \irref{hyp}:
The side condition says conclusion $\phi$ satisifies $\phi \in \G$.
Then assumption (A) is $\ftrans{\G}\ s = \kwprod_{i \in \abs{\G}} \ftrans{\psi_i}\ s = \ftrans{\psi_1}\ s \kwprod \cdots \kwprod \ftrans{\phi}\ s \kwprod \cdots \kwprod \ftrans{\psi_{\abs{\G}}}\ s$ for some $\psi_i$ for $i \in [1,\abs{\G}]$.
By projection, $\ftrans{\phi}\ s$ as desired.

\mycase \irref{brandomI}:
WTS for all $s$ such that (A)  then $\ftrans{\dbox{\prandom{x}}{\phi}}\ s,$ which is
$(\pity{v}{\xty}{\ftrans{\phi}\ (\lset{s}{x}{v})})$, so assume any $v:\xty$ to prove $\ftrans{\phi}\ (\lset{s}{x}{v})$.
Since $x \notin \freevars{\G}$ by the side condition, then by \rref{lem:app-formula-coincide} and (A) have
$\ftrans{\G}\ (\lset{s}{x}{v}),$ thus by modus ponens on the premiss have $\ftrans{\phi}\ (\lset{s}{x}{v})$ for all $v$, which is to say
$\ftrans{\dbox{\prandom{x}}{\phi}}\ s$.

\mycase \irref{drandomI}:
From premisses have (0) $\proves{\G}{M}{\tsub{\phi}{x}{f}}$ for some $f$ such that $\tsub{\phi}{x}{f}$ is admissible, or by IH have
(0A) $\ftrans{\tsub{\phi}{x}{f}}\ s$.
We can apply \rref{lem:app-formula-subst} since $\tsub{\phi}{x}{f}$ is admissible, yielding
(1) $\ftrans{\phi}\ (\tsub{s}{x}{f})$.
Let $v = f\ s,$ so (1) simplifies to $\ftrans{\phi}\ (\lset{s}{x}{v}),$ thus $v$ is a witness to
$\sity{v}{\xty}{\ftrans{\phi}\ (\lset{s}{x}{v})}$ which is to say $\ftrans{\ddiamond{\prandom{x}}{\phi}}\ s$.

\mycase \irref{seqI}:
We give the diamond case, the box case is symmetric.
Assumption gives (0) $\proves{\G}{M}{\ddiamond{\alpha}{\ddiamond{\beta}{\phi}}}$ and by IH,
(0A)
$\ftrans{\ddiamond{\alpha}{\ddiamond{\beta}{\phi}}}\ s
= \atrans{\alpha}\ (\atrans{\beta}\ \ftrans{\phi})\ s
= \atrans{\alpha;\beta}\ \ftrans{\phi}\ s
= \ftrans{\ddiamond{\alpha;\beta}{\phi}}\ s$

\mycase \irref{asgnI}:
We give the diamond case, the box case is symmetric.
Assume (0) $\proves{\G}{M}{\tsub{\phi}{x}{f}}$ so by IH
(0A) $\ftrans{\tsub{\phi}{x}{f}}\ s$ and by \rref{lem:app-formula-subst} (since $\tsub{\phi}{x}{f}$ admissible)
then $\ftrans{\phi}\ \tsub{s}{x}{f} = \ftrans{\phi}\ (\lset{s}{x}{(f\ s)}),$ that is $\ftrans{\ddiamond{\humod{x}{f}}{\phi}}\ s$.

\mycase \irref{brandomE}:
Assumed (0) $\proves{\G}{M}{\ddiamond{\prandom{x}}{\phi}}$
an  (1) $\proves{\G}{N}{\lforall{x}{\phi \limply \psi}}$
then by the IH's have
    (0A) $\ftrans{\ddiamond{\prandom{x}}{\phi}}\ s = (\sity{v}{\xty}{\ftrans{\phi}\ (\lset{s}{x}{v})})$
and (1A) $\ftrans{\lforall{x}{\phi \limply \psi}}\ s = (\pity{v}{\xty}{\ftrans{\phi}\ (\lset{s}{x}{v}) \limply \ftrans{\psi}\ (\lset{s}{x}{v})})$

We unpack $v:\xty$ from (0A) so that $\ftrans{\phi}\ (\lset{s}{x}{v})$
In (1A) specialize to the same $v$ to get $\ftrans{\phi}\ (\lset{s}{x}{v}) \limply \ftrans{\psi}\ (\lset{s}{x}{v})$ and, by modus ponenst with (0A), have $\ftrans{\psi}\ (\lset{s}{x}{v})$.
By side condition that $x \notin \freevars{\psi}$, apply \rref{lem:app-formula-coincide} to get $\ftrans{\psi}\ s$.

\mycase \irref{drandomE}:
Assume side condition  $\tsub{\phi}{x}{f}$ admissible.
Assume
(0) $\proves{\G}{M}{\dbox{\prandom{x}}{\phi}}$ so by IH
(0A) $\ftrans{\dbox{\prandom{x}}{\phi}}\ s$ so
(1) $\ftrans{\phi}\ (\lset{s}{x}{v})$ for all $v:\xty$.

Seek to show $\ftrans{\tsub{\phi}{x}{f}}\ t$ in some arbitrary $t$.
By \rref{lem:app-formula-subst} since $\tsub{\phi}{x}{f}$ is admissible, it suffices to show $\ftrans{\phi} (\lset{t}{x}{(f\ t)})$.
Then apply (1) letting $s=t$ and $v = (f\ t)$ yielding
$\ftrans{\phi}\ (\lset{t}{x}{(f\ t)})$ as desired.

\mycase \irref{mon}:
Assume (0) $\proves{\G}{M}{\ddiamond{\alpha}{\phi}}$
and (1) $\proves{\earen{\G}{\alpha},\phi}{N}{\psi}$

Let $x$ be the vector of variables $x_i \in \boundvars{\alpha}$ in some canonical (e.g. lexicographic) order.
Let $y$ be a fresh vector of variables of the same length as $x$.
Let $z$ be $\freevars{\G} \setminus (\boundvars{\alpha})$.

Recall the discrete ghost rule defined in prior work~\cite{esop20}:
\[\cinferenceRule[ghost|iG]{}
{\linferenceRule[formula]
  {\proves{\G,\pvx:x=f}{M}{\phi}}
  {\proves{\G}{\eghost{x}{f}{\pvx}{M}}{\phi}}
}
{\m{x}\text{ fresh except free in }\m{M,}\ \m{\pvx}\text{ fresh}}
\]

Applying \irref{ghost} repeatedly to (1) have
$\proves{\G,y=x}{M}{\ddiamond{\alpha}{\phi}}$.
The IH applies because the context $\G,y=x$ is satisfied
by taking assumption (A) $\ftrans{\G}\ s$ and show
(AA) $\ftrans{\G}\ (\lset{s}{y}{x})$
which reflexivity satisfies $y=x$ while preserving $\G$ by \rref{lem:app-formula-coincide} since definitionally $y \cap \freevars{\G} = \emptyset$.
The the IH results in
$ \ftrans{\ddiamond{\alpha}{\phi}}\ (\lset{s}{y}{(s\ x)})
= \atrans{\alpha}\ \ftrans{\phi}\ (\lset{s}{y}{(s\ x)}),$
then apply \rref{lem:app-bound-effect} with $V=y,z$ giving
(B) $\atrans{\alpha}\ (\lambda t.~ \ftrans{\phi}\ t \mathop{\kwprod} t\ V = s\ V)\ (\lset{s}{y}{(s\ x)})$
To prove the assumption of (1) conclusion from (B), first prove a monotone lemma (Mon):
\begin{claim}[Mon]
Assume $\ftrans{\G}\ s$.
Then for all $t:\sty$

$\ftrans{\phi}\ t \land t\ V = s\ V$ implies $\ftrans{\phi \land \earen{\G}{\alpha}}\ t$
\end{claim}
\begin{proof}[Subproof]
The first conjunct holds trivially from the assumption.
For the latter, decompose $\G$ into two contexts $\G = \Delta_1,\Delta_2$ 
such that context $\Delta_1$ contains formulas $\psi$ such that $\freevars{\psi} \cap \boundvars{\alpha} = \emptyset,$ with all other formulas are in $\Delta_2$.

We have $\ftrans{\Delta_1}\ t$ from (A) by repeating \rref{lem:app-formula-coincide} because for each $\psi \in \Delta_1,$ we have
$s = t$ on $\freevars{\psi}$ as a consequent of $t\ V = s\ V$ specifically because $V \supseteq z \supseteq \freevars{\psi}$.

Consider $\Delta_2$ next.
Recall $\earen{\Delta_2}{\alpha} = \land_{\psi\in{\Delta_2}}(\eren{\psi}{x}{y})$.
For each such $\psi$ we appeal $s\ y = t\ y$ (since $y \subseteq V$) and $s\ y = s\ x$ by ghosting,
thus $t\ y = s\ x$ by transitivity.

To show the desired $\ftrans{\eren{\psi}{x}{y}}\ t$ we note by \rref{lem:app-formula-rename} it suffices
to show $\ftrans{\psi}\ \eren{t}{x}{y},$
which, because $x \in \freevars{\psi}$ and $y \notin \freevars{\psi}$, reduces to
$\ftrans{\psi} (\lset{t}{x}{(s\ y)})$.

Since (AA) includes $\ftrans{\psi}\ (\lset{s}{y}{x})$,
it suffices to apply \rref{lem:app-formula-coincide}, providing the assumption
that on $\freevars{\psi} \subseteq \{z\} \cup \boundvars{\psi} = \{x,z\}$ we have
$(\lset{s}{y}{(s\ x)})= (\lset{t}{x}{(s\ y)})$.
We already have  $s=t$ on $y,z$.
Then $s\ y$ is set to $s\ x$ but $s\ x = s\ y$ anyway, a no-op.
$t\ x$ is set to $s\ y = s\ x$ so that
$(\lset{s}{y}{(s\ x)})= (\lset{t}{x}{(s\ y)})$ agree on $\{x,y,z\}$ as desired.
\end{proof}

We observe the first assumption of (Mon) is just (A),
call the resulting universally quantifier statement (Mono).
By  \rref{lem:app-monotone} on (Mono) and (B) have

$\atrans {\alpha}\ (\ftrans{\phi \land \earen{\G}{\alpha}})\ s$
we finally reach
(Mid)  $\atrans{\alpha}\ \ftrans{\phi \land \earen{\G}{\alpha}}\ (\lset{s}{y}{(s\ x)})$.
By commuting $\phi$ and $\earen{\G}{\alpha}$ in (Mid) we now have the precondition to 
apply \rref{lem:app-monotone} on (1), getting $\atrans{\alpha}\ \ftrans{\psi}\ (\lset{s}{y}{(s\ x)})$.
Lastly, $y$ is fresh in $\alpha$ and $\phi,$ so \rref{lem:app-game-coincide}
gives the desired
$\atrans{\alpha}\ \ftrans{\psi}\ s$.

\mycase \irref{dualI}:
We give the case for diamonds, the case for boxes is symmetric.
Assumed (0) $\proves{\G}{M}{\dbox{\alpha}{\phi}}$, then by IH have
(0A) $\ftrans{\dbox{\alpha}{\phi}}\ s
= \dtrans{\alpha}\ \ftrans{\phi}\ s
= \atrans{\pdual{\alpha}}\ \ftrans{\phi}\ s
= \ftrans{\ddiamond{\pdual{\alpha}}{\phi}}\ s$ as desired.

\mycase \irref{drcase}:
Have assumptions
(0) $\proves{\G}{A}{\ddiamond{\prepeat{\alpha}}{\phi}}$
(1) $\proves{\G,\phi}{B}{\psi}$
(2) $\proves{\G,\ddiamond{\alpha}{\ddiamond{\prepeat{\alpha}}{\phi}}}{C}{\psi}$
so by IH have
(0A) $\ftrans{\ddiamond{\prepeat{\alpha}}{\phi}}\ s$
(1A) $\ftrans{\phi \limply \psi}\ s$
(2A) $\ftrans{\ddiamond{\alpha}{\ddiamond{\prepeat{\alpha}}{\phi}} \limply \psi}\ s$

From (0A) and the $\prepeat{\alpha}$ semantics, since $s$ belongs to the least fixed point
$(\lindty{\tau'\mathrel{:}(\sty \to \alltype)}{}
\lambda{t:\sty}\,
(\psi\ t \to \tau'\ t)
\mathop{\kwprod}
(\atrans{\alpha}\ \tau'\ t \to \tau'\ s)
)$
then we have

(3)$\psi\ s \mathop{\kwsum}
(\atrans{\alpha}\ \ftrans{\ddiamond{\prepeat{\alpha}}{\phi}}\ s \to \tau'\ s),$
so we proceed on cases on (3).
In case (L) $\psi\ s$ apply modus ponens on (1A) giving $\ftrans{\psi}\ s$ as desired.
In case (R) $\atrans{\alpha}\ \ftrans{\ddiamond{\prepeat{\alpha}}{\phi}}\ s \to \tau'\ s$
thus $\ftrans{\ddiamond{\alpha}{\ddiamond{\prepeat{\alpha}}{\phi}}}\ s$ and by modus ponens on (2A) have $\ftrans{\psi}\ s$ as desired.

\mycase \irref{bunroll}:
Assume (0) $\proves{\G}{M}{\dbox{\prepeat{\alpha}}{\phi}}$
and by IH (0A) $\ftrans{\dbox{\prepeat{\alpha}}{\phi}}\ s$.
Since $s$ belongs to the fixed point $(\lcoty{\tau'\mathrel{:}(\sty \to \alltype)}{\lambda{t:\sty}\,{
(\tau'\ t \to \dtrans{\alpha}\ \tau'\ t)
\kwprod
(\tau'\ t \to \ftrans{\phi}\ t)
}})$
then
$\dtrans{\alpha}\ \ftrans{\dbox{\prepeat{\alpha}}{\phi}}\ s
\kwprod
\ftrans{\phi}\ s$
which commutes to
$\ftrans{\phi}\ s
\kwprod
\dtrans{\alpha}\ \ftrans{\dbox{\prepeat{\alpha}}{\phi}}\ s$
which simplifies to
$\ftrans{\phi \land \dbox{\alpha}{\dbox{\prepeat{\alpha}}{\phi}}}\ s$.

\mycase \irref{dstop}:
Assume (0) $\proves{\G}{M}{\phi}$
and by IH (0A) $\ftrans{\phi}\ s,$ thus $s$ belongs to the fixed point $(\lindty{\tau'\mathrel{:}(\sty \to \alltype)}{\lambda{t:\sty}\,{
(\psi\ t \to \tau'\ t)
\mathop{\kwprod}
(\atrans{\alpha}\ \tau'\ t \to \tau'\ t)
}})$ and $\ftrans{\ddiamond{\prepeat{\alpha}}{\phi}}\ s$ as desired.

\mycase \irref{dgo}:
Assume (0) $\proves{\G}{M}{\ddiamond{\alpha}{\ddiamond{\prepeat{\alpha}}{\phi}}}\ s$
and by IH (0A) $\ftrans{\ddiamond{\alpha}{\ddiamond{\prepeat{\alpha}}{\phi}}}\ s,$ 
thus $\atrans{\alpha}\ \ftrans{\ddiamond{\prepeat{\alpha}}{\phi}}\ s,$ 
thus $s$ belongs to the fixed point $(\lindty{\tau'\mathrel{:}(\sty \to \alltype)}{}\lambda{t:\sty}\,
(\psi\ t \to \tau'\ t)
\mathop{\kwprod}
(\atrans{\alpha}\ \tau'\ t \to \tau'\ s)
)$ and $\ftrans{\ddiamond{\prepeat{\alpha}}{\phi}}\ s$ as desired.

\mycase \irref{broll}:
Assume (0) $\proves{\G}{M}{\phi \land \dbox{\alpha}{\dbox{\prepeat{\alpha}}{\phi}}}$ and by IH
(0A) $\ftrans{\phi \land \dbox{\alpha}{\dbox{\prepeat{\alpha}}{\phi}}}\ s$
so $\ftrans{\phi}\ s$ and $\dtrans{\alpha}\ \ftrans{\dbox{\prepeat{\alpha}}{\phi}}\ s$
so $s$ belongs to the fixed point
$(\lcoty{\tau'\mathrel{:}(\sty \to \alltype)}{}\lambda{t:\sty}\,
(\tau'\ t \to \dtrans{\alpha}\ \tau'\ t)
\kwprod
(\tau'\ t \to \ftrans{\phi}\ t)
)$

thus $\ftrans{\dbox{\prepeat{\alpha}}{\phi}}\ s$.

\mycase \irref{bloopI}:
Assume (0) $\proves{\G}{A}{J}$
(1) $\proves{J}{B}{\dbox{\alpha}{J}}$.
(2) $\proves{J}{C}{\phi}$.
By IH's, have (0A) $\ftrans{J}\ s$ and for all $t \in \sty,$ (1A) $\ftrans{J}\ t \to \ftrans{\dbox{\alpha}{J}}\ t$
and (2A) $\ftrans{J}\ s \to \ftrans{\phi}\ s$.
From (1A) and (2A), $J$ is a fixed point $\tau'$ of the equivalence
$\forall t,\ ((\tau'\ t) =
         ((\tau'\ t \to \dtrans{\alpha}\ \tau'\ t)
\kwprod  (\tau'\ t \to \ftrans{\phi}\ t)))$,
and from (0A) we have that $s$ belongs to the fixed point $J$.
Then $s$ belongs also to the greatest fixed point, that is $s$ belongs to
$(\lcoty{\tau'\mathrel{:}(\sty \to \alltype)}{}\lambda{t:\sty}\,
(\tau'\ t \to \dtrans{\alpha}\ \tau'\ t)
\kwprod
(\tau'\ t \to \ftrans{\phi}\ t)
),$ which is to say
$\ftrans{\dbox{\prepeat{\alpha}}{\phi}}\ s$.

\mycase \irref{dloopE}:
Assume
(0) $\proves{\G}{A}{\ddiamond{\prepeat{\alpha}}{\phi}}$
(1) $\proves{\phi}{B}{\psi}$
(2) $\proves{\ddiamond{\alpha}{\psi}}{C}{\psi}$
so by IH have
(0A) $\ftrans{\ddiamond{\prepeat{\alpha}}{\phi}}\ s$
(1A) $\ftrans{\phi\ \to \psi}\ t$ for all $t : \sty$
(2A) $\ftrans{\ddiamond{\alpha}{\psi}\to \psi}\ t$ for all $t : \sty$.

(0A) simplifies to
(3) $(\lindty{\tau'\mathrel{:}(\sty \to \alltype)}{\lambda{t:\sty}\,{
(\phi\ t \to \tau'\ t)
\mathop{\kwprod}
(\atrans{\alpha}\ \tau'\ t \to \tau'\ t)
}})\ s,$ from which we proceed by induction on the fixed point membership of $s$.

In the first case, (L) $\phi\ s$ so by modus ponens with (1) have $\ftrans{\psi}\ s$.
In the second case have $(\atrans{\alpha}\ \tau'\ s)$ where $\tau'\ s = \ftrans{\psi}\ s$, so
$\atrans{\alpha}\ \ftrans{\psi}\ s$ and $\ftrans{\ddiamond{\alpha}{\psi}}\ s$ and by modus ponens on (2A) have $\ftrans{\psi}\ s$.
So in each case the conclusion $\ftrans{\psi}\ s$ holds.

\mycase \irref{dloopI}:
Assume (0) $\proves{\G}{A}{\conv}$
(1) $\proves{\conv,(\metz \metle \met \land \met_0 = \met)}{B}{\ddiamond{\alpha}{(\conv \land \met_0 \metgr \met)}}$
(2) $\proves{(\conv \land \metz \metgeq \met)}{C}{\phi}$


By the IH have
(0A) $\ftrans{\conv}\ s$
and have (1A)
$\ftrans{\conv \land (\metz \metle \met\ s \land \met_0 = \met)}\ t$  $\to$ $\ftrans{\ddiamond{\alpha}{(\conv\land \met_0 \metgr \met)}}\ t$ for all $t : \sty$.
and have (2A) $\ftrans{\conv \land \metz \metgeq \met \to \phi}\ t$ for all $t : \sty$.

By the side condition, the order $\metgr$  is Noetherian and satisfies the ascending chain condition.
Our main proof shows that (Main) $s$ belongs to the fixed point $\ftrans{\ddiamond{\prepeat{\alpha}}{(\conv \land \metz \metgeq \met)}}\ s$.
We proceed by induction on the value of $\met\ s$; by the ascending chain condition, this induction is well-founded.
Specifically, we show that $\met\ s$ decreases with every iteration $\alpha$ until $\met = \metz$ where the loop terminates.
Invariant $\conv$ is maintained in every iteration.

For the base case, (0A) gave $\ftrans{\conv}\ s$ and by conjunction $\ftrans{(\conv \land \metz \metgeq \met)}\ s$ which fulfills the base case of the fixed point.

In the inductive case have $\metz \metle \met\ s$, and still have assumption $\ftrans{\conv}\ s$.
Introduce a fresh discrete ghost (rule \irref{ghost}) $\met_0 = \met\ s$.
We from (1A) have the following fact $\ftrans{\conv \land (\met \metle \met\ s \land \met_0 = \met)}\ (\lset{s}{\met_0}{(\met\ s)})$ by \rref{lem:app-formula-coincide} since $\met_0$ was fresh.
By modus ponens, (1A) yields (4) $\ftrans{\ddiamond{\alpha}{(\conv \land \met_0 \metgr \met)}}\ s$.
Since every point $t$ satisfying $\ftrans{(\conv \land \met_0 \metgr \met)}\ t$ trivially has a lesser value of $\met$ than $s,$
we can apply the IH on region $\ftrans{(\conv \land \met_0 \metgr \met)},$ yielding for all $t$ that $\ftrans{(\conv \land \met_0 \metgr \met)}\ t \to \ftrans{\ddiamond{\prepeat{\alpha}}{(\conv \land \metz \metgeq \met)}}\ t$.
Then by \rref{lem:app-monotone} have
(5) $\ftrans{\ddiamond{\alpha}{\ddiamond{\prepeat{\alpha}}{(\conv\land \metz \metgr \met)}}}\ (\lset{s}{\met_0}{(\met\ s)})$
which satisfies the inductive case, giving
$\ftrans{\ddiamond{\prepeat{\alpha}}{(\conv \land \metz \metgr \met)}} (\lset{s}{\met_0}{(\met\ s)})$.
By freshness of $\met_0,$ \rref{lem:app-formula-coincide} gives $\ftrans{\ddiamond{\prepeat{\alpha}}{(\conv \land \metz \metgr \met)}}\ s$.
This completes the inner induction.
By \rref{lem:app-monotone} and (2A) have $\ftrans{\ddiamond{\prepeat{\alpha}}{\phi}}\ s$ as desired.

\mycase \irref{di}:
Assume (0) $\proves{\G}{}{\phi}$
(1) $\proves{\G}{}{\lforall{x}{(\ivr \limply \dbox{\humod{\D{x}}{f}}{\der{\phi}})}}$
so by IH have
(I0) $\ftrans{\phi}\ s$
and (2) for all  $v : \xty, \ftrans{\ivr}\ (\lset{s}{x}{v}) \limply \ftrans{\der{\phi}}\ (\lset{s}{(x,\D{x})}{(v,f\ (\lset{s}{x}{v}))})$.
To show the conclusion, first assume $d$ and $sol$ such that
(I1) $(\solves{sol}{s}{d}{\D{x}=f})$ and
(4) $\ftrans{\ivr}\ {(\lset{s}{x}{(sol\ t)})}$ for $t \in [0,d]$
and then we conclude that (P) $\ftrans{\phi}\ (\lset{s}{(x,\D{x})}{(sol\ r, f\ (\lset{s}{x}{(sol\ r)}))})$.
From (2) and (4) we then have
$ \ftrans{\der{\phi}}\ (\lset{(\lset{s}{x}{(sol\ t)})}{\D{x}}{(f\ (\lset{s}{x}{(sol\ t)}))})
= \ftrans{\der{\phi}}\ (\lset{s}{(x,\D{x})}{(sol\ t, \der{sol}\ t)})$ (I2)
by unpacking the definition of $(\solves{sol}{s}{d}{\D{x}=f})$ in (I1).

We now prove the conclusion (P) by induction on differentiable formula $\phi$, maintaining assumptions (I0), (I1), (I2) throughout the induction.
Each case reduces to a lemma proven elsewhere in the literature.
However, those prior works appeal to the time-derivative of a term $f,$ while our term $\der{f}$ is a spatial derivative.
Our \rref{lem:app-differential-lemma} shows that these derivatives are equal in the postcondition of an ODE, thus the lemmas are applicable.

\mycase $g = h$:
Let $G(t) = g\ (\lset{s}{(x,\D{x})}{(sol\ t, \der{sol}\ t)})$ and $H(t) = h\ (\lset{s}{(x,\D{x})}{(sol\ t, \der{sol}\ t)})$.
In this case (I0) specializes to
 $g\ s = h\ s$ so (Init) $G(0) = H(0)$ since $sol\ 0 = \lget{s}{x}$ by (I1) and \rref{lem:app-formula-coincide} since $\D{x} \notin \freevars{g} \cup \freevars{h}$.
Then, (I2) specializes for all $t \in [0,d]$ to
$ \der{g}\ (\lset{s}{(x,\D{x})}{(sol\ t, \der{sol}\ t)})
= \der{h}\ (\lset{s}{(x,\D{x})}{(sol\ t, \der{sol}\ t)}),$ that
is (EqT) $\D{G}(t) = \D{H}(t)$ for $t \in [0,d]$.
Apply the $=$ case of \rref{thm:app-mvt} to (Init) and (EqT) on interval $[0,d],$ thus $G(t) = H(t)$ for $t \in [0,d]$.
Therefore (Post) $\ftrans{g = h}\ (\lset{s}{(x,\D{x})}{(sol\ t, \der{sol}\ t)})$.
Since (Post) holds under the assumptions (I1) and (4) then
$\ftrans{\dbox{\pevolvein{\D{x}=f}{\ivr}}{g = h}}\ s$ as desired.

\mycase $g > h$:
Let $F(t) = g\ (\lset{s}{(x,\D{x})}{(sol\ t, \der{sol}\ t)}) - h\ (\lset{s}{(x,\D{x})}{(sol\ t, \der{sol}\ t)})$.
In this case (I0) specializes to
 $g\ s > h\ s$ so (Init) $F(0) > 0$ since $sol\ 0 = \lget{s}{x}$ by (I1) and \rref{lem:app-formula-coincide} since $\D{x} \notin \freevars{g} \cup \freevars{h}$.
Then, (I2) specializes for all $t \in [0,d]$ to
$ \der{g}\ (\lset{s}{(x,\D{x})}{(sol\ t, \der{sol}\ t)})
> \der{h}\ (\lset{s}{(x,\D{x})}{(sol\ t, \der{sol}\ t)}),$ so that
(GrT) $\D{F}(t) > 0 $ for $t \in [0,d]$.
Apply the $>$ case of \rref{thm:app-mvt}  to (Init) and (GrT) on interval $[0,d],$ thus $F(d) > F(0)$ so by transitivity with (Init), have $F(d) > 0$ so
(Post) $\ftrans{g > h}\ (\lset{s}{(x,\D{x})}{(sol\ t, \der{sol}\ t)})$.
Since (Post) holds under the assumptions (I1) and (4) then
$\ftrans{\dbox{\pevolvein{\D{x}=f}{\ivr}}{g > h}}\ s$ as desired.

\mycase $g \geq h$:
Let $F(t) = g\ (\lset{s}{(x,\D{x})}{(sol\ t, \der{sol}\ t)}) - h\ (\lset{s}{(x,\D{x})}{(sol\ t, \der{sol}\ t)})$.
In this case (I0) specializes to
 $g\ s \geq h\ s$ so (Init) $F(0) \geq 0$ since $sol\ 0 = \lget{s}{x}$ by (I1) and \rref{lem:app-formula-coincide} since $\D{x} \notin \freevars{g} \cup \freevars{h}$.
Then, (I2) specializes for all $t \in [0,d]$ to
$ \der{g}\ (\lset{s}{(x,\D{x})}{(sol\ t, \der{sol}\ t)})
\geq \der{h}\ (\lset{s}{(x,\D{x})}{(sol\ t, \der{sol}\ t)}),$ so that
(GeqT) $\D{F}(t) \geq 0 $ for $t \in [0,d]$.
Apply the $\geq$ case of \rref{thm:app-mvt}  to (Init) and (GeqT) on interval $[0,d],$ thus $F(d) \geq F(0)$ so by transitivity with (Init), have $F(d) \geq 0$ so
(Post) $\ftrans{g \geq h}\ (\lset{s}{(x,\D{x})}{(sol\ t, \der{sol}\ t)})$.
Since (Post) holds under the assumptions (I1) and (4) then
$\ftrans{\dbox{\pevolvein{\D{x}=f}{\ivr}}{g \geq h}}\ s$ as desired.

\mycase $\phi_1 \land \phi_2$:
In this case (I0) specializes to
$\ftrans{\phi_1 \land \phi_2}\ s$ so (L0) $\ftrans{\phi_1}\ s$ and (R0) $\ftrans{\phi_2}\ s$.
Then since $\der{\phi_1 \land \phi_2} = \der{\phi_1} \land \der{\phi_2}$ we also have
(L2) $\ftrans{\der{\phi_1}}\ (\lset{s}{(x,\D{x})}{(sol\ t, \der{sol}\ t)})$ and
(R2) $\ftrans{\der{\phi_2}}\ (\lset{s}{(x,\D{x})}{(sol\ t, \der{sol}\ t)})$.
By the IH on (L0) (I1) (L2) and  have (IHL) $\ftrans{\dbox{\pevolvein{\D{x}=f}{\ivr}}{\phi_1}}\ s$.
By the IH on (R0) (I1) (R2) and then we have (IHR) $\ftrans{\dbox{\pevolvein{\D{x}=f}{\ivr}}{\phi_2}}\ s$.
By applying assumptions (I1) and (4) to (IHL) and (IHR) have
(L3) $\ftrans{\phi_1}\ (\lset{s}{(x,\D{x})}{(sol\ t, \der{sol}\ t)})$ and
(R3) $\ftrans{\phi_2}\ (\lset{s}{(x,\D{x})}{(sol\ t, \der{sol}\ t)})$ so
(LR) $\ftrans{\phi_1 \land \phi_2}\ (\lset{s}{(x,\D{x})}{(sol\ t, \der{sol}\ t)})$ thus
$\ftrans{\dbox{\pevolvein{\D{x}=f}{\ivr}}{(\phi_1 \land \phi_2)}}\ s$ as desired.

\mycase $\phi_1 \lor \phi_2$:
Since by definition $\der{\phi_1 \lor \phi_2} = \der{\phi_1} \land \der{\phi_2}$ then we have
(L2) $\ftrans{\der{\phi_1}}\ (\lset{s}{(x,\D{x})}{(sol\ t, \der{sol}\ t)})$ and
(R2) $\ftrans{\der{\phi_2}}\ (\lset{s}{(x,\D{x})}{(sol\ t, \der{sol}\ t)})$.
By the IH on (I1) (L2) have (IHL) $\ftrans{\phi_1}\ s \limply \ftrans{\dbox{\pevolvein{\D{x}=f}{\ivr}}{\phi_1}}\ s$.
By the IH on (I1) (R2) have (IHR) $\ftrans{\phi_2}\ s \limply \ftrans{\dbox{\pevolvein{\D{x}=f}{\ivr}}{\phi_2}}\ s$.
In this case (I0) specializes to
$\ftrans{\phi_1 \lor \phi_2}\ s$, proceed by cases.
First case:  (L0) $\ftrans{\phi_1}\ s$, by (IHL) have
$\ftrans{\dbox{\pevolvein{\D{x}=f}{\ivr}}{\phi_1}}\ s$ thus
$\ftrans{\dbox{\pevolvein{\D{x}=f}{\ivr}}{(\phi_1 \lor \phi_2)}}\ s$ and
$\ftrans{\dbox{\pevolvein{\D{x}=f}{\ivr}}{(\phi_1 \lor\phi_2)}}\ s$.
Second case: (R0) $\ftrans{\phi_2}\ s,$ by (IHR) have
$\ftrans{\dbox{\pevolvein{\D{x}=f}{\ivr}}{\phi_2}}\ s$ thus
$\ftrans{\dbox{\pevolvein{\D{x}=f}{\ivr}}{(\phi_1 \lor \phi_2)}}\ s$ and
$\ftrans{\dbox{\pevolvein{\D{x}=f}{\ivr}}{(\phi_1 \lor\phi_2)}}\ s$.

\mycase \irref{dc}:
Assume (0) $\proves{\G}{}{\dbox{\pevolvein{\D{x}=f}{\ivr}}{\rho}}$
and (1) $\proves{\G}{}{\dbox{\pevolvein{\D{x}=f}{\ivr \land \rho}}{\phi}}$
so by IH have
(0A) $\ftrans{\dbox{\pevolvein{\D{x}=f}{\ivr}}{\rho}}\ s$
and (1A) $\ftrans{\dbox{\pevolvein{\D{x}=f}{\ivr \land \rho}}{\phi}}\ s$.
To show the conclusion, first assume a duration $d$ and solution $sol$ such that
(S) $(\solves{sol}{s}{d}{\D{x}=f})$
and (C) $(\pity{t}{[0,d]}{\ftrans{\ivr}\ {(\lset{s}{x}{(sol\ t)})}})$
after which we show that (P) $\ftrans{\phi}\ (\lset{s}{(x,\D{x})}{(sol\ d, f\ (\lset{s}{x}{(sol\ t)}))})$.
Because ``solves'' and universal quantification are prefix-closed, we also have:
(S0) $(\solves{sol}{s}{r}{\D{x}=f})$ for all $r \in [0,d]$
(C0) $(\pity{t}{[0,r]}{\ftrans{\ivr}\ {(\lset{s}{x}{(sol\ t)})}})$ for all $r \in [0,d]$.
Thus apply (S0) and (C0) to (0) for each $r \in [0,d],$ getting

(2) $\ftrans{\rho}\ (\lset{s}{(x,\D{x})}{(sol\ r, f\ (\lset{s}{x}{(sol\ r)}))})$ for all $r \in [0,d]$.
Since $\rho$ appears in a domain constraint, we must have $\D{x} \notin \freevars{\rho}$ by syntactic restriction, so by \rref{lem:app-formula-coincide} have
(2A) $\ftrans{\rho}\ (\lset{s}{x}{(sol\ r)})$ for all $r \in [0,d]$.
Conjoining (2A) with (C) we have
(CA) $(\pity{t}{[0,d]}{\ftrans{\ivr \land rho}\ {(\lset{s}{x}{(sol\ t)})}})$.
Applying (S) and (CA) to (1) we have
(P) $\ftrans{\phi}\ (\lset{s}{(x,\D{x})}{(sol\ r, f\ (\lset{s}{x}{(sol\ r)}))})$ as desired.

\mycase \irref{dw}:
By assumption we have (0) $\proves{\G}{}{\lforall{(x,\D{x})}{(\ivr \limply \phi)}}$
so by IH have (0A) $\ftrans{\ivr \to \phi}\ (\lset{s}{(x,\D{x})}{(v,\D{v})})$ for all $v, \D{v} \in \xty$.
To show $\dbox{\pevolvein{\D{x}=f}{\ivr}}{\phi}$ assume
some $d > 0$ and $sol : [0,d] \to \xty$ such that
(1) $(\solves{sol}{s}{d}{\D{x}=f})$ and
(2) $\ftrans{\ivr}\ {(\lset{s}{x}{(sol\ t)})}$ for $t \in [0,d]$
and then show $\ftrans{\phi}\ (\lset{s}{x}{(sol\ d)})$
For each $t$ we can apply \rref{lem:app-formula-coincide} in (2) because $\D{x} \notin \freevars{\ivr}$ by syntactic restriction, giving
(3) $\ftrans{\ivr}\ (\lset{s}{(x,\D{x})}{(sol\ t, f\ (\lset{s}{x}{(sol\ t)}))})\ s$.
Apply (3) to (0A) have $\ftrans{\phi}\ (\lset{s}{(x,\D{x})}{(sol\ t, f\ (\lset{s}{x}{(sol\ t)}))})$ as desired.

\mycase \irref{dg}:
Assume the side conditions that $y$ is fresh and that $a,b$ are continuous.
Assume (0) $\proves{\G}{}{\lexists{y}{\dbox{\pevolvein{\D{x}=f,\D{y}=a(x)y + b(x)}{\ivr}}{\phi}}}$
then by IH (0A) $\ftrans{\lexists{y}{\dbox{\pevolvein{\D{x}=f,\D{y}=a(x)y + b(x)}{\ivr}}{\phi}}}\ s$.
Unpack $v : \xty$ such that we will have
(0B) $\ftrans{\dbox{\pevolvein{\D{x}=f,\D{y}=a(x)y + b(x)}{\ivr}}{\phi}}\ (\lset{s}{y}{v})$.
To show $\ftrans{\dbox{\pevolvein{\D{x}=f}{\ivr}}{\phi}}\ s$ first assume $d$ and $sol$ are such that
    (1) $(\solves{sol}{s}{d}{\D{x}=f})$
and (2) $(\pity{t}{[0,d]}{\ftrans{\ivr}\ {(\lset{s}{x}{(sol\ t)})}})$ to show
(P) $\ftrans{\phi}\ (\lset{s}{(x,\D{x})}{(sol\ d, f\ (\lset{s}{x}{(sol\ d)}))})$

We will do so by applying (0B), which first requires constructing a solution
$xysol$ to $\D{x}=f,\D{y}=a(x)y + b(x)$ of the same duration as $sol$.
We construct $xysol$ according to form $xysol(t) = (sol(t), ysol(t))$ for some $ysol(t)$ which solves $\D{y}(t)=a(sol\ t)y + b(sol\ t)$.
So we first construct $ysol$ according to the initial value problem:
\begin{align*}
  ysol(0)             &=   v\\
  \D{ysol}(t,v_y) &= a(sol(t)) v_y + b(sol(t))
\end{align*}
We show that $\der{xysol}(t) = (\der{sol}(t),\der{ysol}(t))$ is Lipschitz continuous in $t,v_y$.
It is continuous  because addition, multiplication, and composition preserve continuity, because $a$ and $b$ are continuous by the side condition, and because $sol$ is continuous since it is the solution of an ODE.
It satisfies the Lipschitz condition with constant
\[L = \max_{t\in[0,d]}(sol\ t \cdot a(sol(t)))\]
where the maximum exists because it is the maximum of a continuous function on a compact interval.
Thus, \rref{thm:app-pl} applies and there exists a solution $ysol$ on $[0,d]$ such that
(solY) $(\solves{ysol}{(\lset{s}{y}{v})}{d}{\{\D{y} = a(sol\ t)y + b(sol\ t)\}})$.
From (solY) and (1) have
(1XY) $(\solves{xysol}{(\lset{s}{y}{v})}{d}{\{\D{x} = f, \D{y} = a(sol\ t)y + b(sol\ t)\}})$.
From (2) have
(2XY) $(\pity{t}{[0,d]}{\ftrans{\ivr}\ {(\lset{s}{(x,y)}{(xysol\ t)})}})$
by \rref{lem:app-formula-coincide} using the assumption that $y$ was fresh in $\phi$.
Now we abbreviate $ss = (\lset{(\lset{s}{y}{v})}{(x,y)}{(xysol\ d)}) = \lset{s}{(x,y)}{(xysol\ d)}$.
Apply $xysol, d,$ (1XY), and (2XY) to (0A) to get
$\ftrans{\phi}\ (\lset{ss}{(\D{x},\D{y}))}{(f\ ss,(a(x)y + b(x))\ ss)}$
from which we have (P) because $\sprojL{xysol} = sol$ and by \rref{lem:app-formula-coincide} since $y$ was fresh in $\phi$.

\mycase \irref{dv}:
Assume side conditions
(V1) $t$ fresh and (V2) $x,\D{x},t,\D{t}$ not free in $d,\veps$.
Assume (0) $\proves{\G}{}{\ddiamond{\humod{t}{0};\{\pevolvein{\D{t}=1,\D{x}=f}{\ivr}\}}{t \geq d}}$
Assume (1) $\proves{\G}{}{\dbox{\pevolve{\D{x}=f}}{(\der{h}-\der{g}) \geq \veps}}$
Assume (2) $\proves{\G}{}{d > 0 \land \veps > 0 \land h-g \geq -d\veps}$
Assume (3) $\proves{\ivr, h \geq g}{}{\phi}$, then by IH have
\begin{align*}
\text{(0A)} &\ftrans{\ddiamond{\humod{t}{0};\{\pevolvein{\D{t}=1,\D{x}=f}{\ivr}\}}{t \geq d}}\ s\\
\text{(1A)} &\ftrans{\dbox{\pevolve{\D{x}=f}}{(\der{h}-\der{g}) \geq \veps}}\ s\\
\text{(2A)} &\ftrans{d > 0 \land \veps > 0 \land h-g \geq -d\veps}\ s\\
\text{(3A)} &\ftrans{\ivr \land h \geq g \limply \phi}\ t\text{ for all }t : \sty
\end{align*}
To show (P) $\ftrans{\ddiamond{\pevolvein{\D{x}=f}{\ivr}}{\phi}}\ s$ it suffices to choose $sol$ and $d$ such that
\begin{align*}
\text{(P0)} &(\solves{sol}{s}{d}{\D{x}=f})\\
\text{(P1)} &(\pity{t}{{[0,d]}}{\ftrans{\phi}\ {(\lset{s}{x}{(sol\ t)})}})\\
\text{(P2)} &\ftrans{\psi}\ (\lset{s}{(x,\D{x})}{(sol\ d, f\ (\lset{s}{x}{(sol\ d)}))})
\end{align*}
From (0A) unpack $tsol$ and $dur$ such that
 (B) $(\solves{tsol}{(\lset{s}{t}{0}}{dur}{\{\D{t}=1,\D{x}=f}\})$ and
 (C) $(\pity{tt}{{[0,dur]}}{\ftrans{\ivr}\ {(\lset{s}{(t,x)}{(sol\ tt)})}})$ and also have
 (D) $\ftrans{t \geq d}\ (\lset{s}{(t,x,\D{t},\D{x})}{((tsol\ dur), (1,f\ (\lset{s}{x}{(sol\ dur)})))})$.

Since $(\lambda x.~x)$ is the unique solution of $\D{t} = 1$ we have  $tsol(r) = (r,sol(r))$ for some solution $sol$.
This also yields (Dur)

$dur \geq d\ (\lset{s}{(t,x,\D{t},\D{x})}{((tsol\ dur), (1,f\ (\lset{s}{x}{(sol\ dur)})))}) = d\ s$
by \rref{lem:app-formula-coincide} and (V2).

We project the solution for $f,$ noting \rref{lem:app-formula-coincide} applies since $t \notin \freevars{f} \cup \freevars{\ivr} \cup \freevars{\phi}$ by (V1).
We get
 (B1) $(\solves{sol}{s}{dur}{\D{x}=f})$
and
 (C1) $\pity{t}{{[0,dur]}}{\ftrans{\phi}\ {(\lset{s}{(x)}{(sol\ t)})}}$
and
 (D1) $\ftrans{\phi}\ (\lset{s}{(x,\D{x})}{(sol\ dur, f\ (\lset{s}{x}{(sol\ dur)}))})$.

We can now instantiate (1A) with solution $sol$ and duration $dur$ using (B1) and (CT1) to get
(Der) $\ftrans{\der{h}-\der{g} \geq \veps}\ (\lset{s}{(x,\D{x})}{(sol\ dur, f\ (\lset{s}{x}{(sol\ dur)}))})$, which also holds for $t \in [0,dur]$ since solutions and quantifiers are prefix-closed.
From (2A) have (Dpos) $d\ s > 0$ with $d$ constant by (V2)
and (Epos) $\veps\ s > 0$ and $\veps$ constant by (V2)
and (HG) $h\ s - g\ s \geq -d\veps$ with $d,\veps$ constant.
Let $HG(t) = (h - g)\ (\lset{s}{(x,\D{x})}{(sol\ t, f\ (\lset{s}{x}{(sol\ t)}))})$ so that
$\D{HG}(t) = \ftrans{\der{h}-\der{g}}\ (\lset{s}{(x,\D{x})}{(sol\ dur, f\ (\lset{s}{x}{(sol\ dur)}))})$.
Let $c = \veps\ = (\lset{s}{(x,\D{x})}{(sol\ dur, f\ (\lset{s}{x}{(sol\ dur)}))}) = (\lset{s}{(x,\D{x})}{(sol\ t, f\ (\lset{s}{x}{(sol\ t)}))})$ (for all $t \in [0,d]$) $ = \veps\ s$ by (V2) and \rref{lem:app-term-coincide}.
We write $\veps$ for short since it's constant anyway.

Now \rref{lem:app-dv} applies to $HG$ and $c$ on $[0,dur]$ by (Der) so that
(Progress) $HG(dur) - HG(0) \geq \veps\,dur > \veps\,d$, i.e.,
$HG(dur) \geq HG(0) + \veps\,dur > \veps\,d$.

By (HG) have $HG(0) \geq -d\,\veps,$ add with (Progress) to get
(Big) $HG(dur) \geq 0$.
From (Big) and instantiating (C1) at (Dur), we have the premisses of (3A), so that
(Post) $\ftrans{\phi}\ (\lset{s}{(x,\D{x})}{(sol\ dur, f\ (\lset{s}{x}{(sol\ dur)}))})$.
We conclude $\ftrans{\ddiamond{\pevolvein{\D{x}=f}{\ivr}}{\phi}}\ s$ by choosing solution $sol,$ duration $dur,$ and noting
(B1), (C1), and (Post).

\mycase \irref{bsolve}:
Assume the side condition (V) that $\D{x} \notin \freevars{\phi}$.
Assume the side condition (Sol) that $sln$ is the unique, global solution of $\D{x}=f$ such that $sln\ 0 = \lget{s}{x}$.
Not every syntactically valid \CdGL ODE has unique solutions, nor global solutions, let alone global solutions which lend themselves to tractable first-order proof obligations.
In practical implementations, the latter requirement demands nilpotence so that solutions are first-order-expressible, which trivially implies unique global solutions because nilpotent ODEs are both linear and Lipschitz.

Assume (0) $\proves{\G}{}{\lforall[{\reals_{\geq0}}]{d}{((\lforall[{[0,d]}]{t}{\dbox{\humod{x}{(sln\ t)}}{\ivr}})\limply\dbox{\humod{x}{(sln\ d)}}{\phi})}}$
so by IH (0A)
$\ftrans{((\lforall[{[0,d]}]{t}{\dbox{\humod{x}{(sln\ t)}}{\ivr}})
 \limply
\dbox{\humod{x}{(sln\ d)}}{\phi})
}\ (\lset{s}{d}{v})$ for all $v \geq 0$.
To show the conclusion $\dbox{\pevolvein{\D{x}=f}{\ivr}}{\phi}$ first assume some $d$ and $sol$
such that (1) $(\solves{sol}{s}{d}{\D{x}=f})$ and
(2) $\pity{t}{[0,d]}{\ftrans{\ivr}\ {(\lset{s}{x}{(sol\ t)})}}$
which allows us to show
(P) $\ftrans{\phi}\ (\lset{s}{(x,\D{x})}{(sol\ d, f\ (\lset{s}{x}{(sol\ d)}))})$.
By (Sol), solutions are unique, so $sln=sol$.
Then (2) is equivalent to $\ftrans{\lforall[{[0,d]}]{t}{\dbox{\humod{x}{(sln\ t)}}{\ivr}}}\ s$ so by (0A)
have (3) $\ftrans{\phi}\ (\lset{s}{d}{v})$ which even holds for all $v \in [0,d]$ since universal quantification is prefix-closed, but especially $d$.
By \rref{lem:app-formula-coincide} and (V) on (3) have
$\ftrans{\phi}\ (\lset{s}{(x,\D{x})}{(sln\ d, f\ (\lset{s}{x}{(sln\ d)}))})$
which by (Sol) again is equivalent to (P).

\mycase \irref{dsolve}:
Assume the side condition (V) that $\D{x} \notin \freevars{\phi}$.
Assume the side condition (Sol) that $sln$ is a global solution of $\D{x}=f$ such that $sln\ 0 = \lget{s}{x}$.
In practice it usually also the unique solution, but that is not strictly required in this rule.

Assume (0) $\proves{\G}{}{\lexists[{\reals_{\geq0}}]{d}{((\lforall[{[0,d]}]{t}{\dbox{\humod{x}{(sln\ t)}}{\ivr}})\land\dbox{\humod{x}{(sln\ d)}}{\phi})}}$
so by IH (0A)
$\ftrans{((\lforall[{[0,d]}]{t}{\dbox{\humod{x}{(sln\ t)}}{\ivr}})
 \land
\dbox{\humod{x}{(sln\ d)}}{\phi})
}\ (\lset{s}{d}{v})$ for some $v \geq 0$.
To show the conclusion $\dbox{\pevolvein{\D{x}=f}{\ivr}}{\phi}$
must exhibit $d$ and $sol$
such that
(1) $(\solves{sol}{s}{d}{\D{x}=f})$ and
(2) $\pity{t}{[0,d]}{\ftrans{\ivr}\ {(\lset{s}{x}{(sol\ t)})}}$ and
(3) $\ftrans{\phi}\ (\lset{s}{(x,\D{x})}{(sol\ d, f\ (\lset{s}{x}{(sol\ d)}))})$.
Choose $sol = y$ and $d = v$.
Then (1) is simply (Sol), (2) is left component of (0A) and (3) follows from the right component of (0A) by \rref{lem:app-formula-coincide} from side condition (V).
\end{proof}

\begin{lemma}[Existence Property]
Let $s \in \sty$. If $M : (\ftrans{\sity{x}{\xty}{\phi}}\ s)$ then there exist terms $v:\xty$ and $N : (\ftrans{\tsub{\phi}{x}{v}}\ s)$.
\label{lem:app-term-ep}
\end{lemma}
\begin{proof}
$(\ftrans{\sity{x}{\tau}{\phi}}\ s)
= \sity{v}{\xty}{(\ftrans{\phi}\ (\lset{s}{x}{v}))}$
so by inversion on $M$ there exists
$v:\xty$ and $N$ such that $N:(\ftrans{\phi}\ (\lset{s}{x}{v}))$
and such that the substitution $\tsub{\phi}{x}{v}$ is admissible.
Then by \rref{lem:app-formula-coincide} have $\esub{N}{x}{v} : (\ftrans{\tsub{\phi}{x}{v}}\ s)$ as desired.
\end{proof}

\begin{lemma}[Disjunction Property]
  If $M : (\ftrans{\phi \lor \psi}\ s)$ then there exists constructively an $N$ such that either $N : (\ftrans{\phi}\ s)$ or $N : (\ftrans{\psi}\ s)$.
\end{lemma}
\begin{proof}
$(\ftrans{\phi \lor \psi}\ s)
= \ftrans{\phi}\ s \mathop{\kwsum} \ftrans{\psi}\ s$
so by inversion on $M$ there exists
$L : \ftrans{\phi}\ s$ or $R : \ftrans{\psi}\ s$ as desired.
\end{proof}

%% file: hybrid.bbl
\begin{thebibliography}{10}
\providecommand{\url}[1]{\texttt{#1}}
\providecommand{\urlprefix}{URL }
\providecommand{\doi}[1]{\href{https://doi.org/#1}{\nolinkurl{#1}}}

\bibitem{DBLP:journals/iandc/AbramskyJM00}
Abramsky, S., Jagadeesan, R., Malacaria, P.: Full abstraction for {PCF}. Inf.
  Comput.  \textbf{163}(2),  409--470 (2000). \doi{10.1006/inco.2000.2930}

\bibitem{DBLP:journals/jacm/AlurHK02}
Alur, R., Henzinger, T.A., Kupferman, O.: Alternating-time temporal logic. J.
  {ACM}  \textbf{49}(5),  672--713 (2002). \doi{10.1145/585265.585270}

\bibitem{DBLP:series/lncs/Benthem15}
van Benthem, J.: Logic of strategies: What and how? In: van Benthem, J., Ghosh,
  S., Verbrugge, R. (eds.) Models of Strategic Reasoning - Logics, Games, and
  Communities, LNCS, vol.~8972, pp. 321--332. Springer (2015).
  \doi{10.1007/978-3-662-48540-8\_10}

\bibitem{DBLP:journals/sLogica/BenthemP11}
van Benthem, J., Pacuit, E.: Dynamic logics of evidence-based beliefs. Studia
  Logica  \textbf{99}(1-3),  61--92 (2011). \doi{10.1007/s11225-011-9347-x}

\bibitem{DBLP:journals/games/BenthemPR11}
van Benthem, J., Pacuit, E., Roy, O.: Toward a theory of play: {A} logical
  perspective on games and interaction. Games  (2011). \doi{10.3390/g2010052}

\bibitem{bishop1967foundations}
Bishop, E.: Foundations of constructive analysis  (1967)

\bibitem{DBLP:journals/corr/abs-1908-05535}
Bohrer, B., Platzer, A.: Toward structured proofs for dynamic logics. CoRR
  \textbf{abs/1908.05535} (2019), \url{http://arxiv.org/abs/1908.05535}

\bibitem{esop20}
Bohrer, B., Platzer, A.: Constructive game logic. In: ESOP (2020), in press

\bibitem{DBLP:conf/pldi/BohrerTMMP18}
Bohrer, B., Tan, Y.K., Mitsch, S., Myreen, M.O., Platzer, A.: {VeriPhy}:
  Verified controller executables from verified cyber-physical system models.
  In: Grossman, D. (ed.) PLDI. pp. 617--630. {ACM} (2018).
  \doi{10.1145/3192366.3192406}

\bibitem{bridges2007techniques}
Bridges, D.S., Vita, L.S.: Techniques of constructive analysis. Springer (2007)

\bibitem{DBLP:journals/fuin/Celani01}
Celani, S.A.: A fragment of intuitionistic dynamic logic. Fundam. Inform.
  \textbf{46}(3),  187--197 (2001),
  \url{http://content.iospress.com/articles/fundamenta-informaticae/fi46-3-01}

\bibitem{DBLP:conf/concur/ChatterjeeHP07}
Chatterjee, K., Henzinger, T.A., Piterman, N.: Strategy logic. In: Caires, L.,
  Vasconcelos, V.T. (eds.) CONCUR. LNCS, Springer (2007).
  \doi{10.1007/978-3-540-74407-8\_5}

\bibitem{DBLP:journals/iandc/CoquandH88}
Coquand, T., Huet, G.P.: The calculus of constructions. Inf. Comput.
  \textbf{76}(2/3),  95--120 (1988). \doi{10.1016/0890-5401(88)90005-3}

\bibitem{DBLP:conf/colog/CoquandP88}
Coquand, T., Paulin, C.: Inductively defined types. In: Martin{-}L{\"{o}}f, P.,
  Mints, G. (eds.) COLOG. LNCS, vol.~417. Springer (1988).
  \doi{10.1007/3-540-52335-9\_47}

\bibitem{DBLP:conf/mkm/Cruz-FilipeGW04}
Cruz{-}Filipe, L., Geuvers, H., Wiedijk, F.: C-{CoRN}, the constructive {C}oq
  repository at {Nijmegen}. In: Asperti, A., Bancerek, G., Trybulec, A. (eds.)
  MKM. LNCS, vol.~3119. Springer (2004). \doi{10.1007/978-3-540-27818-4\_7},
  \url{https://github.com/coq-community/corn}, accessed: git commits 9c44dae
  and 6411967

\bibitem{curry1967combinatory}
Curry, H., Feys, R.: Combinatory logic. In: Heyting, A., Robinson, A. (eds.)
  Studies in logic and the foundations of mathematics. North-Holland (1958)

\bibitem{degen2006towards}
Degen, J., Werner, J.: Towards intuitionistic dynamic logic. Logic and Logical
  Philosophy  \textbf{15}(4),  305--324 (2006). \doi{10.12775/LLP.2006.018}

\bibitem{DBLP:journals/fac/Dybjer94}
Dybjer, P.: Inductive families. Formal Asp. Comput.  \textbf{6}(4),  440--465
  (1994). \doi{10.1007/BF01211308}

\bibitem{DBLP:conf/IEEEcca/FilippidisDLOM16}
Filippidis, I., Dathathri, S., Livingston, S.C., Ozay, N., Murray, R.M.:
  Control design for hybrid systems with {TuLiP}: The temporal logic planning
  toolbox. In: {IEEE} {Conference} on {Control} {Applications}. {IEEE} (2016),
  \url{https://doi.org/10.1109/CCA.2016.7587949}

\bibitem{DBLP:conf/iros/FinucaneJK10}
Finucane, C., Jing, G., Kress{-}Gazit, H.: {LTLMoP}: Experimenting with
  language, temporal logic and robot control. In: {IROS}. {IEEE} (2010),
  \url{https://doi.org/10.1109/IROS.2010.5650371}

\bibitem{foster2010bidirectional}
Foster, J.N.: Bidirectional programming languages. Ph.D. thesis, University of
  Pennsylvania (2010), \url{https://repository.upenn.edu/edissertations/56}

\bibitem{ghosh2008strategies}
Ghosh, S.: Strategies made explicit in dynamic game logic. Workshop on Logic
  and Intelligent Interaction at ESSLLI pp. 74 --81 (2008)

\bibitem{DBLP:harel2000}
Harel, D., Kozen, D., Tiuryn, J.: Dynamic logic. {MIT} Press (2000).
  \doi{10.1145/568438.568456}

\bibitem{10.1007/3-540-48320-9_23}
Henzinger, T.A., Horowitz, B., Majumdar, R.: Rectangular hybrid games. In:
  Baeten, J.C.M., Mauw, S. (eds.) CONCUR'99 Concurrency Theory. pp. 320--335.
  Springer, Berlin, Heidelberg (1999)

\bibitem{DBLP:journals/cacm/Hoare69}
Hoare, C.A.R.: An axiomatic basis for computer programming. Commun. {ACM}
  \textbf{12}(10),  576--580 (1969). \doi{10.1145/363235.363259}

\bibitem{DBLP:conf/atal/HoekJW05}
van~der Hoek, W., Jamroga, W., Wooldridge, M.J.: A logic for strategic
  reasoning. In: Dignum, F., Dignum, V., Koenig, S., Kraus, S., Singh, M.P.,
  Wooldridge, M.J. (eds.) AAMAS. {ACM} (2005). \doi{10.1145/1082473.1082497}

\bibitem{DBLP:journals/aml/HofmannOS06}
Hofmann, M., van Oosten, J., Streicher, T.: Well-foundedness in realizability.
  Arch. Math. Log.  \textbf{45}(7),  795--805 (2006).
  \doi{10.1007/s00153-006-0003-5}

\bibitem{howard1980formulae}
Howard, W.A.: The formulae-as-types notion of construction. To HB Curry: essays
  on combinatory logic, lambda calculus and formalism  \textbf{44},  479--490
  (1980)

\bibitem{IsaacsDifferentialGames}
Isaacs, R.: Differential games: a mathematical theory with applications to
  warfare and pursuit, control and optimization. Series in Applied Mathematics
  (SIAM), Wiley, New York (1965)

\bibitem{kamide2010strong}
Kamide, N.: Strong normalization of program-indexed lambda calculus. Bull.
  Sect. Logic Univ. \L \'{o}d\'{z}  \textbf{39}(1-2),  65--78 (2010)

\bibitem{DBLP:journals/tac/KloetzerB08}
Kloetzer, M., Belta, C.: A fully automated framework for control of linear
  systems from temporal logic specifications. {IEEE} Trans. Automat. Contr.
  \textbf{53}(1),  287--297 (2008),
  \url{https://doi.org/10.1109/TAC.2007.914952}

\bibitem{lipton1992constructive}
Lipton, J.: Constructive {K}ripke semantics and realizability. In: Moschovakis,
  Y. (ed.) Logic from Computer Science. pp. 319--357. Springer (1992).
  \doi{10.1007/978-1-4612-2822-6\_13}

\bibitem{constructiveRealAlgebra}
Lombardi, H.: Th\'{e}ories g\'{e}om\'{e}triques pour l'alg\`{e}bre des nombres
  r\'{e}els sans test de signe ni axiome de choix d\'{e}pendant  (2019),
  \url{http://hlombardi.free.fr/Reels-geometriques.pdf}, accessed: 2019-07-10.
  Unpublished draft (in French)

\bibitem{DBLP:conf/itp/MakarovS13}
Makarov, E., Spitters, B.: The {Picard} algorithm for ordinary differential
  equations in {Coq}. In: Blazy, S., Paulin{-}Mohring, C., Pichardie, D. (eds.)
  ITP. LNCS, vol.~7998. Springer (2013). \doi{10.1007/978-3-642-39634-2\_34}

\bibitem{DBLP:journals/corr/Mamouras16}
Mamouras, K.: Synthesis of strategies using the {Hoare} logic of angelic and
  demonic nondeterminism. Logical Methods in Computer Science  \textbf{12}(3)
  (2016). \doi{10.2168/LMCS-12(3:6)2016}

\bibitem{DBLP:journals/fmsd/MitschP16}
Mitsch, S., Platzer, A.: {ModelPlex}: Verified runtime validation of verified
  cyber-physical system models. Form. Methods Syst. Des.  \textbf{49}(1),
  33--74 (2016). \doi{10.1007/s10703-016-0241-z}, special issue of selected
  papers from RV'14

\bibitem{DBLP:conf/lics/VIICHP04}
{Murphy VII}, T., Crary, K., Harper, R., Pfenning, F.: A symmetric modal lambda
  calculus for distributed computing. In: LICS. {IEEE} (2004).
  \doi{10.1109/LICS.2004.1319623}

\bibitem{DBLP:journals/mscs/Oosten02}
van Oosten, J.: Realizability: {A} historical essay. Mathematical Structures in
  Computer Science  \textbf{12}(3),  239--263 (2002).
  \doi{10.1017/S0960129502003626}

\bibitem{DBLP:conf/focs/Parikh83}
Parikh, R.: Propositional game logic. In: FOCS. pp. 195--200. {IEEE} (1983).
  \doi{10.1109/SFCS.1983.47}

\bibitem{DBLP:journals/jar/Platzer08}
Platzer, A.: Differential dynamic logic for hybrid systems. J. Autom. Reas.
  \textbf{41}(2),  143--189 (2008). \doi{10.1007/s10817-008-9103-8}

\bibitem{DBLP:journals/logcom/Platzer10}
Platzer, A.: Differential-algebraic dynamic logic for differential-algebraic
  programs. J. Log. Comput.  \textbf{20}(1),  309--352 (2010).
  \doi{10.1093/logcom/exn070}

\bibitem{DBLP:journals/tocl/Platzer15}
Platzer, A.: Differential game logic. {ACM} Trans. Comput. Log.
  \textbf{17}(1),  1:1--1:51 (2015). \doi{10.1145/2817824}

\bibitem{DBLP:conf/cade/Platzer15}
Platzer, A.: A uniform substitution calculus for differential dynamic logic.
  In: Felty, A.P., Middeldorp, A. (eds.) CADE. LNCS, vol.~9195, pp. 467--481.
  Springer (2015). \doi{10.1007/978-3-319-21401-6\_32}

\bibitem{Platzer18}
Platzer, A.: Logical Foundations of Cyber-Physical Systems. Springer,
  Switzerland (2018). \doi{10.1007/978-3-319-63588-0}

\bibitem{DBLP:conf/cade/Platzer18}
Platzer, A.: Uniform substitution for differential game logic. In: Galmiche,
  D., Schulz, S., Sebastiani, R. (eds.) IJCAR. LNCS, vol. 10900, pp. 211--227.
  Springer (2018). \doi{10.1007/978-3-319-94205-6\_15}

\bibitem{DBLP:conf/lics/PlatzerT18}
Platzer, A., Tan, Y.K.: Differential equation axiomatization: The impressive
  power of differential ghosts. In: Dawar, A., Gr{\"{a}}del, E. (eds.) LICS.
  pp. 819--828. ACM, New York (2018). \doi{10.1145/3209108.3209147}

\bibitem{DBLP:conf/focs/Pratt76}
Pratt, V.R.: Semantical considerations on floyd-hoare logic. In: FOCS. pp.
  109--121. {IEEE} (1976). \doi{10.1109/SFCS.1976.27}

\bibitem{DBLP:conf/kr/RamanujamS08}
Ramanujam, R., Simon, S.E.: Dynamic logic on games with structured strategies.
  In: Brewka, G., Lang, J. (eds.) Knowledge Representation. pp. 49--58. {AAAI}
  Press (2008), \url{http://www.aaai.org/Library/KR/2008/kr08-006.php}

\bibitem{shakernia2000decidable}
Shakernia, O., Pappas, G.J., Sastry, S.: Decidable controller synthesis for
  classes of linear systems. In: Lynch, N.A., Krogh, B.H. (eds.) HSSC. LNCS,
  vol.~1790, pp. 407--420. Springer (2000). \doi{10.1007/3-540-46430-1\_34}

\bibitem{shakernia2001semi}
Shakernia, O., Pappas, G.J., Sastry, S.: Semi-decidable synthesis for
  triangular hybrid systems. In: Benedetto, M.D.D., Sangiovanni{-}Vincentelli,
  A.L. (eds.) HSCC. LNCS, vol.~2034, pp. 487--500. Springer (2001).
  \doi{10.1007/3-540-45351-2\_39}

\bibitem{DBLP:conf/emsoft/TalyT10}
Taly, A., Tiwari, A.: Switching logic synthesis for reachability. In: Carloni,
  L.P., Tripakis, S. (eds.) {EMSOFT}. pp. 19--28. {ACM} (2010),
  \url{https://doi.org/10.1145/1879021.1879025}

\bibitem{DBLP:conf/fm/TanP19}
Tan, Y.K., Platzer, A.: An axiomatic approach to liveness for differential
  equations. In: ter Beek, M., McIver, A., Oliviera, J.N. (eds.) FM. LNCS, vol.
  11800, pp. 371--388. Springer (2019). \doi{10.1007/978-3-030-30942-8\_23}

\bibitem{COQ}
{The Coq development team}: The {Coq} proof assistant reference manual (2019),
  \url{https://coq.inria.fr/}

\bibitem{tomlin2000game}
Tomlin, C.J., Lygeros, J., Sastry, S.S.: A game theoretic approach to
  controller design for hybrid systems. Proceedings of the IEEE
  \textbf{88}(7),  949--970 (2000)

\bibitem{van2001games}
Van~Benthem, J.: Games in dynamic-epistemic logic. Bulletin of Economic
  Research  \textbf{53}(4),  219--248 (2001)

\bibitem{DBLP:series/txtcs/Weihrauch00}
Weihrauch, K.: Computable Analysis - An Introduction. Texts in Theoretical
  Computer Science. An {EATCS} Series, Springer (2000).
  \doi{10.1007/978-3-642-56999-9}

\bibitem{DBLP:journals/apal/Wijesekera90}
Wijesekera, D.: Constructive modal logics {I}. Ann. Pure Appl. Logic
  \textbf{50}(3),  271--301 (1990). \doi{10.1016/0168-0072(90)90059-B}

\bibitem{DBLP:journals/apal/WijesekeraN05}
Wijesekera, D., Nerode, A.: Tableaux for constructive concurrent dynamic logic.
  Ann. Pure Appl. Logic  (2005). \doi{10.1016/j.apal.2004.12.001}

\end{thebibliography}
